%% file: m.tex
\documentclass[runningheads]{llncs}

\input{General-preamble}

\input{kl}
\begin{document}
\title{Parameterized Broadcast Networks with Registers: from NP to the Frontiers of Decidability\thanks{\footnotesize{Partly supported by ANR project PaVeDyS (ANR-23-CE48-0005).}}}

\titlerunning{Parameterized broadcast networks with registers} 



\author{Lucie Guillou\inst{1} \and Corto Mascle\inst{2} \and Nicolas Waldburger\inst{3}}

\institute{IRIF, CNRS, Université Paris Cité \and 
LaBRI, Université de Bordeaux \and 
IRISA, Universit\'e de Rennes}
\authorrunning{L. Guillou, C. Mascle, N. Waldburger}



	\maketitle

\newcommand{\cortoin}[1]{\todo[color=blue!20,inline]{\small #1}}
\newcommand{\corto}[1]{\todo[color=blue!20]{\small #1}}

\newcommand{\nicoin}[1]{\todo[color=red!20,inline]{\small #1}}
\newcommand{\nico}[1]{\todo[color=red!20]{\small #1}}

\newcommand{\luin}[1]{\todo[color=teal!20,inline]{\small #1}}
\newcommand{\lu}[1]{\todo[color=teal!20]{\small #1}}

\newif\ifproofs
\proofstrue

\newif\ifintuition
\intuitionfalse

\newif\ifbasic
\basicfalse

	\begin{abstract}
	We consider the parameterized verification of networks of agents which communicate through unreliable broadcasts. In this model, agents have local registers whose values are unordered and initially distinct and may therefore be thought of as identifiers.
	When an agent broadcasts a message, it appends to the message the value stored in one of its registers. 
	Upon reception, an agent can store the received value or test it for equality against one of its own registers. 
	We consider the coverability problem, where one asks whether a given state of the system may be reached by at least one agent. We establish that this problem is decidable, although non-primitive recursive. We contrast this with the undecidability of the closely related target problem where all agents must synchronize on a given state. 
	On the other hand, we show that the coverability problem is \NP-complete when each agent only has one register.
	
	 
	\end{abstract}

	\input{Intro}

		\paragraph*{}
	\AP In this document, each "notion" is linked to its ""definition"" using the \href{https://www.irif.fr/~colcombe/knowledge_en.html}{knowledge package}. On electronic devices, clicking on words or symbols allows to access their definitions.

	\input{General-defs}

	\input{General-abstraction}

	\input{General-target}
	\section{Cover in 1-BNRA}
	\label{sec:cover-1BNRA}
	\input{defs}

\input{abstraction}

	\section{Conclusion}
	We established the decidability (and $\Fcomplexity{\omega^\omega}$-completeness) of the coverability problem for BNRA, as well as the \NP-completeness of the problem for 1-BNRA.
	Concerning future work, one may want to push decidability further, for instance by enriching our protocols
	with inequality tests, as done in classical models such as data nets  \cite{datanetsinequalityfomegaomegaomega}. 
	Reductions of other distributed models to this one are also being studied.
	
	\paragraph*{Acknowledgements.} We are grateful to Arnaud Sangnier for encouraging us to work on BNRA, for the discussions about his work in \cite{DelzannoST13} and for his valuable advice. We also thank Philippe Schnoebelen for the interesting discussion and Sylvain Schmitz for the exchange on complexity class $\Fcomplexity{\omega^{\omega}}$ and related topics.   
	\bibliography{biblio}
	
	\newpage
	
	\appendix
	\input{Appendix/Preuve-LCS-reduction} 
	\input{Appendix/definition-unfolding-trees}
	\input{Appendix/Preuve-trees-sound-complete}

	\input{Appendix/Preuves-reductions-branches}

	\input{Appendix/Preuve-Tower-lemma}

	\input{Appendix/Preuves-tree-bounds}

	\input{Appendix/decidability}
	\input{Appendix/Preuve-Target}
	\input{Appendix/one-reg-proofs-v3}
\end{document}

%% file: General-preamble.tex
\usepackage[utf8]{inputenc}
\usepackage[dvipsnames]{xcolor}
\usepackage{hyperref}

\usepackage[demo]{graphicx}
\usepackage{caption}
\usepackage{subcaption}
\usepackage[]{todonotes}
\usepackage[notion, quotation, electronic]{knowledge}
\usepackage{extarrows}
\usepackage[normalem]{ulem}
\usepackage{mathtools}
\usepackage{amssymb}
\usepackage{amsmath}
\usepackage{xspace}
\usepackage{multicol}
\usepackage{tikz}
\usepackage{booktabs}
\usepackage{enumitem}
\usepackage[capitalise]{cleveref}
\usepackage{cite}
\usepackage{multidef}
\usepackage{xspace}
\usepackage{mathrsfs}
\usepackage{thm-restate}

\usetikzlibrary{arrows,calc,automata,fit,shapes,positioning}
\tikzset{AUT style/.style={>=angle 60,initial text= ,every edge/.append,every state/.style={minimum size=20,inner sep=2}}}

\let\llncsproof\proof
\renewcommand{\proof}[1][]{%
  \ifx!#1!\else\renewcommand{\proofname}{#1}\fi
  \llncsproof
}

\declaretheorem[name=Observation]{observation}
\makeatletter
\let\c@proposition\c@theorem
\let\c@corollary\c@theorem
\let\c@lemma\c@theorem
\let\c@definition\c@theorem
\let\c@example\c@theorem
\let\c@remark\c@theorem
\let\c@fact\c@theorem
\let\c@observation\c@theorem
\makeatother

\let\saveendexample\endexample
\def\endexample{\qed\saveendexample}
\let\saveendproof\endproof
\def\endproof{\qed\saveendproof}

\newcommand{\nats}{\mathbb{N}}
\newcommand{\set}[1]{\{#1\}}
\newcommand{\powerset}[1]{2^{#1}}
\newcommand{\step}[1]{\xrightarrow{#1}}
\newcommand{\extlabel}[2]{\mathsf{ext}(#1,#2)} 
\newcommand{\intlabel}[1]{\mathsf{int}(#1)}
\newcommand{\locallabel}{\lambda} 
\newcommand{\extbr}[2]{\xrightarrow{\extlabel{#1}{#2}}}
\newcommand{\intstep}[1]{\xrightarrow{\intlabel{#1}}}
\newcommand{\size}[1]{|#1|}
\newcommand{\length}[1]{\mathsf{len}(#1)}
\newcommand{\nset}[2]{[#1,#2]} 

\multidef{\text{\sc{#1}}\xspace}{co,coNP,NP,PSPACE,coPSPACE,NPSPACE,PTIME,EXPTIME,EXPSPACE,NEXPTIME, LOGSPACE, TOWER, BNRA}
\newcommand{\Fcomplexity}[1]{\mathbf{F}_{#1}} 
\newcommand{\Fomegaomega}{\Fcomplexity{\omega^{\omega}}}
\newcommand{\Ffunction}[1]{\mathscr{F}_{#1}}
\knowledgenewrobustcmd{\COVER}{\withkl{\kl[coverability problem]}{\cmdkl{\textsc{Cover}}}}

\newcommand{\prot}{\mathcal{P}} 
\newcommand{\transitions}{\Delta} 
\newcommand{\atrans}{\delta}

\newcommand{\config}{\gamma}
\renewcommand{\epsilon}{\varepsilon}

\newcommand{\brone}[1]{\textbf{br}(#1)}
\newcommand{\recone}[2]{\textbf{rec}(#1, #2)}

\newcommand{\br}[2]{\textbf{br}(#1, #2)}
\newcommand{\rec}[3]{\textbf{rec}(#1, #2, #3)}

\newcommand{\brsymb}{\textbf{br}}
\newcommand{\recsymb}{\textbf{rec}}

\newcommand{\actions}{\mathsf{Actions}}
\newcommand{\anact}{\alpha}
\newcommand{\enregact}{\downarrow}
\newcommand{\dummyact}{{*}}
\newcommand{\diseqtestact}{{\ne}}
\newcommand{\eqtestact}{{=}}

\newcommand{\st}[1]{\mathsf{st}(#1)}
\newcommand{\data}[1]{\mathsf{data}(#1)}

\newcommand{\aval}{v}

\newcommand{\messages}{\mathcal{M}}
\newcommand{\amessage}{m}

\newcommand{\regnum}{r}
\newcommand{\operations}{\mathsf{Op}}

\newcommand{\query}{\phi}

\newcommand{\run}{\rho}

\newcommand{\agents}{\mathbb{A}}
\knowledgenewrobustcmd{\statesin}[1]{\mathsf{\cmdkl{cov}}\cmdkl{(}#1\cmdkl{)}}
\knowledgenewrobustcmd{\valsof}[1]{\mathsf{\cmdkl{val}}\cmdkl{(}#1\cmdkl{)}}
\knowledgenewrobustcmd{\agentsof}[1]{\mathsf{\cmdkl{ag}}\cmdkl{(}#1\cmdkl{)}}

\knowledgenewrobustcmd{\lessthan}{~\cmdkl{\trianglelefteq}~}

\newcommand{\boss}{\mathsf{b}}
\newcommand{\clique}{\mathsf{K}}
\newcommand{\gang}{\mathsf{G}}
\newcommand{\gangset}{\mathcal{G}}

\knowledgenewrobustcmd{\gangof}[2]{\mathsf{\cmdkl{gang}}_{#1}\cmdkl{(}#2\cmdkl{)}}
\knowledgenewrobustcmd{\bossof}[2]{\mathsf{\cmdkl{\mathsf{\boss}}}_{#1}\cmdkl{(}#2\cmdkl{)}}
\knowledgenewrobustcmd{\cliqueof}[2]{\mathsf{\cmdkl{\clique}}_{#1}\cmdkl{(}#2\cmdkl{)}}
\newcommand{\noboss}{\bot}

\newcommand{\cliquesucc}[3]{\overrightarrow{#1^{#2,#3}}}

\newcommand{\covset}{S}
\newcommand{\aconfig}{\sigma}
\newcommand{\aconfigs}[1]{\Sigma_{#1}}
\newcommand{\allaconfigs}{\Sigma}
\newcommand{\aconfiginit}{\aconfig_0}
\newcommand{\aconfiginitset}{\allaconfigs_{\mathsf{init}}}
\newcommand{\arun}{\nu}

\knowledgenewrobustcmd{\absproj}[2]{\cmdkl{\mathsf{abs}}_{#1}\cmdkl{(}#2\cmdkl{)}}

\newcommand{\agentbr}{a_{\brsymb}} 
\newcommand{\agentboss}{a_{\mathsf{boss}}}
\newcommand{\statebr}{q_{\brsymb}}

\newcommand{\val}{v}
\newcommand{\trace}[1]{\mathsf{tr}(#1)}
\knowledgenewrobustcmd{\Input}[1]{\cmdkl{\mathsf{In}}(#1)}
\knowledgenewrobustcmd{\vinput}[2]{\cmdkl{\mathsf{In}}_{#1}(#2)}
\knowledgenewrobustcmd{\Output}[1]{\cmdkl{\mathsf{Out}}(#1)}
\knowledgenewrobustcmd{\voutput}[2]{\cmdkl{\mathsf{Out}}_{#1}(#2)}
\knowledgenewrobustcmd{\vproj}[2]{{#2}\cmdkl{|}_{#1}}

\newcommand{\decsymb}{\mathtt{dec}}

\knowledgenewrobustcmd{\langdec}[1]{\cmdkl{\mathcal{L}}^{\mathtt{#1}}}

\newcommand{\tree}{\tau}
\newcommand{\node}{\mu}

\knowledgenewrobustcmd{\localrunlabel}[1]{\cmdkl{\mathbf{lr}}(#1)}
\knowledgenewrobustcmd{\valuelabel}[1]{\cmdkl{\mathbf{val}}(#1)}
\knowledgenewrobustcmd{\bosslabel}[1]{\cmdkl{\mathbf{bw}}(#1)}
\knowledgenewrobustcmd{\followlabelword}[1]{\cmdkl{\mathbf{fw}}(#1)}
\knowledgenewrobustcmd{\followlabelmessage}[1]{\cmdkl{\mathbf{fm}}(#1)}
\knowledgenewrobustcmd{\speclabel}[1]{\cmdkl{\mathbf{spec}}(#1)}
\newcommand{\spec}{\mathsf{spec}}
\newcommand{\bossspec}{\mathsf{bw}}
\newcommand{\followwordspec}{\mathsf{fw}}
\newcommand{\followmessagespec}{\mathsf{fm}}

\knowledgenewrobustcmd{\TARGET}{\withkl{\kl[target problem]}{\cmdkl{\textsc{Target}}}}
\newcommand{\Loc}{\text{Loc}}
\newcommand{\cpt}{\ensuremath{\mathtt{x}}}
\newcommand{\Cpt}{\ensuremath{\mathtt{X}}}
\newcommand{\dec}[1]{\ensuremath{\mathtt{#1}-}}
\newcommand{\inc}[1]{\ensuremath{\mathtt{#1}+}}
\newcommand{\testz}[1]{\ensuremath{\mathtt{#1}=0?}}
\newcommand{\stepMM}[1]{\xrightarrow{#1}}
\newcommand{\altitude}[1]{\mathbf{alt}(#1)}
\newcommand{\altmax}{\mathbf{altmax}}
\newcommand{\altmin}{\mathbf{altmin}}

\newcommand{\localrun}{u}
\newcommand{\localdata}{\nu}

\newcommand{\pstep}[1]{\step{#1}_{p}}

\newcommand{\towerfun}{\psi}

\newcommand{\Vinit}{W}

\knowledgenewrobustcmd{\subword}{~\cmdkl{\preceq}~}

\newcommand{\binrel}[3]{#1 \mathrel{#2} #3}

\newcommand{\los}{\mathcal{L}}
\newcommand{\lstep}[1]{\xrightarrow{#1}_{\los}}

\newcommand{\lstates}{L}
\newcommand{\lstate}{l}
\newcommand{\ltrans}{d}
\newcommand{\ltransitions}{D}
\newcommand{\startstate}[1]{\mathbf{s}(#1)}
\newcommand{\transstateone}[1]{\mathbf{t}(#1)}
\newcommand{\transstatetwo}[1]{\mathbf{u}(#1)}
\newcommand{\finstate}[1]{\mathbf{f}(#1)} 
\newcommand{\waitstate}{\mathbf{wait}}

\newcommand{\popact}[1]{\mathsf{read}(#1)}
\newcommand{\pushact}[1]{\mathsf{write}(#1)}
\newcommand{\atrace}{\mathsf{tr}}

\knowledgenewrobustcmd{\memoryproj}[1]{\cmdkl{\mathsf{mem}}(#1)}
\knowledgenewrobustcmd\perfproj[1]{\cmdkl{\mathsf{perf}}(#1)}

\newcommand{\tuple}[1]{\langle #1 \rangle}

\newcommand{\exready}{\textsf{rdy}}
\newcommand{\exgotwo}{\textsf{go}}
\newcommand{\exgothree}{\textsf{hlt}}
\newcommand{\regbar}[2]{
	\node (reg1) at (#1, #2+0.15-0.5) {reg $1$};
	\node (reg2) at (#1, #2+0.15-1) {reg $2$};
	\node (reg3) at (#1, #2+0.15-1.5) {reg $3$};
}

\newcommand{\regbartworeg}[2]{
	\node (reg1) at (#1, #2+0.15-0.5) {reg $1$};
	\node (reg2) at (#1, #2+0.15-1) {reg $2$};
}

\newcommand{\onerow}[9]{
\draw  (#1,#2) rectangle (#1+1, #2-0.5);
\draw[fill = #5, opacity = 0.4]  (#1,#2) rectangle (#1+1, #2-0.5);
\draw  (#1,#2-0.5) rectangle (#1+1, #2-1);
\draw[fill = #7, opacity = 0.4]  (#1,#2-0.5) rectangle (#1+1, #2-1);
\draw  (#1,#2-1) rectangle (#1+1, #2-1.5);
\draw[fill = #9, opacity = 0.4]  (#1,#2-1) rectangle (#1+1, #2-1.5);
\node [] (state) at (#1 +0.5 ,#2 + 0.2) {#3};
\node[] (reg1) at (#1 +0.5 ,#2 + 0.25 - 0.5) {#4};
\node[] (reg2) at (#1 +0.5 ,#2 + 0.25 -1) {#6};
\node[] (reg3) at (#1 +0.5 ,#2 + 0.25 -1.5) {#8};
}
\newcommand{\transtable}[3]{
	\node [align = center] at (#1, #2+0.2) {$\rightarrow$};
	\node [align = center, font = {\scriptsize}] at (#1,#2+0.6) {#3};
}

\newcommand{\onerowtworeg}[7]{
	\draw  (#1,#2) rectangle (#1+1, #2-0.5);
	\draw[fill = #5, opacity = 0.4]  (#1,#2) rectangle (#1+1, #2-0.5);
	\draw  (#1,#2-0.5) rectangle (#1+1, #2-1);
	\draw[fill = #7, opacity = 0.4]  (#1,#2-0.5) rectangle (#1+1, #2-1);
	\node [] (state) at (#1 +0.5 ,#2 + 0.2) {#3};
	\node[] (reg1) at (#1 +0.5 ,#2 + 0.25 - 0.5) {#4};
	\node[] (reg2) at (#1 +0.5 ,#2 + 0.25 -1) {#6};
}

\usepackage{skull}
\usetikzlibrary{backgrounds}

\creflabelformat{listlem}{#2\thethm.#1#3}

\DeclareMathSymbol{\mlq}{\mathord}{operators}{``}
\DeclareMathSymbol{\mrq}{\mathord}{operators}{`'}
\newcommand{\quotemarks}[1]{\mlq #1 \mrq}

%% file: kl.tex
\definecolor{Blue Sapphire}{HTML}{002346} 
\definecolor{Gamboge}{HTML}{ee9b00}
\definecolor{Ruby Red}{HTML}{800000}

\IfKnowledgePaperModeTF{
}{
	\knowledgestyle{intro notion}{color={Ruby Red}, emphasize}
	\knowledgestyle{notion}{color={Blue Sapphire}}
	\hypersetup{
		colorlinks=true,
		breaklinks=true,
		linkcolor={Blue Sapphire}, 
		citecolor={Blue Sapphire}, 
		filecolor={Blue Sapphire}, 
		urlcolor={Blue Sapphire},
	}
}
\IfKnowledgeCompositionModeTF{
	\knowledgeconfigure{anchor point color={Ruby Red}, anchor point shape=corner}
	\knowledgestyle{intro unknown}{color={Gamboge}, emphasize}
	\knowledgestyle{intro unknown cont}{color={Gamboge}, emphasize}
	\knowledgestyle{kl unknown}{color={Gamboge}}
	\knowledgestyle{kl unknown cont}{color={Gamboge}}
}{
}

\knowledge{notion}
| signature BNRA
| signature \BNRA 

\knowledge{notion}
| broadcast-only

\knowledge{notion}
| reception-only

\knowledge{notion}
| signature
| signature protocol
| signature protocols

\knowledge{notion}
| specification@sg
| specifications@sg

\knowledge{notion}
| unfolding trees@signature
| unfolding tree@sg
| node@sg

\knowledge{notion}
| coverability witness@sg

\knowledge{notion}
| size@treesg
| size@sg

\knowledge{notion}
| length

\knowledge{notion}
| active@sg

\knowledge{notion}
| signature case

\knowledge{notion}
| missing broadcast
| missing broadcasts

\knowledge{notion}
|Broadcast Network of Register Automata
|BNRA

\knowledge{notion}
| copycat principle
| copycat

\knowledge{notion}
| protocol
| protocols

\knowledge{notion}
| notion
| definition

\knowledge{notion}
| broadcasts
| broadcast

\knowledge{notion}
| receptions
| reception

\knowledge{notion}
| local tests
| local test
| local equality test
| local equality tests
| local disequality tests
| local disequality test

\knowledge{notion}
|coverability witness
|coverability witnesses

\knowledge{notion}
| simple
| simple protocol
| simple protocols

\knowledge{notion}
| size

\knowledge{notion}
| step
| steps

\knowledge{notion}
| agents
| agent

\knowledge{notion}
| covers

\knowledge{notion}
| configuration
| configurations

\knowledge{notion}
| initial configuration
| initial configurations

\knowledge{notion}
| transition
| transitions

\knowledge{notion}
| action
| actions

\knowledge{notion}
| equality test

\knowledge{notion}
| disequality test

\knowledge{notion}
| store action

\knowledge{notion}
| dummy action

\knowledge{notion}
| run
| runs

\knowledge{notion}
| Length function theorem

\knowledge{notion}
| subword 
| subword order

\knowledge{notion}
| size@tree 

\knowledge{notion}
| message type
| message types

\knowledge{notion}
| good

\knowledge{notion}
| bad
| bad sequence

\knowledge{notion}
| initial run
| initial@run
| initial runs

\knowledge{notion}
| initial partial run
| initial@partial
| initial partial runs

\knowledge{notion}
| input
| $v$-input
| $v$-inputs
| $v'$-input
| $\aval$-input
| $\aval'$-input
| $\val$-input
| $\val'$-input

\knowledge{notion}
| output
| $v$-output
| $v$-outputs
| $v'$-output
| $\aval$-output
| $\aval'$-output
| $\val$-output
| $\val'$-output

\knowledge{notion}
| local run
| local runs

\knowledge{notion}
| local configuration
| local configurations

\knowledge{notion}
| query
| queries

\knowledge{notion}
| query coverability problem

\knowledge{notion}
| contradictory
| non-contradictory

\knowledge{notion}
| abstract configuration
| abstract configurations

\knowledge{notion}
| abstract step
| abstract steps

\knowledge{notion}
| abstract run
| abstract runs

\knowledge{notion}
| gang
| gangs

\knowledge{notion}
| clique

\knowledge{notion}
| Broadcast from boss
| broadcast from boss

\knowledge{notion}
| Broadcast from clique
| broadcast from clique

\knowledge{notion}
| External broadcast
| external broadcast

\knowledge{notion}
| Gang reset
| gang reset
| gang resets
| resets
| reset

\knowledge{notion}
| boss nodes
| boss node
| boss
| bosses

\knowledge{notion}
| boss specification
| boss specifications
| Boss specification
| Boss specifications

\knowledge{notion}
| follower
| followers
| follower node
| Follower node
| follower nodes

\knowledge{notion}
| follower specification
| follower specifications
| Follower specification
| Follower specifications

\knowledge{notion}
| unfolding tree
| unfolding trees
\knowledge{notion}
| partial run
| partial runs

\knowledge{notion}
| altitude
| altitudes

\knowledge{notion}
| admits decomposition
| admit decomposition

\knowledge{notion}
| decomposition
| decompositions

\knowledge{notion}
| trace
| traces

\knowledge{notion}
| specification
| specifications

\knowledge{notion}
| active
| active register
| active registers

\knowledge{notion}
| shortening property

\knowledge{notion}
| local step
| local steps

\knowledge{notion}
| reception step
| reception steps

\knowledge{notion}
| internal step
| internal steps

\knowledge{notion}
| unmatched reception
| unmatched receptions

\knowledge{notion}
| Cover
| cover problem
| coverability problem
| coverability
| \textsc{Cover}

\knowledge{notion}
| Target
| target problem
| target reachability problem
| \textsc{Target}

\knowledge{notion}
| lossy channel system
| lossy channel systems
| Lossy channel systems
| LCS

\knowledge{notion}
| root

\knowledge{notion}
| link

\knowledge{notion}
| push

\knowledge{notion}
| pop

\knowledge{notion}
| reachability problem@lcs

\knowledge{notion}
| memory slot
| memory slots

\knowledge{notion}
| perfect configurations
| perfect configuration
| perfect

\knowledge{notion}
| internal test@global

\knowledge{notion}
| internal message@global
| internal messages@global

\knowledge{notion}
| external message@global
| external messages@global

\knowledge{notion}
| internal step@global

\knowledge{notion}
| partial step
| partial steps

\knowledge{notion}
| projection

\knowledge{notion}
| input@partial

\knowledge{notion}
| output@partial

\knowledge{notion}
| subtree
| subtrees

\knowledge{notion}
| initial
| initial value
| initial values
| Initial Values

\knowledge{notion}
| non-initial
| non-initial value
| non-initial values
| Non-initial Values

%% file: Intro.tex
\section{Introduction}
We consider Broadcast Networks of Register Automata (BNRA), a model for networks of agents communicating by broadcasts. These systems are composed of an arbitrary number of agents whose behavior is specified with a finite automaton. This automaton is equipped with a finite set of private registers that contain values from an infinite unordered set. Initially, registers all contain distinct values, so these values can be used as identifiers. 
A broadcast message is composed of a symbol from a finite alphabet along with the value of one of the sender's registers. When an agent broadcasts a message, any subset of agents may receive it; this models unreliable systems with unexpected crashes and disconnections. Upon reception, an agent may store the received value or test it for equality with one of its register values. For example, an agent can check that several received messages have the same value.

This model was introduced in \cite{DelzannoST13}, as a natural extension of Reconfigurable Broadcast Networks~\cite{DelzannoSZ2010Adhoc}. In \cite{DelzannoST13}, the authors established that coverability is undecidable if the agents are allowed to send two values per message. They moreover claimed that, with one value per message, coverability was decidable and \PSPACE-complete; however, the proof turned out to be incorrect \cite{ArnaudErratum}. As we will see, the complexity of that problem is in fact much higher. 

In this paper we establish the decidability of the "coverability problem" and its completeness for the hyper-Ackermannian complexity class $\Fcomplexity{\omega^\omega}$, showing that the problem has nonprimitive recursive complexity. The lower bound comes from "lossy channel systems", which consist (in their simplest version) of a finite automaton that uses an unreliable FIFO memory from which any letter may be erased at any time \cite{AbdullaJ1996verif, Schnoebelen2002verifying,ChambartS08ordinal}. 
We further establish that our model lies at the frontier of decidability by showing undecidability of the target problem (where all agents must synchronize in a given state). We contrast these results with the \NP-completeness of the "coverability problem" if each agent has only one register. 

\paragraph*{Related work} 
Broadcast protocols are a widely studied class of systems in which processes are represented by nodes of a graph and can send messages to their neighbors in the graph. There are many versions depending on how one models processes, the communication graph, the shape of messages... 
A model with a fully connected communication graph and messages ranging over a finite alphabet was presented in~\cite{emerson1998model}. When working with parameterized questions over this model (\emph{i.e.}, working with systems of arbitrary size), many basic problems are undecidable~\cite{EsparzaFM1999verification}; similar negative results were found for Ad Hoc Networks where the communication graph is fixed but arbitrary \cite{DelzannoSZ2010Adhoc}. This lead the community to consider Reconfigurable Broadcast Networks (RBN) where a broadcast can be received by an arbitrary subset of agents~\cite{DelzannoSZ2010Adhoc}.

Parameterized verification problems over RBN have been the subject of extensive study in recent years, concerning for instance reachability questions~\cite{DelzannoSTZ12, BalasubramanianGW22}, liveness~\cite{DBLP:journals/computing/ChiniMS22} or alternative communication assumptions~\cite{Balasubramanian18}; however, RBN have weak expressivity, in particular because agents are anonymous. In~\cite{DelzannoST13}, RBN were extended to BNRA, the model studied in this article, by the addition of registers allowing processes to exchange identifiers. 

Other approaches exist to define parameterized models with registers~\cite{BolligRS21}, such as dynamic register automata in which processes are allowed to spawn other processes with new identifiers and communicate integers values~\cite{AbdullaAKR14}. While basic problems on these models are in general undecidable, some restrictions on communications allow to obtain decidability~\cite{AbdullaAKR15, Rezine17}.

Parameterized verification problems often relate to the theory of well quasi-orders
and the associated high complexities obtained from bounds on the length of sequences with no increasing pair (see for example \cite{WSTS}). 
In particular, our model is linked to data nets, a classical model connected to well-quasi-orders. Data nets are Petri nets in which tokens are labeled with natural numbers and can exchange and compare their labels using inequality tests \cite{LazicNORW08}; in this model, the "coverability problem" is $\Fcomplexity{\omega^{\omega^{\omega}}}$-complete \cite{datanetsinequalityfomegaomegaomega}. When one restricts data nets to only equality tests, the "coverability problem" becomes $\Fomegaomega$-complete~\cite{Rosa-Velardo17}. Data nets with equality tests do not subsume BNRA. Indeed, in data nets, each process can only carry one integer at a time, and problems on models of data nets where tokens carry tuples of integers are typically undecidable \cite{Lasota16}.

\paragraph*{Overview}
We start with the model definition and some preliminary results in Section~\ref{sec:preliminaries}. As our decidability proof is quite technical, we start by proving decidability of the coverability problem in a subcase called \emph{signature protocols} in Section~\ref{sec:cover-decidability}.
We then rely on the intuitions built in that subcase to generalize the proof to the general case in Section~\ref{sec:cover-general-case}. We also show the undecidability of the closely-related "target problem".
Finally, we prove the \NP-completeness of the coverability problem for protocols with one register in Section~\ref{sec:cover-1BNRA}.
Due to space constraints, a lot of proofs, as well as some technical definitions, are only sketched in this version.
Detailed proofs can be found in the appendix.

%% file: General-defs.tex
\section{Preliminaries}
\label{sec:preliminaries}

\subsection{Definitions of the Model}
A ""Broadcast Network of Register Automata"" ("BNRA") \cite{DelzannoST13} is a model describing broadcast networks of agents with local registers. A finite transition system describes the behavior of an agent; an agent can broadcast and receive messages with integer values, store them in local registers and perform (dis)equality tests. 
There are arbitrarily many agents. When an agent broadcasts a message, every other agent may receive it, but does not have to do so.

\begin{definition}
	A ""protocol"" with $\regnum$ registers is a tuple $\prot = (Q, \messages, \transitions, q_0)$  with $Q$ a finite set of states, $q_0 \in Q$ an initial state, $\messages$ a finite set of ""message types""  and $\transitions \subseteq Q \times \operations \times Q$ a finite set of transitions, with operations $\operations =$
	\[
	 \set{\br{\amessage}{i}, \rec{\amessage}{i}{\dummyact}, \rec{\amessage}{i}{\enregact}, \rec{\amessage}{i}{\eqtestact}, \rec{\amessage}{i}{\diseqtestact} \hspace{-1.5pt}\mid\hspace{-1.5pt} \amessage \in \messages, 1 \leq i \leq \regnum}.\]
	Label $\brsymb$ stands for ""broadcasts"" and $\recsymb$ for ""receptions"".
	In a reception $\rec{\amessage}{i}{\anact}$, $\anact$ is its \emph{action}. 
The set of actions is $\actions := \set{\eqtestact, \diseqtestact, \enregact, \dummyact}$, where 
$\quotemarks{\eqtestact}$ is an \emph{equality test}, $\quotemarks{\diseqtestact}$ is a \emph{disequality test}, $\quotemarks{\enregact}$ is a \emph{store action} and $\quotemarks{\dummyact}$ is a \emph{dummy action} with no effect.
The ""size"" of $\prot$ is $\size{\prot} := \size{Q} + \size{\messages} + \size{\transitions} + r$.

\end{definition}

We now define the semantics of those systems. Essentially, we have a finite set of agents with $r$ registers each; all registers initially contain distinct values. A step consists of an agent broadcasting a message that other agents may receive.

\begin{definition}[Semantics]
	Let $(Q,\messages, \transitions, q_0)$ be a "protocol" with $\regnum$ registers, and $\agents$ a finite non-empty set of ""agents"".
	A ""configuration"" over $\agents$ is a function $\config : \agents \to Q \times \nats^{\regnum}$ mapping each agent to its state and its register values. 
	We write $\st{\config}$ for the state component of $\config$ and $\data{\config}$ for its register component.
	
	\AP An \intro{initial configuration} $\config$ is one where for all $a \in \agents$, $\st{\config}(a) = q_0$ and $\data{\config}(a, i) \neq \data{\config}(a', i')$ for all $(a,i) \neq (a', i')$.
	
	\AP Given a finite non-empty set of agents $\agents$ and two "configurations" $\config, \config'$ over $\agents$, a ""step"" $\config \step{} \config'$ is defined when there exist $\amessage \in \messages$, $a_0 \in \agents$ and $i \in \nset{1}{r}$ such that \linebreak $(\st{\config}(a_0),\br{\amessage}{i}, \st{\config'}(a_0)) \in \transitions$, $\data{\config}(a_0) = \data{\config'}(a_0)$ and, for all $a \ne a_0$, either $\config'(a) = \config(a)$ or there exists $(\st{\config}(a),\rec{m}{j}{\anact},\st{\config'}(a)) \in \transitions$
		s.t. $\data{\config'}(a, j') = \data{\config}(a, j')$ for $j' \neq j$ and:
		\begin{itemize}
				\item if $\anact = \quotemarks{\dummyact}$ 
				then $\data{\config'}(a,j) = \data{\config}(a,j)$,
				\item if $\anact = \quotemarks{\enregact}$ then $\data{\config'}(a,j) = \data{\config}(a_0,i)$,
				\item if $\anact = \quotemarks{\eqtestact}$ then $\data{\config'}(a,j) = \data{\config}(a,j) =\data{\config}(a_0,i)$,
				\item if $\anact = \quotemarks{\diseqtestact}$ then $\data{\config'}(a,j) = \data{\config}(a,j) \ne \data{\config}(a_0,i)$.
			\end{itemize}
	
	\AP A ""run"" over $\agents$ is a sequence of steps $\run : \config_0 \step{} \config_1 \step{} \cdots \step{} \config_k$ with $\gamma_0, \dots, \gamma_k$ configurations over $\agents$. 
	We write $\config_0 \step{*} \config_k$ when there exists such a "run".
	A "run" is ""initial@@run"" when $\config_0$ is an "initial configuration".  
\end{definition}
\begin{remark}
\label{rem:several_values_per_message}
In our model, agents may only send one value per message. Indeed, "coverability" is undecidable if agents can broadcast several values at once~\cite{DelzannoST13}.
\end{remark}

\begin{example}\label{ex:example-1}

	Figure \ref{fig:ex1} shows a "protocol" with 2 registers. 
	Let $\agents = \{a_1, a_2\}$. We denote by $\tuple{\st{\config}(a_1), \data{\config}(a_1), \st{\config}(a_2),\data{\config}(a_2)}$ a "configuration" $\config$ over $\agents$. The following sequence is an "initial run":\vspace{-0.2cm}
	\begin{multline*}
		\tuple{q_0, (1,2), q_0,(3,4)} \step{} \tuple{q_1, (1,2), q_2,(1,4)} \step{} 
		\tuple{q_3, (1,4), q_3,(1,4)} \\ \step{} \tuple{q_4, (1,4), q_3,(1,4)} \step{} \tuple{q_4, (1,4), q_4,(1,4)}
	\end{multline*}
	The broadcast messages are, in this order: $(m_2,1)$ by $a_1$, $(m_3,4)$ by $a_2$, $(m_4,1)$ by $a_2$ and $(m_4,1)$ by $a_1$. In this "run", each broadcast message is received by the other agent; in general, however, this does not have to be true.
\end{example}

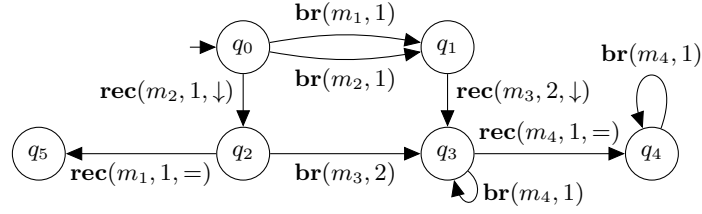
\begin{figure}[t]
	\centering
		\input{Figures/fig-ex1}
	\caption{Example of a "protocol".}\label{fig:ex1}
\end{figure}

\begin{remark}
	\label{rem:copycat-principle}
	From a run $\run : \config_0 \step{\ast} \config$, we can build a larger run $\run'$ in which, for each agent $a$ of $\run$, there are arbitrarily many extra agents in $\run'$ that end in the same state as $a$, all with distinct register values. To obtain this, $\run'$ make many copies of $\run$ run in parallel on disjoint sets of agents. Because all these copies of $\run$ do not interact with one another and because all agents start with distinct values in "initial configurations", the different copies of $\run$ have no register values in common. This property is called ""copycat principle"": if state $q$ is coverable, then for all $n$ there exists an augmented "run" which puts $n$ agents on $q$.
\end{remark}

\begin{definition}

	\AP The ""coverability problem"" \intro*\COVER\ asks, given a "protocol" $\prot$ and a state $q_f$, whether there is a finite non-empty set of agents $\agents$, an "initial run" $\config_0 \step{*} \config_f$ over $\agents$ that ""covers"" $q_f$, \emph{i.e.}, there is $a \in \agents$  such that $\st{\config_f}(a)  = q_f$.
	
	\AP The ""target problem"" \intro*\TARGET\ asks, given a "protocol" $\prot$ and a state $q_f$, whether there is there is a finite non-empty set of agents $\agents$ and an "initial run" $\config_0 \step{*} \config_f$ over $\agents$  such that, for every $a \in \agents$, $\st{\config_f}(a) = q_f$, \emph{i.e.}, all agents end on $q_f$.
\end{definition}

\begin{example}\label{example-2}
	Let $\prot$ the "protocol" of Figure~\ref{fig:ex1}. As proven in Example~\ref{ex:example-1}, $(\prot,q_4)$ is a positive instance of \COVER~and \TARGET. However, let $\prot'$ the protocol obtained from $\prot$ by removing the loop on $q_4$; $(\prot',q_4)$ becomes a negative instance of \TARGET. Indeed, there must be an agent staying on $q_3$ to broadcast $m_4$. Also, $(\prot, q_5)$ is a negative instance of \COVER: we would need to be able to have one agent on $q_2$ and one agent on $q_0$ with the same value in their first registers. However, an agent in $q_0$ has performed no transition so it cannot share register values with other agents.
\end{example}

\begin{remark}
	In \cite{DelzannoST13}, the authors consider the \AP \emph{query problem} where one looks for a run reaching a "configuration" satisfying some \emph{queries}.	
	In fact, this problem exponentially reduces to \COVER\ hence our complexity result of $\Fcomplexity{\omega^\omega}$ also holds for the query problem. In the case with one register, one can even find a polynomial-time reduction hence our \NP result also holds with queries. 
\end{remark}

\AP We finally introduce \intro{signature BNRA}, an interesting restriction of our model where register $1$ is ""broadcast-only"" and all other registers are ""reception-only"". Said otherwise, the first register acts as a permanent identifier with which agents sign their messages. An example of such a protocol is displayed in \cref{fig:example-signature-protocol}. Under this restriction, a message is composed of a message type along with the identifier of the sender. This restriction is relevant for pedagogical purposes: we will see that it falls into the same complexity class as the general case but makes the decidability procedure simpler. 

\begin{definition}[Signature protocols]
A ""signature protocol"" with $\regnum$ registers is a "protocol" $\prot = (Q, \messages, \transitions, q_0)$ where register $1$ appears only in "broadcasts" in $\transitions$ and registers $i \geq 2$ appear only in "receptions" in $\transitions$. 
\end{definition}

\subsection{Classical Definitions}

\paragraph*{Fast-growing hierarchy}

For $\alpha$ an ordinal in Cantor normal form, we denote by $\Ffunction{\alpha}$ the class of functions corresponding to level $\alpha$ in the Fast-Growing Hierarchy. We denote by $\Fcomplexity{\alpha}$ the associated complexity class and use the notion of $\Fcomplexity{\alpha}$-completeness. All these notions are defined in \cite{Schmitz16}. We will specifically work with complexity class $\Fcomplexity{\omega^{\omega}}$. For readers unfamiliar with these notions, $\Fcomplexity{\omega^{\omega}}$-complete problems are decidable but with very high complexity (non-primitive recursive, and even much higher than the Ackermann class $\Fcomplexity{\omega}$). 

We highlight that our main result is the decidability of the problem. We show that the problem lies in $\Fcomplexity{\omega^{\omega}}$ because it does not complicate our decidability proof significantly; also, it fits nicely into the landscape of high-complexity problems arising from well quasi-orders. 

\paragraph*{Well-quasi orders}

For our decidability result, we rely on the theory of well quasi-orders in the context of subword ordering.
Let $\Sigma$ be a finite alphabet, $w_1, w_2 \in \Sigma^*$, $w_1$ is a ""subword"" of $w_2$, denoted $w_1 \intro*\subword w_2$, when $w_1$ can be obtained from $w_2$ by erasing some letters. 
A sequence of words $w_0, w_1, \ldots$ is \emph{good} if there exist $i<j$ such that $w_i \subword w_j$, and ""bad"" otherwise. Higman's lemma \cite{Higman52} states that every "bad" sequence of words over a finite alphabet is finite, but there is no uniform bound.
In order to bound the length of all "bad" sequences, one must bound the growth of the sequence of words. 
We will use the following result, known as the Length function theorem \cite{SchmitzS2011upperHigman}:

\begin{theorem}[""Length function theorem"" \cite{SchmitzS2011upperHigman}]
	\label{thm:lengthfcttheorem}
	Let $\Sigma$ a finite alphabet and $g : \nats \to \nats$ a primitive recursive function.
	There exists a function $f \in \Ffunction{\omega^{\size{\Sigma} - 1}}$ such that, for all $n \in \nats$, every "bad" sequence $w_1, w_2, \ldots$ such that $\size{w_i} \leq g^{(i)}(n)$ for all $i$ has at most $f(n)$ terms (where $g^{(i)}$ denotes $g$ applied $i$ times). 
\end{theorem}

\subsection{A Complexity Lower Bound for COVER Using LCS}

""Lossy channel systems"" ("LCS") are systems where finite-state processes communicate by sending messages from a finite alphabet through lossy FIFO channels. Unlike in the non-lossy case \cite{BZ83}, reachability of a state is decidable for "lossy channel systems" \cite{AbdullaJ1996verif}, but has non-primitive recursive complexity \cite{Schnoebelen2002verifying} and is in fact $\Fcomplexity{\omega^{\omega}}$-complete \cite{ChambartS08ordinal}. 
By simulating LCS using BNRA, we obtain our $\Fcomplexity{\omega^{\omega}}$ lower bound for the "coverability problem":

\begin{restatable}{proposition}{propReductionLCS}
	\label{prop:reduction-LCS}
	\COVER\ for "signature BNRA" is $\Fcomplexity{\omega^\omega}$-hard.
\end{restatable}
\begin{proof}[Proof sketch]
Given an "LCS" $\los$, we build a "signature protocol" $\prot$ with two registers. 
Each agent starts by receiving a foreign identifier and storing it in its second register; using equality tests, it then only accepts messages with this identifier.
Each agent has at most one predecessor, so the communication graph is a forest where messages propagate from roots to leaves.
 Each branch simulates an execution of $\los$. Each agent of the branch simulates a step of the execution: it receives from its predecessor a "configuration" of $\los$, chooses the next "configuration" of $\los$ and broadcasts it, sending first the location of $\los$ and then, letter by letter, the content of the channel. 
 It could be that some messages are not received, hence the lossiness.  
The full proof can be found in Appendix~\ref{app:reduction-lcs}. 
\end{proof}


%% file: Figures/fig-ex1.tex
\begin{tikzpicture}[xscale=0.5,AUT style,node distance=2cm,auto,>= triangle
	45]
	\tikzstyle{initial}= [initial by arrow,initial text=,initial
	distance=.7cm]
	
	\node[state,initial, minimum width=0.1pt] (0) at (0,0) {$q_0$};
	\node[state] [right = of 0, yshift = 0] (1) {$q_1$};
	\node[state] [below = 0.7 of 1] (4) {$q_3$};
	\node[state] [below = 0.7 of 0, yshift = 0] (5) {$q_2$};
	\node[state] [right = of 4] (7) {$q_4$};
	\node[state] [left = of 5] (5b) {$q_5$};
	
	\path[->] 	
	(0) edge [bend left = 5]  node [above, xshift = 0] {\small$\br{m_1}{1}$} (1)
	edge [bend right = 5]  node [below, xshift = 0] {\small$\br{m_2}{1}$} (1)
	edge node [left, xshift = 0, yshift = 3] {\small$\rec{m_2}{1}{\enregact}$} (5)
	(1) edge node  [right,yshift = 3] {\small$\rec{m_3}{2}{\enregact}$} (4)
	(5) edge node [below, xshift = 0, yshift = 0]  {\small$\br{m_3}{2}$} (4)
	(4)  edge [out=330,in=290,looseness=4] node[yshift = 10] {\small$\br{m_4}{1}$} (4)
	(4) edge [bend right = 0] node [above]{\small$\rec{m_4}{1}{=}$} (7)
	(7)  edge [out=60,in=120,looseness=4] node[above] {\small$\br{m_4}{1}$} (7)
	(5) edge node {\small$\rec{m_1}{1}{=}$} (5b)
	;
	
\end{tikzpicture}

%% file: General-abstraction.tex
\section{Coverability Decidability for Signature Protocols}
\label{sec:cover-decidability}

This section and the next one are dedicated to the proof of our main result:

\begin{restatable}{theorem}{decidablecover}
\label{thm:decidable-cover}
\COVER\ for BNRA is decidable and $\Fcomplexity{\omega^\omega}$-complete.
\end{restatable}

For the sake of clarity, in this section, we will first focus on the case of "signature BNRA". 
As a preliminary, we start by defining a notion of "local run" meant to represent the projection of a run onto a given agent.

\subsection{Local runs}
\AP A ""local configuration"" is a pair $(q, \localdata) \in Q \times \nats^r$.  
\AP An ""internal step"" from $(q,\localdata)$ to $(q',\localdata')$ with transition $\atrans \in \transitions$, denoted $(q,\localdata) \intstep{\atrans} (q',\localdata')$, is defined when $\localdata = \localdata'$ and $\atrans =(q, \br{m}{i}, q')$ is a "broadcast".  
\AP A ""reception step"" from $(q,\localdata)$ to $(q',\localdata')$ with transition $\atrans \in \transitions$ and value $\aval \in \nats$, denoted $(q,\localdata) \extbr{\atrans}{\aval} (q',\localdata')$, is defined when $\atrans$ is of the form $(q,\rec{m}{j}{\anact},q')$ with $\localdata(j') = \localdata'(j')$ for all $j' \neq j$ and:
	
	\begin{minipage}[t]{5.6cm}
		\begin{itemize}
			\item if $\anact = \quotemarks{\dummyact}$ 
			then $\localdata(j) = \localdata'(j)$,
			\item if $\anact = \quotemarks{\enregact}$ then $\localdata'(j) = v$,
		\end{itemize}
	\end{minipage}
	\begin{minipage}[t]{5.6cm}
		\begin{itemize}
			\item if $\anact = \quotemarks{\eqtestact}$ then $\localdata(j) = \localdata'(j)= v$,
			\item if $\anact = \quotemarks{\diseqtestact}$ then $\localdata(j) = \localdata'(j) \ne v$.
		\end{itemize}
	\end{minipage}

	\AP Such a reception step corresponds to receiving message $(m,v)$; in a "local run", one does not specify the origin of a received message.
	A ""local step"" $(q,\localdata) \step{} (q',\localdata')$ is either a "reception step" or an "internal step". A ""local run"" $\localrun$ is a sequence of "local steps" denoted $(q_0, \nu_0) \step{*} (q, \nu)$. Its ""length"" $\size{\localrun}$ is its number of steps. 
	
	A value $\aval \in \nats$ appearing in $\localrun$ is ""initial"" if it appears in $\nu_0$ and ""non-initial"" otherwise. 
	For $\aval \in \nats$, the $\aval$-""input"" $\intro*\vinput{\aval}{\localrun}$ (resp. $\aval$-""output"" $\intro*\voutput{\aval}{\localrun}$) is the sequence $m_0 \cdots m_{\ell} \in \messages^*$ of message types received (resp. broadcast) with value $\aval$ in $\localrun$.

\subsection{Unfolding Trees}
\label{sec:unfolding_tree_signature}

We first prove decidability of \COVER\ for "signature BNRA". Note that, in "signature protocols", the initial values of "reception-only" registers are not relevant as they can never be shared with other agents. We deduce from this idea the following informal observation:
\begin{observation}
\label{obs:sg_reception} 
In "signature BNRA", when some agent receives a message, it can compare the value of the message only with the ones of previously received messages, \emph{i.e.}, check whether the sender is the same.
\end{observation}

If we want to turn a "local run" $u$ of an agent $a$ into an actual "run", we must match $a$'s receptions with broadcasts. Because of Observation~\ref{obs:sg_reception}, what matters is not the actual values of the receptions in $u$ but which ones are equal to which.  Therefore, for a value $v$ received in $u$, if $m_1 \dots m_k \in \messages^*$ are the message types received in $u$ with value $v$ in this order, it means that to execute $u$, $a$ need another agent $a'$ to broadcast messages types $m_1$ to $m_k$, all with the same value.  
We describe what an agent needs from other agents as a set of "specifications@@sg" which are words of $\messages^*$. 

To represent runs, we consider "unfolding trees@@signature" that abstract runs by representing such specifications, dependencies between them and how they are carried out. In this tree, each node is assigned a "local run" and the "specification@@sg" that it carries out. 
Because of copycat arguments, we will in fact be able to duplicate agents so that each agent only accomplishes one task, hence the tree structure.

\begin{definition}
\label{def:unfolding_tree_signature}
\AP An ""unfolding tree@@sg"" $\tree$ over $\prot$ is
a finite tree where nodes $\node$ have three labels:
\begin{itemize}
	\item a "local run" of $\prot$, written $\intro*\localrunlabel{\node}$;
	
	\item a value in $\nats$, written $\intro*\valuelabel{\node}$;
	
	\item a ""specification@@sg"" $\intro*\speclabel{\node} \in \messages^*$.
\end{itemize} 
Moreover, all nodes $\node$ in $\tree$ must satisfy the three following conditions:
\begin{enumerate}[label= (\roman*), ref=(\roman*)]
	\item \label{item:condition1_initial_value_sg} Initial values of $\localrunlabel{\node}$ are never received in $\localrunlabel{\node}$,
	\item \label{item:condition3_boss_node_sg} $\speclabel{\node} \subword \voutput{\valuelabel{\node}}{\localrunlabel{\node}}$, \emph{(recall that $\subword$ denotes the "subword" relation)}
	\item \label{item:condition2_non_initial_value_sg} For each value $\aval$ received in $\localrunlabel{\node}$, $\node$ has a child $\node'$ s.t. $\vinput{\aval}{\localrunlabel{\node}} \subword \speclabel{\node'}$.
\end{enumerate}

\AP Lastly, given $\tree$ an "unfolding tree@@sg", we define its ""size@@treesg"" by $\size{\tree} := \sum_{\node \in \tree} \size{\node}$ where $\size{\node} := \size{\localrunlabel{\node}} + \size{\speclabel{\node}}$. Note that the "size@@treesg" of $\tree$ takes into account the size of its nodes, so that a tree $\tree$ can be stored in space polynomial in $\size{\tree}$ (renaming the values appearing in $\tree$ if needed). 
\end{definition}

We explain this definition. Condition \ref{item:condition1_initial_value_sg} enforces that the "local run" cannot cheat by receiving its "initial values". 
 Condition \ref{item:condition3_boss_node_sg} expresses that $\localrunlabel{\node}$ broadcasts (at least) the messages of $\speclabel{\node}$. We can use the subword relation $\subword$ (instead of equality) because messages do not have to be received.  
Condition \ref{item:condition2_non_initial_value_sg} expresses that, for each value $v$ received in the "local run" $\localrunlabel{\node}$, $\node$ has a child who is able to broadcast the sequence of messages that $\localrunlabel{\node}$ receives with value $v$. 
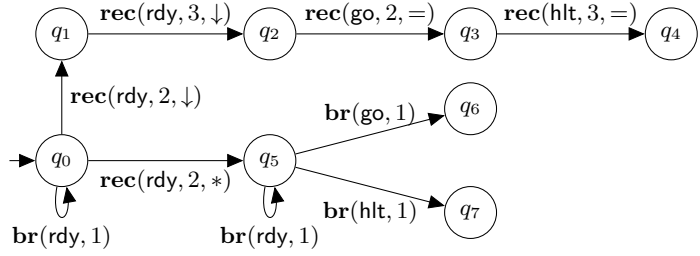
\begin{figure}[t]
	\centering
	\resizebox*{!}{3.5cm}{
	\input{Figures/fig-ex2}
	}
	\caption{Example of a "signature protocol".}\label{fig:ex2}\label{fig:example-signature-protocol}
\end{figure}
\begin{example}
\label{ex:protocol-signature}
	Figure~\ref{fig:example-signature-protocol} provides an example of a "signature protocol".
Let $\agents = \set{a_1, a_2,a_3}$. We denote a configuration $\config$ by $\tuple{\st{\config}(a_1),(\data{\config}(a_1)), \\ \st{\config}(a_2),(\data{\config}(a_2)), \st{\config}(a_3), (\data{\config}(a_3))}$. Irrelevant register values are denoted by $\_$. Let $\run$ be the run over $\agents$ of initial configuration  \\$\tuple{q_0, (1,\_,\_), q_0, (2,\_,\_), q_0, (3,\_,\_)}$ where the following occurs:
\begin{itemize}
\item $a_2$ broadcasts $\exready$, $a_1$ receives: $\tuple{q_1, (1,2,\_), q_0, (2,\_,\_), q_0, (3,\_,\_)}$,
\item $a_3$ broadcasts $\exready$, $a_1$ and $a_2$ receive: $\tuple{q_2, (1,2,3), q_5, (2,\_,\_), q_0, (3,\_,\_)}$,
\item $a_2$ broadcasts $\exready$, $a_3$ receives: $\tuple{q_2, (1,2,3), q_5, (2,\_,\_), q_5, (3,\_,\_)}$,
\item $a_2$ broadcasts $\exgotwo$, $a_1$ receives: $\tuple{q_3, (1,2,3), q_6, (2,\_,\_), q_5, (3,\_,\_)}$,
\item $a_3$ broadcasts $\exgothree$, $a_1$ receives: $\tuple{q_4, (1,2,3), q_6, (2,\_,\_), q_7, (3,\_,\_)}$.
\end{itemize}
	
Figure~\ref{fig:ex-unfolding-tree-signature} provides an "unfolding tree@@sg" derived from $\run$ by applying a procedure introduced later. Because agents $a_2$ and $a_3$ broadcast to several other agents, they each correspond to several nodes of the tree.  

	We explain why this tree is an "unfolding tree@@sg". Condition \ref{item:condition1_initial_value_sg} is trivially satisfied. 
	Condition \ref{item:condition3_boss_node_sg} holds at every node because the local run of each node exactly broadcasts the "specification@@sg" of the node. Condition  \ref{item:condition2_non_initial_value_sg} is satisfied at $\node_1$: $\vinput{2}{\localrunlabel{\node_1}}= \exready \cdot \exgotwo = \speclabel{\node_2}$ and $\vinput{3}{\localrunlabel{\node_1}}= \exready \cdot \exgothree = \speclabel{\node_3}$.
	It is also satisfied at $\node_2$, $\node_3$ and $\node_5$ because their local runs only receive $\exready$ and they each have a child with "specification@@sg" $\exready$. 
	It is trivially satisfied at $\node_4$ and $\node_6$ as their "local runs" have no reception. 
\end{example}

\begin{figure}[t]
	\centering
	\input{Figures/tree-example-sign}
	\vspace{-0.5cm}
	\caption{Example of an "unfolding tree@@sg" derived from $\run$. Grids correspond to "local runs", a column of a grid is a "local configuration". Transition $\atrans_{ij}$ is the transition between state $q_i$ and state $q_j$, for example $\atrans_{01} = (q_0, \rec{\exready}{2}{\enregact}, q_1)$. If $\atrans$ is a reception of $m\in \messages$, $\extlabel{\atrans}{v}$ corresponds to receiving message $(m,v)$; if $\atrans$ is a broadcast of $m \in \messages$, $\intlabel{\atrans}$ corresponds to broadcasting $(m,\mathsf{id})$ where $\mathsf{id}$ is the value in the first register of the agent. Initial values of "reception-only" registers are irrevelant and written as $\quotemarks{\_}$. Colors correspond to message types.}
	\label{fig:ex-unfolding-tree-signature}
\end{figure}
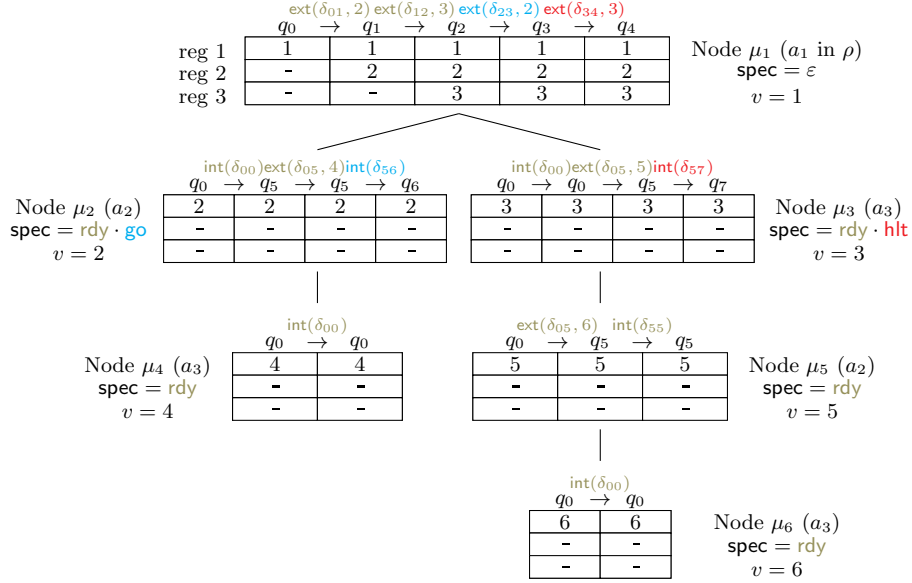

\begin{lemma}
\label{lem:coverability_witness_sg}
Given a "signature protocol" $\prot$ with a state $q_f$, $q_f$ is coverable in $\prot$ if and only if there exists an "unfolding tree@@sg" whose root is labelled by a "local run" covering $q_f$. We call such an "unfolding tree@@sg" a ""coverability witness@@sg"".
\end{lemma}
\begin{proof}
Given a "run" $\run$, agent $a$ \emph{satisfies a "specification@@sg"} $w \in \messages^*$ in $\run$ if the sequence of "message types" broadcast by $a$ admits $w$ as "subword".

Let $\tree$ be a "coverability witness@@sg". 
We prove the following property by strong induction on the depth of $\node$: for every $\node$ in $\tree$, there exists a "run" $\run$ with an agent $a$ whose "local run" in $\run$ is $\localrunlabel{\node}$ and who satisfies "specification@@sg" $\speclabel{\node}$. This is trivially true for leaves of $\tree$ because their "local runs" have no reception (by condition~\ref{item:condition2_non_initial_value_sg}) hence are actual "runs" by themselves. 
Let $\node$ a node of $\tree$, $\localrun := \localrunlabel{\node}$ and $\aval_1, \dots, \aval_c$ the values received in $\localrun$. 
These values are "non-initial" thanks to condition~\ref{item:condition1_initial_value_sg}; applying condition \ref{item:condition2_non_initial_value_sg} gives the existence of corresponding children $\node_1, \dots, \node_c$ in $\tree$. 
We apply the induction hypothesis on the subtrees rooted in $\node_1, \dots, \node_c$ to obtain "runs" $\run_1, \dots, \run_c$ satisfying the "specifications@@sg" of the children of $\node$. 
Up to renaming agents, we can assume the set of agents of these runs are disjoint; up to renaming values, we can assume that $\aval_j = \valuelabel{\node_j}$ for all $j$ and that all agents start with distinct values. 
We build an "initial run" $\run$ whose "agents" is the union of the "agents" of the $c$ runs along with a fresh agent $a$. In $\run$, we make $\run_1$ to $\run_c$ progress in parallel and make $a$ follow the "local run" $\localrun$, matching each reception with value $v_j$ in $\localrun$ with a broadcast in $\run_j$. 
This is possible because, for all $j$, $\vinput{v_j}{\localrun} \subword \speclabel{\node_j} \subword \voutput{v_j}{\run_j}$ (by \ref{item:condition3_boss_node_sg}). 

Conversely, we prove the following by induction on the length of $\run$: for every "initial run" $\run$, for every agent $a$ in $\run$ and for every $v \in \nats$, there exists an "unfolding tree@@sg" whose root has as "local run" the projection of $\run$ onto $a$ and as "specification@@sg" the $v$-"output" of $a$ in $\run$. If $\run$ is the empty run, consider the "unfolding tree@@sg" with a single node whose "local run" and "specification@@sg" are empty. Suppose now that $\run$ has non-zero length, let $a$ an agent in $\run$, $v \in \nats$ and let $\run_p$ the prefix run of $\run$ of length $\size{\run}-1$.
Let $\tree_1$ the "unfolding tree@@sg" obtained by applying the induction hypothesis to $\run_p$, $a$ and $v$, and consider $\tree_2$ obtained by simply appending the last step of $a$ in $\run$ to the "local run" at the root of $\tree_1$. If this last step is a broadcast, we obtain an "unfolding tree@@sg"; if the broadcast value is $v$, we append the broadcast "message type" to the "specification@@sg" at the root of $\tree_2$ and we are done. 
Suppose that, in the last step of $\run$, $a$ performs a reception $(q, \rec{m}{i}{\anact},q')$ of a message $(m,\aval')$. We might need to adapt $\tree_2$ to respect condition \ref{item:condition2_non_initial_value_sg} at the root. Let $a'$ the agent broadcasting in the last step of $\run$. Let $\tree_3$ the "unfolding tree@@sg" obtained by applying the induction to $\run_p$, $a'$ and $v'$. Let $\tree_4$ the "unfolding tree@@sg" obtained by appending the last broadcast to the "local run" at the root of $\tree_3$ and the corresponding "message type" to the "specification@@sg" at the root of $\tree_3$. Attaching $\tree_4$ below the root of $\tree_2$ gives an "unfolding tree@@sg" satisfying the desired properties. 
\end{proof}

	The "unfolding tree@@sg" $\tree$ of Figure~\ref{fig:ex-unfolding-tree-signature} is built from $\run$ of Example~\ref{ex:protocol-signature} using the previous procedure.
	 Observe that the "unfolding tree@@sg" $\tree$  is a "coverability witness@@sg" for $q_4$. However, one can find a smaller "coverability witness@@sg". 
	Indeed, in the right branch of $\tree$, $\node_5$ and $\node_6$ have the same "specification@@sg", therefore $\node_5$ can be deleted and replaced with $\node_6$. More generally, we would have also been able to shorten the tree if we had $\speclabel{\node_5} \subword \speclabel{\node_6}$.
	
	\begin{remark}\label{remark:tree-root-spec-1}
	With the previous notion of "coverability witness", the root has to cover $q_f$ but may have an empty specification. However, we will later need the length of the specification of a node to be equal to the number of tasks that it must carry out. For this reason, we will, in the rest of this paper, consider that the roots of "coverability witnesses" have a specification of length $1$. This can be formally achieved by introducing a new message type $m_f$ that may only be broadcast from $q_f$ and require that, at the root, $\spec = m_f$. 
	\end{remark}

\subsection{Bounding the Size of a Coverability Witness}
In all the following, we fix a positive instance $(\prot,q_f)$ of \COVER\ with $r+1$ registers (\emph{i.e.}, $r$ registers used for reception) and a  "coverability witness@@sg" $\tree$ of minimal size.
We turn the observation above into an argument that will be useful towards bounding the length of branches of a "coverability witness@@sg":

\begin{lemma}
\label{lem:no_subword_in_branch_sg}
If a "coverability witness@@sg" $\tau$ for $(\prot, q_f)$ of minimal size has two nodes $\node, \node'$ with $\node$ a strict ancestor of $\node'$ then  $\speclabel{\node}$ cannot be a subword of $\speclabel{\node'}$. 
\end{lemma}
\begin{proof}
Otherwise, replacing the subtree rooted in $\node$ with the one rooted in $\node'$ would contradict minimality of $\tree$.
\end{proof}

We would now like to use the "Length function theorem" to bound the height of $\tree$, using the previous lemma. To do so, we need a bound on the "size@@treesg" of a node with respect to its depth. The following lemma bounds the number of steps of a "local run" between two local configurations: we argue that if the "local run" is long enough we can replace it with a shorter one that can be executed using the same input. This will in turn bound the "length" of a "local run" of a node with respect to the "size@@sg" of its "specification@@sg", which is the first step towards our goal.

\begin{lemma}
\label{lem:towerbound_signature}
There exists a primitive recursive function $\towerfun$ so that, for every local run $\localrun: (q,\localdata) \step{*} (q', \localdata')$, there exists $u' : (q,\localdata) \step{*} (q',\localdata')$ with $\size{u'} <~\towerfun(\size{\prot},r)$ and for all value $v' \in \nats$, there exists $v \in \nats$ such that  $\vinput{v'}{u'} \subword \vinput{v}{u}$. 
\end{lemma}
\begin{proof}
Let $\towerfun(n,0) = n+1$ and $\towerfun(n,k+1) = 2 \, \towerfun(n,k) \cdot ({\size{\transitions}}^{2\,\towerfun(n,k)}+1)+1$ for all $k$. Observe that $\towerfun(n,k)$ is a tower of exponentials of height $k$, which is primitive-recursive although non-elementary. A register $i \geq 2$ is ""active@@sg"" in a "local run" $u$ if $u$ has some $\quotemarks{\enregact}$ action on register $i$. Let $u$ a "local run", $k$ the number of "active@@sg" registers in $\localrun$, $n := \size{\prot}$ and $M := \towerfun(n,k)$.
We prove by induction on the number $k$ of "active@@sg" registers in $u$  that if $\size{u} \geq \towerfun(n,k)$ then $u$ can be shortened.

If $k=0$, any state repetition can be removed. Suppose that $\size{u} > \towerfun(n,k+1)$ and that the set $I$ of "active@@sg" registers of $u$ is such that $\size{I} = k+1$. If there exists an infix run of $u$ of length $M$ with only $k$ "active@@sg" registers, we shorten $u$ using the induction hypothesis. Otherwise, every sequence of $M$ steps in $u$ has a $\quotemarks{\enregact}$ on every register of $I$. Because $\size{\localrun} > 2 M \, (\size{\transitions}^{2M} +1)$, $\localrun$ contains at least $\size{\transitions}^{2M}+1$ disjoint sequences of length $2M$ and some $s \in \transitions^{2M}$ appears twice: in infix run $u_1$ first, then in infix run $u_2$. We build a shorter run $u'$ by removing all steps between $u_1$ and $u_2$ and merging $u_1$ and $u_2$ (see \cref{fig:proof-pumping-signed}). 
We need suitable values for the reception steps in $s$ in the shortened run $u'$. 
For a given register $i \in I$, we would like to pick a $\quotemarks{\enregact}$ step on register $i$ in $s$, use values from $u_1$ before that step and values from $u_2$ after that step. This would guarantee that all equality and disequality tests still pass. 
However, there is an issue if a value $v$ appears in several registers in $u$. For example, if $v_1 = v_2 = v$ in Figure~\ref{fig:proof-pumping-signed}, we might interleave receptions of $v$ on registers $2$ and $4$: if we had a $\extlabel{\rec{{m_1}}{2}{\eqtestact}}{v}$ in $u_1$ and a $\extlabel{\rec{{m_2}}{4}{\eqtestact}}{v}$ in $u_2$, we could have ${m_1}$ before ${m_2}$ in $\vinput{v}{u}$ but ${m_1}$ after ${m_2}$ in $\vinput{v}{u'}$, so that we do not have $\vinput{v}{u'} \subword \vinput{v}{u}$. We solve this issue by introducing fresh values between values of $u_1$ and values of $u_2$; because $\size{s} = 2M$, there is a $\quotemarks{\enregact}$ for each register in $I$ in each half of $s$. In the shortened run $u'$, before the \emph{first} $\quotemarks{\enregact}$ on register $i$ (excluded), we use values of $u_1$, and after the \emph{last} $\quotemarks{\enregact}$ on register $i$ (included), we use values of $u_2$. For every value $v$ appearing in register $i$ between these two steps in $u_1$, we select a fresh value $v_f$ (\emph{i.e.}, a value that does not appear anywhere in the run) and consistently replace $v$ with $v_f$ (hatched blocks in \cref{fig:proof-pumping-signed}).
With this technique, receptions with values from $u_1$ and receptions with values from $u_2$ cannot get interleaved in $u'$. Therefore, for every value that appeared in $u$, we have $\vinput{v}{u'} \subword \vinput{v}{u}$. Also, for every fresh value $v'$ there is a value $v$ such that $\vinput{v'}{u'} \subword \vinput{v}{u}$. Moreover, $u'$ is shorter than $u$; we conclude by iterating this shortening procedure. 
\end{proof}

\begin{figure}[t]

	\centering
	\input{Figures/fig-illustration-signed-pumping-v2}
	\caption{Illustration of the proof of \cref{lem:towerbound_signature}. 
	}
	\label{fig:proof-pumping-signed}
\end{figure}

Using the previous lemma, we will bound the size of a node in $\tree$ with respect to its "specification@@sg" therefore with respect to its parent's size. By induction, we will then obtain a bound depending on the depth, and apply the "Length function theorem" to bound the height of the tree. 

\begin{lemma}
\label{lem:bounds_tree_sg}
For all nodes $\node, \node'$ in $\tree$: \begin{enumerate}
\item \label{item:bound_node_1_sg} $\size{\localrunlabel{\node}} \leq \towerfun(\size{\prot}, r) \, \size{\speclabel{\node}}$,
\item \label{item:bound_node_2_sg} if $\node$ is the child of $\node'$, $\size{\speclabel{\node}} \leq \towerfun(\size{\prot}, r) \, \size{\speclabel{\node'}}$.
\end{enumerate} 
\end{lemma}
\begin{proof}
Thanks to \cref{remark:tree-root-spec-1}, we assume that the specification at the root is of length $1$.	
For the first item, by minimality of $\tree$, $\localrunlabel{\node}$ ends with the last broadcast required by $\speclabel{\node}$; we identify in $\localrunlabel{\node}$ the broadcast steps witnessing $\speclabel{\node}$ and shorten the "local run" between these steps using Lemma~\ref{lem:towerbound_signature}.  We thus obtain
 $\size{\localrunlabel{\node}} \leq \towerfun(\size{\prot},r) \, \size{\speclabel{\node}}$, proving \ref{item:bound_node_1_sg}.
For the second item, by minimality of $\tree$, $\size{\speclabel{\node}} \leq \max_{v \in \nats} \size{\vinput{v}{\localrunlabel{\node'}}} \leq \size{\localrunlabel{\node'}} \leq \towerfun(\size{\prot}, r) \, \size{\speclabel{\node'}}$. 
\end{proof}

\begin{proposition}
\label{prop:bounded_witness_sg}
There exists a function $f$ of class $\Ffunction{\omega^{\size{\messages}-1}}$ s.t. $\size{\tree} \leq f(\size{\prot})$. 
\end{proposition}
\begin{proof}
Let $n := \size{\prot}$, let $r+1$ be the number of registers in $\prot$. Thanks to Lemma~\ref{lem:no_subword_in_branch_sg}, for all $\node \ne \node'$ in 
$\tree$ with $\node$ ancestor of $\node'$, $\speclabel{\node}$ is not a subword of $\speclabel{\node'}$.  Let $\node_1, \dots, \node_m$ the node appearing in a branch of $\tree$, from root to leaf. The sequence $\speclabel{\node_1}, \dots, \speclabel{\node_m}$ is a "bad sequence".
For all $i \in \nset{1}{m}$, $\size{\speclabel{\node_{i+1}}}  \leq \towerfun(n,r) \, \size{\speclabel{\node_i}}$ by Lemma~\ref{lem:bounds_tree_sg}. By direct induction, $\size{\speclabel{\node_i}}$ is bounded by $g^{(i)}(n)$ where $g: n \mapsto n  \, \towerfun(n,n)$ is a primitive recursive function. Let $h$ of class $\Ffunction{\omega^{\size{\messages}-1}}$ the function obtained when applying the "Length function theorem" on $g$ and $\messages$; we have $m \leq h(n)$. 

By immediate induction, thanks to Lemma~\ref{lem:bounds_tree_sg}.\ref{item:bound_node_2_sg}, for every node $\node$ at depth $d$, $\size{\speclabel{\node}} \leq \towerfun(n,r)^{d+1}$ which, by Lemma~\ref{lem:bounds_tree_sg}.\ref{item:bound_node_1_sg} and because $d \leq h(n)$, bounds the size of every node by $h'(n) = \towerfun(n,n)^{h(n) + 2}$. 
By minimality of $\tree$, the number of children of a node is bounded by the number of values appearing in its "local run" hence by $h'(n)$, so the total number of nodes in $\tree$ is bounded by $h'(n)^{h(n)+1}$ and the "size@@treesg" of $\tree$ by $f(n) := h'(n)^{h(n)+2}$. Because $\Ffunction{\omega^{\size{\messages}-1}}$ is closed under composition with primitive-recursive functions, $f$ is in $\Ffunction{\omega^{\size{\messages}-1}}$.
\end{proof}

The previous argument shows that \COVER\ for "signature protocols" is decidable and lies in complexity class $\Fcomplexity{\omega^\omega}$. Because the hardness from Proposition~\ref{prop:reduction-LCS} holds for "signature protocols", \COVER\ is in fact complete for this complexity class.

We now extend this method to the general case. 


\section{Coverability Decidability in the General Case}
\label{sec:cover-general-case}

\subsection{Generalizing Unfolding Trees}
\label{sec:unfolding-trees-general}
In the general case, a new phenomenon appears: an agent may broadcast a value that it did not initially have but that it has received and stored. In particular, an agent starting with value $v$ could broadcast $v$ then require someone else to make a broadcast with value $v$ as well. For example, in the run described in \cref{ex:example-1}, $1$ is initially a value of $a_1$ that $a_2$ receives and rebroadcasts to $a_1$.

 We now have two types of specifications. ""Boss specifications"" describe the task of broadcasting with one of its own initial values; this is the "specification@@sg" we had in "signature protocols" and, as before, it consists of a word $\bossspec \in \messages^\ast$ describing a sequence of "message types" that should be all broadcast with the same value. ""Follower specifications"" describe the task of broadcasting with a non-initial value received previously. More precisely, a "follower specification" is a pair $(\followwordspec, \followmessagespec) \in \messages^*\times \messages$ asking to broadcast a message $(\followmessagespec,v)$ under the condition of previously receiving the sequence of "message types" $\followwordspec$ with value $v$.


A key idea is that, if an agent that had $v$ initially receives some message $(m,v)$, then intuitively we can isolate a subset of agents that did not have $v$ initially but that are able to broadcast $(m,v)$ after receiving a sequence of messages with that value. We can then copy them many times in the spirit of the "copycat principle". Each copy receives the necessary sequence of messages in parallel, and they then provide us with an unbounded supply of messages $(m,v)$. In short, if an agent broadcasts $(m,v)$ while not having $v$ as an "initial value", then we can consider that we have an unlimited supply of messages $(m,v)$.

\begin{example}

\label{ex:decomposition}
Assume that $\agents = \set{a_1,a_2,a_3}$ and let $v$ be initial for $a_1$. Consider an execution where the broadcasts with value $v$ are: $a_1$ broadcasts $\textsf{a} \cdot \textsf{b}$, then $a_2$ broadcasts $\textsf{c}$, then $a_1$ broadcasts $\textsf{a}^3$ then $a_3$ broadcasts $\textsf{b}$. The "follower specification" of $a_2$'s task would be of the form $(w, \textsf{c})$ where $w \subword \textsf{a} \cdot \textsf{b}$: $a_2$ must be able to broadcast $(c,v)$ once $\textsf{a} \cdot \textsf{b}$ has been broadcast with value $v$. By contrast, $a_3$'s "follower specification" would be of the form $(w \cdot w', \textsf{c})$ where $w \subword \textsf{a} \cdot \textsf{b}$ and $w' \in \set{\textsf{a},\textsf{c}}^*$ is a subword of $\textsf{a}^3$ enriched with as many $\textsf{c}$ as desired, because $a_2$ may be cloned at will.
For example, one could have $w= \textsf{b} $ and $w' = \textsf{c} \cdot \textsf{a} \cdot \textsf{c}^4 \cdot \textsf{a} \cdot \textsf{c}^2$. This idea is formalized in Appendix~\ref{app:def-trees} with the notion of \emph{decomposition}. Using this notion, the previous condition becomes: $w \cdot w'$ \emph{admits decomposition} $(\textsf{a} \cdot \textsf{b}, \textsf{c}, \textsf{a}^3)$.   
\end{example}

\AP In our new ""unfolding trees"", a node is either a \emph{boss node} or a \emph{follower node}, depending on its type of specification. 
A ""boss node"" with a "boss specification" $\bossspec$ must broadcast that sequence of "message types" with one of its "initial values". 
A ""follower node"" $\node$ with "follower specification" $(\followwordspec, \followmessagespec)$ is allowed to receive sequence of messages $\followwordspec$ with value $\valuelabel{\node}$ (which must be "non-initial") without it being broadcast by its children. 
Other conditions are similar to the ones for "signature protocols": if $\node$ is a node and $v \ne \valuelabel{\node}$ a "non-initial" value received in its "local run", $\node$ must have a "boss" child broadcasting this word. Moreover, for each $(m,v)$ received where $v$ is an "initial value" of the "local run", $\node$ must have a "follower" child that is able to broadcast $(m,v)$ after receiving messages sent previously with value $v$; the formal statement is more technical because it takes into account the observation of Example~\ref{ex:decomposition}. 
The formal definition of \emph{unfolding tree} is given in Appendix~\ref{app:def-trees}.

\begin{example}
	Figure~\ref{fig:ex-unfolding-tree} depicts the "unfolding tree" associated to $a_1$ in the "run" of \cref{ex:example-1}. 
	"Follower node" $\node_3$ can have a $m_2$ reception that is not matched by its children because $m_2$ is in $\followwordspec(\node_3)$. $\node_1$ broadcasts $(m_2,1)$ before receiving $(m_4,1)$ hence the "follower specification" of $\node_3$ witnesses broadcast of $(m_4,1)$. 
\end{example}
\begin{figure}[t]
	\begin{center}
					\input{Figures/tree-example1}
		\end{center}
	\vspace*{-0.5cm}
	\caption{Example of an "unfolding tree". $\delta_{ri}$ (resp. $\delta_{bi}$) denotes the reception (resp. broadcast) transition of message $m_i$ in the protocol described in \cref{fig:ex1}. Values that are never broadcast are omitted and written as $\quotemarks{\_}$.}\label{fig:ex-unfolding-tree}
\end{figure}
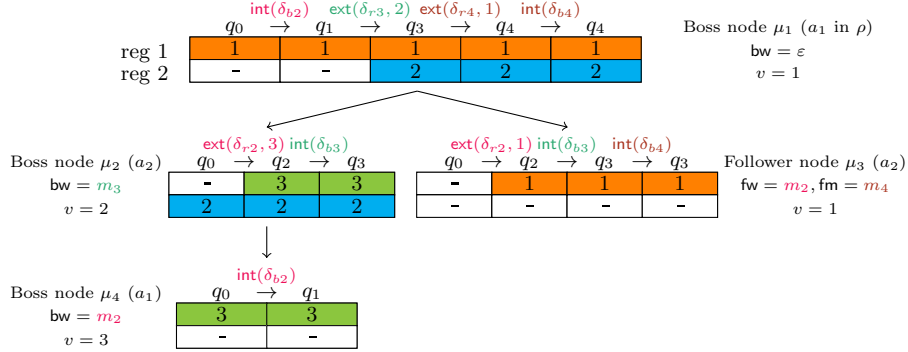

A ""coverability witness"" is again an "unfolding tree" whose root covers $q_f$ (or broadcasts a message $m_f$, see \cref{remark:tree-root-spec-1}), with the extra condition that the root is a "boss node" (a "follower node" implicitly relies on its parent's ability to broadcast). 

\begin{restatable}{proposition}{treessoundcomplete}
	\label{prop:trees-sound-complete}
	An instance of \COVER\ $(\prot,q_f)$ is positive if and only if there exists a "coverability witness" for that instance.
\end{restatable}
\begin{proof}[Proof sketch]
	The proof is quite similar to the one of Lemma~\ref{lem:coverability_witness_sg}, but is made more technical by the addition of "follower" nodes. 
	When translating an "unfolding tree" to a "run", if the root of the tree is a "follower node" $\node$ of specification $(\followwordspec, \followmessagespec)$, then we actually obtain a \emph{partial run}, \emph{i.e.}, a "run" except that the receptions from $\followwordspec$ are not matched by broadcasts in the "run". We then combine this partial run with the run corresponding to the parent of $\node$ and with the runs of other children of $\node$ so that every reception is matched with a broadcast.
	For the translation from "run" to tree, we inductively construct the tree by extracting from the run the agents and values responsible for satisfying the specifications of each node and analyzing the messages they receive to determine their set of children (as in Example~\ref{ex:decomposition}). 
	See Appendix~\ref{app:trees-sound-complete} for the proof.
\end{proof}

\subsubsection{Bounding the Size of the Unfolding Tree.}
\label{sec:tree-bounds}

Our aim is again to bound the size of a minimal "coverability witness". In the following, we fix an instance $(\prot,q_f)$ with $r$ registers and a "coverability witness" of minimal size. We start by providing new conditions under which a branch can be shortened; for "boss specifications", it is the condition of Lemma~\ref{lem:no_subword_in_branch_sg} but for "follower specifications", the subword relation goes the opposite direction because the shorter the requirement $\followwordspec$, the better.
Details can be found in Appendix~\ref{app:proofs-reduction-branches}.

\begin{restatable}{lemma}{lemShorteningBranches} 
\label{lem:shortening-branches}
	Let $\node \ne \node'$ be two nodes of $\tree$ such that $\node$ is an ancestor of $\node'$. If one of those conditions holds, then $\tree$ can be shortened  (contradicting its minimality):
	\begin{itemize}
	\item $\node$ and $\node'$ are "boss nodes" with "boss specifications" respectively $\bossspec$ and $\bossspec'$, and $\bossspec \subword \bossspec'$; 
	\item $\node$ and $\node'$ are "follower nodes" with "follower specifications" respectively $(\followwordspec, \followmessagespec)$ and $(\followwordspec', \followmessagespec')$, and $\followwordspec' \subword \followwordspec$ and $\followmessagespec'=\followmessagespec$.
	\end{itemize} 
\end{restatable}

We can generalize \cref{lem:towerbound_signature} to bound the size of a node by the number of messages that it must broadcast times a primitive-recursive function $\towerfun(\size{\prot},r)$. The proof is more technical than the one of \cref{lem:towerbound_signature} but the idea is essentially the same. The formal statement is given below.
, and the proof can be found in Appendix~\ref{app:tower-lemma}. 
One can therefore bound the size of a node with respect to the size of the nodes that it must broadcast to.

	\begin{restatable}{lemma}{lemShortLocalRuns}
	\label{lem:short-local-runs}
	There exists a primitive recursive function $\towerfun$ such that, for every protocol $\prot$ with $r$ registers, 
	for all "local runs" $\localrun_0: (q_0, \localdata_0) \step{*} (q, \localdata)$, $\localrun: (q, \localdata) \step{*} (q', \localdata')$, $\localrun_f: (q', \localdata') \step{*} (q_f, \localdata_f)$, there exists a "local run" $\localrun': (q, \localdata) \step{*} (q', \localdata')$ with $\size{\localrun'} \leq \towerfun(\size{\prot} ,r)$ and for all $\aval' \in \nats$: 
	\begin{enumerate}
		\item if $\aval'$ appears in $\localrun_0$, $\localrun$, or $\localrun_{f}$, $\vinput{\aval'}{\localrun'} \subword \vinput{\aval'}{\localrun}$,
		\item  otherwise, there exists $\aval \in \nats$, not "initial" in $\localrun_0$, such that $\vinput{\aval'}{\localrun'} \subword \vinput{\aval}{\localrun}$.
	\end{enumerate}
\end{restatable}

It is however now much harder than in the "signature" case to bound the size of the "coverability witness". Indeed, the broadcasts no longer go only from children to parents in the "unfolding tree". If $\node_p$ is the parent of $\node_c$, then $\node_c$ broadcasts to $\node_p$ if $\node_c$ is a "boss node", but $\node_p$ broadcasts to $\node_c$ if $\node_c$ is a "follower node", in which case $\node_c$ only broadcasts one message to $\node_p$. Therefore, we cannot in general bound $\size{\node_p}$ with respect to $\size{\node_c}$ nor $\size{\node_c}$ with respect to $\size{\node_p}$, making us unable to apply the "Length function theorem" immediately. 

This leads us to arrange the "unfolding tree" so that long broadcast sequences are sent upwards, using the notion of "altitude" depicted in Figure~\ref{fig:rearrange-tree}, formally defined as follows.
The ""altitude"" of the root is $0$, the altitude of a "boss node" is the altitude of its parent minus one, and the altitude of a "follower node" is the altitude of its parent plus one.
We denote the "altitude" of $\node$ by $\altitude{\node}$.
This way the nodes of maximal "altitude" are the ones that do not need to send long sequences of messages. We will bound the size of nodes with respect to their "altitude", from the highest to the lowest, and then use the "Length function theorem" to bound the maximal and minimal "altitudes". We present here a sketch of the proof.
; details are postponed to Appendix~\ref{app:proofs_bounds}.

\begin{figure}[t]
	\input{Figures/rearrangement-tree}
	\caption{Rearrangement of a tree. The root is in red, black solid arrows connect parents to children, blue dashed arrows highlight that long words of messages are sent upwards.}
	\label{fig:rearrange-tree}
\end{figure}
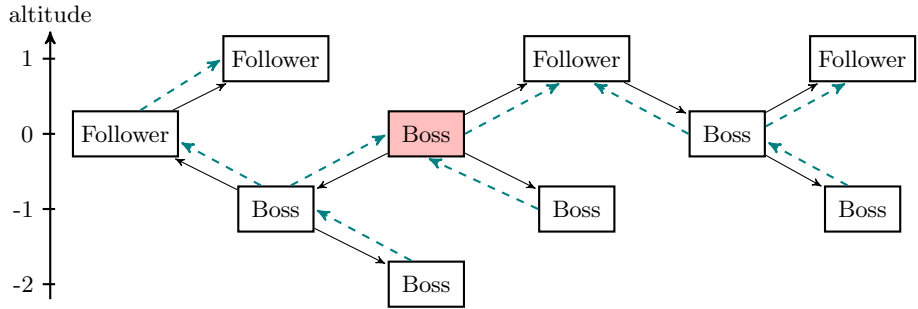

Let $\altmax \geq 0$ (resp. $\altmin \leq 0$) denote the maximum (resp. minimum) altitude in $\tree$.
We first bound the size of a node with respect to the difference between its altitude and $\altmax$.

\begin{restatable}{lemma}{lemBoundLengthHeightH}
	\label{lem:bound-length-at-height-h}
	There is a primitive recursive function $f_0$ such that, for every node $\node$ of $\tree$, $\size{\node} \leq f_0(\size{\prot} + \altmax - \altitude{\node})$.
\end{restatable}
\begin{proof}[Proof sketch]
	We proceed by induction on the altitude, from highest to lowest. 
	A node of maximal "altitude"  has at most one message to broadcast (a "follower node" must broadcast one message to its parent), so its size is bounded by $\towerfun(\size{\prot},r)$ by \cref{lem:short-local-runs} (applying the Lemma to its local run minus its final step, \emph{i.e.}, the step making the broadcast to its parent).
	Let $\node$ be a node of $\tree$ whose neighbors of higher altitude have size bounded by $K$.
	We claim that $\size{\node} \leq (\towerfun(\size{\prot}, r)+2) \, (\size{\messages} \, r \, K + K)$, with $\psi$ the primitive-recursive function defined in Lemma~\ref{lem:short-local-runs}. 
	The idea is similar to the one for Lemma~\ref{lem:bounds_tree_sg}. The neighbors of higher "altitude" are the nodes which require sequences of messages from $\node$. Their size bounds the number of messages that $\node$ needs to send; we then apply Lemma~\ref{lem:short-local-runs} to bound the size of the local run of $\node$.
	Lemma~\ref{lem:bound-successor-height} in the appendix details these ideas.
We finally obtain $f_0$ by iteratively applying the inequality above.
\end{proof}

We now bound $\altmax$ and $\altmin$:
\begin{lemma}
$\altmax$ and $|\altmin|$ are bounded by a function of class $\Ffunction{\omega^{\size{\messages}}}$. 
\end{lemma}
\begin{proof}[Proof sketch]
We first bound $\altmax$. Consider a branch of $\tree$ that has a node at "altitude" $\altmax$. We follow this branch from the root to a node of altitude $\altmax$: for every $j \in \nset{1}{\altmax}$, let $\node_{j}$ be the first node of the branch that has altitude $j$. All such nodes are necessarily "follower nodes" as they are above their parent. Sequence $\node_{\altmax}, \dots, \node_2, \node_1$ is so that the $i$th term is at altitude $\altmax-i$ hence its size is bounded by $f_0(\size{\prot} + i)$ (Lemma~\ref{lem:bound-length-at-height-h}). With the observation of Lemma~\ref{lem:shortening-branches}, we retrieve from the "follower specifications" of this sequence of nodes a "bad sequence" and we apply the "Length function theorem" to bound $\altmax$.
This yields in turn a bound on the size of the root of $\tree$. In order to bound $\altmin$, we proceed similarly, using "boss nodes" this time. We follow a branch from the root to a node of "altitude" $\altmin$. The sequence of nodes that are lower than all previous ones yields a sequence of "boss specifications", which is a "bad sequence" by Lemma~\ref{lem:shortening-branches}, and whose growth can be bounded using Lemma~\ref{lem:bound-length-at-height-h} and the bound on $\altmax$. We apply the "Length function theorem" to bound $|\altmin|$.
\end{proof}

Once we have bounded $\altmax$ and $\altmin$, we can infer a bound on the size of all nodes (Lemma~\ref{lem:bound-length-at-height-h}), and then on the length of branches: by minimality, a branch cannot have two nodes with the same specification. 
The bound on the size of the tree then follows from the observation that bounding the size of nodes of $\tree$ also allows to bound their number of children.



We obtain a computable bound (of the class $\Ffunction{\omega^\omega}$) on the size of a minimal "coverability witness" if it exists. 
Our decidability procedure computes that bound, enumerates all trees of size below the bound and checks for each of them whether it is "coverability witness". This yields the main result of this paper, whose proof can be found in Appendix~\ref{app:decidable_cover}:

\decidablecover*

%% file: Figures/fig-ex2.tex
\begin{tikzpicture}[xscale=0.5,AUT style,node distance = 2cm,auto,>= triangle 45]
	\tikzstyle{initial}= [initial by arrow,initial text=,initial
	distance=.7cm]

	\node[state, initial] (0) {$q_0$};
	\node[state] [above = 1 of 0] (1) {$q_1$};
	\node[state] [right = 2.1 of 1] (2) {$q_2$};
	\node[state] [right = 2 of 2] (3) {$q_3$};
	\node[state] [right = 2 of 3] (4) {$q_4$};
	\node[state] [right = 2.1 of 0] (5) {$q_5$};
	\node[state] [right = 2 of 5, yshift = 20] (6) {$q_6$};
	\node[state] [right = 2 of 5, yshift = -20] (7) {$q_7$};
	

	\path[->] 	
			
	(5) edge node[above] {\small$\br{\exgotwo}{1}$} (6)
	(5) edge node[below, yshift = -2] {\small$\br{\exgothree}{1}$} (7)
	(5) edge [out = -105, in = -75, looseness = 4] node[below] {\small$\br{\exready}{1}$} (5)
	(0) edge [out = -105, in = -75, looseness = 4] node[below] {\small $\br{\exready}{1}$} (0)
	(0) edge node[right] {\small $\rec{\exready}{2}{\enregact}$} (1)
	(1) edge node {\small $\rec{\exready}{3}{\enregact}$} (2)
	(2) edge node {\small $\rec{\exgotwo}{2}{\eqtestact}$} (3)
	(3) edge node {\small $\rec{\exgothree}{3}{\eqtestact}$} (4)
	(0) edge node [below]{\small $\rec{\exready}{2}{\dummyact}$} (5)

	;
	
\end{tikzpicture}

%% file: Figures/tree-example-sign.tex
\newcommand{\tikzabstractconfig}[5]{
	\begin{tikzpicture}
		\draw (#1,#2) rectangle (1 + #1,0.5 +#2);
		\draw (#1,-0.5 + #2) rectangle (1 +#1 ,#2);
		\node[] (x1) at (0.5 + #1, 0.15 +#2) {#3};
		\node[] (x2) at (0.5 +#1 , 0.15- 0.5 +#2) {#4};
		\node [] (q) at (0.5 +#1, 0.7+#2) {#5};
	\end{tikzpicture}
}

\newcommand{\tikzvalue}[4]{
	\begin{tikzpicture}
		\node [] (v) at (#1 + #3 + 0.3, #2 + 1.2) {$\mathbf{v}= #4$};
	\end{tikzpicture}
}

\newcommand{\tikzspec}[4]{
	\begin{tikzpicture}
		\node [] (sp) at (#3 /2, #2 - 0.2) {#4};
	\end{tikzpicture}
}

\newcommand{\displayone}[1]{{\color{Yellow!50!black} #1}}
\newcommand{\displaytwo}[1]{{\color{cyan} #1}}
\newcommand{\displaythree}[1]{{\color{Red} #1}}
\begin{tikzpicture} [xscale = 0.94, yscale = 0.6, every node/.style = {scale = 0.9}]

	\node [font = {\footnotesize}] (sp) at (6.5, 0.2-0.5) {Node $\node_1$ ($a_1$ in $\run$)};
	\node [font = {\footnotesize}] (sp) at (6.5, 0.2-1) {$\spec = \epsilon$};
	\node [font = {\footnotesize}] (sp) at (6.5, 0.2-1.5) {$v = 1$};
	\begin{scope}[xscale = 1.2, xshift = -10]
	\regbar{-1}{0}
	\onerow{-0.5}{0}{$q_0$}{$1$}{white}{$\_$}{white}{$\_$}{white}
	\transtable{0.5}{0}{$\displayone{\extlabel{\atrans_{01}}{2}}$}
	\onerow{0.5}{0}{$q_1$}{$1$}{white}{$2$}{white}{$\_$}{white}
	\transtable{1.5}{0}{$\displayone{\extlabel{\atrans_{12}}{3}}$}
	\onerow{1.5}{0}{$q_2$}{$1$}{white}{$2$}{white}{$3$}{white}
	\transtable{2.5}{0}{$\displaytwo{\extlabel{\atrans_{23}}{2}}$}
	\onerow{2.5}{0}{$q_3$}{$1$}{white}{$2$}{white}{$3$}{white}
	\transtable{3.5}{0}{$\displaythree{\extlabel{\atrans_{34}}{3}}$}
	\onerow{3.5}{0}{$q_4$}{$1$}{white}{$2$}{white}{$3$}{white}
	\end{scope}
	\draw (2, -1.7) -- (0, -2.5);
	\draw (2, -1.7) -- (4, -2.5);
	
	\begin{scope}[xshift = -5]
	\onerow{-2}{-3.5}{$q_0$}{$2$}{white}{$\_$}{white}{$\_$}{white}
	\transtable{-1}{-3.5}{$\displayone{\intlabel{\atrans_{00}}}$}
	\onerow{-1}{-3.5}{$q_5$}{$2$}{white}{$\_$}{white}{$\_$}{white}
	\transtable{0}{-3.5}{$\displayone{\extlabel{\atrans_{05}}{4}}$}
	\onerow{-0}{-3.5}{$q_5$}{$2$}{white}{$\_$}{white}{$\_$}{white}
	\transtable{1}{-3.5}{\displaytwo{$\intlabel{\atrans_{56}}$}}
	\onerow{1}{-3.5}{$q_6$}{$2$}{white}{$\_$}{white}{$\_$}{white}
	\node [align=center, font = {\footnotesize}] (sp) at (-3.2, -3-0.8) {Node $\node_2$ ($a_2$)};
	\node [align = center, font = {\footnotesize}] (sp) at (-3.2,-3.5-0.8) {$\spec = \displayone{\exready} \cdot \displaytwo{\exgotwo}$};
	\node [align = center, font = {\footnotesize}] (sp) at (-3.2, -4-0.8) {$v = 2$};
	\end{scope}
	\begin{scope}[xshift = 5]
	\onerow{2}{-3.5}{$q_0$}{$3$}{white}{$\_$}{white}{$\_$}{white}
	\transtable{3}{-3.5}{\displayone{$\intlabel{\atrans_{00}}$}}
	\onerow{3}{-3.5}{$q_0$}{$3$}{white}{$\_$}{white}{$\_$}{white}
	\transtable{4}{-3.5}{$\displayone{\extlabel{\atrans_{05}}{5}}$}
	\onerow{4}{-3.5}{$q_5$}{$3$}{white}{$\_$}{white}{$\_$}{white}
	\transtable{5}{-3.5}{\displaythree{$\intlabel{\atrans_{57}}$}}
	\onerow{5}{-3.5}{$q_7$}{$3$}{white}{$\_$}{white}{$\_$}{white}
	\node [align=center, font = {\footnotesize}] (sp) at (7.2, -3-0.8) {Node $\node_3$ ($a_3$)};
	\node [align = center, font = {\footnotesize}] (sp) at (7.2,-3.5-0.8) {$\spec = \displayone{\exready} \cdot \displaythree{\exgothree}$};
	\node [align = center, font = {\footnotesize}] (sp) at (7.2, -4-0.8) {$v = 3$};
	\end{scope}
\draw (0, -3.5-1.7) -- (0, -3.5-1.7-0.7);
\draw (4, -3.5-1.7) -- (4, -3.5-1.7-0.7);
\begin{scope}[xscale = 1.2, xshift = 0]
	\onerow{-1}{-7}{$q_0$}{$4$}{white}{$\_$}{white}{$\_$}{white}
	\transtable{0}{-7}{\displayone{$\intlabel{\atrans_{00}}{}$}}
	\onerow{-0}{-7}{$q_0$}{$4$}{white}{$\_$}{white}{$\_$}{white}
	\node [align=center, font = {\footnotesize}] (sp) at (-2, -3.5-0.8-3) {Node $\node_4$ ($a_3$)};
	\node [align = center, font = {\footnotesize}] (sp) at (-2,-4-0.8-3) {$\spec = {\displayone{\exready}}$};
	\node [align = center, font = {\footnotesize}] (sp) at (-2, -4.5-0.8-3) {$v = 4$};
	\end{scope}
	\begin{scope}[xscale = 1.2, xshift = -19]
	\onerow{2.5}{-7}{$q_0$}{$5$}{white}{$\_$}{white}{$\_$}{white}
	\transtable{3.5}{-7}{$\displayone{\extlabel{\atrans_{05}}{6}}$}
	\onerow{3.5}{-7}{$q_5$}{$5$}{white}{$\_$}{white}{$\_$}{white}
	\transtable{4.5}{-7}{\displayone{$\intlabel{\atrans_{55}}$}}
	\onerow{4.5}{-7}{$q_5$}{$5$}{white}{$\_$}{white}{$\_$}{white}
	\node [align=center, font = {\footnotesize}] (sp) at (6.5, -3.5-0.8-3) {Node $\node_5$ ($a_2$)};
	\node [align = center, font = {\footnotesize}] (sp) at (6.5,-4-0.8-3) {$\spec = {\displayone{\exready}}$};
	\node [align = center, font = {\footnotesize}] (sp) at (6.5, -4.5-0.8-3) {$v = 5$};
\end{scope}
\draw (4, -7-1.7) -- (4, -7-1.7-0.7);

	\onerow{3}{-10.5}{$q_0$}{$6$}{white}{$\_$}{white}{$\_$}{white}
	\transtable{4}{-10.5}{\displayone{$\intlabel{\atrans_{00}}$}}
	\onerow{4}{-10.5}{$q_0$}{$6$}{white}{$\_$}{white}{$\_$}{white}
	\node [align=center, font = {\footnotesize}] (sp) at (6.5, -10-0.8) {Node $\node_6$ ($a_3$)};
	\node [align = center, font = {\footnotesize}] (sp) at (6.5,-10.5-0.8) {$\spec = {\displayone{\exready}}$};
	\node [align = center, font = {\footnotesize}] (sp) at (6.5, -11-0.8) {$v = 6$};

\end{tikzpicture}

%% file: Figures/fig-illustration-signed-pumping-v2.tex
\usetikzlibrary{calc,shapes.arrows,decorations.pathreplacing, patterns}

\begin{tikzpicture}[yscale = 0.8]

	\node at (3.8,1) [align = center] {Original \\ local run};
	

	\node at (5,1.2) {reg $2$};
	\node at (5,0.7) {reg $3$};
	\node at (5,0.2) {reg $4$};
	\draw (5.5,0) rectangle (10.9,0.5);
	\draw (5.5,0.5) rectangle (10.3,1);
	\draw (5.5,1) rectangle (10.3,1.5);

	\draw[white,fill=blue!20] (7,0) rectangle (10,1.5);
	\draw[white,fill=green!20] (11.5,0) rectangle (14.5,1.5);
	
	\draw[white,fill=blue!30] (7,0) rectangle (7.5,0.5);
	\draw[white,fill=blue!40] (7.5,0) rectangle (8.8,0.5);
	\draw[white,fill=blue!45] (8.8,0) rectangle (10,0.5);
	\draw (5.5,0) rectangle (6.5,0.5);
	\draw (5.5,0) rectangle (7.5,0.5);
	\draw (5.5,0) rectangle (8.8,0.5);
	
	\draw[white,fill=green!20] (11.5,0) rectangle (12,0.5);
	\draw[white,fill=green!40] (12,0) rectangle (13.3,0.5);
	\node at (12.65,0.2) {$v_2$};
	\draw[white,fill=green!57] (13.3,0) rectangle (14.5,0.5);
	\draw (5.5,0) rectangle (12,0.5);
	\draw (5.5,0) rectangle (13.3,0.5);
	
	\draw[white,fill=blue!50] (7,0.5) rectangle (7.7,1);
	\draw[white,fill=blue!20] (7.7,0.5) rectangle (9.4,1);
	\draw[white,fill=blue!35] (9.4,0.5) rectangle (10,1);
	\draw (5.5,0.5) rectangle (5.8,1);
	\draw (5.8,0.5) rectangle (6.8,1);
	\draw (5.5,0.5) rectangle (7.7,1);
	\draw (5.5,0.5) rectangle (9.4,1);
	
	\draw[white,fill=green!70] (11.5,0.5) rectangle (12.2,1);
	\draw[white,fill=green!20] (12.2,0.5) rectangle (13.9,1);
	\draw[white,fill=green!40] (13.9,0.5) rectangle (14.5,1);
	\draw (5.5,0.5) rectangle (12.2,1);
	\draw (5.5,0.5) rectangle (13.9,1);
	
	\draw[white,fill=blue!35] (7,1) rectangle (8.1,1.5);
	\draw[white,fill=blue!55] (8.1,1) rectangle (9.5,1.5);
	\draw[white,fill=blue!25] (9.5,1) rectangle (10,1.5);
	\draw (5.5,1) rectangle (6.7,1.5);
	\draw (5.5,1) rectangle (8.1,1.5);
	\draw (5.5,1) rectangle (9.5,1.5);
	
	\draw[white,fill=green!37] (11.5,1) rectangle (12.6,1.5);
	\draw[white,fill=green!60] (12.6,1) rectangle (14,1.5);
	\draw[white,fill=green!25] (14,1) rectangle (14.5,1.5);
	\draw (5.5,1) rectangle (12.6,1.5);
	\draw (5.5,1) rectangle (14,1.5);
	\draw (5.5,1) rectangle (14.7,1.5);
	\node at (7.4,1.2) {$v_1$};

	\node at (7.3,2) [align = center] {$\quotemarks{\enregact}$ actions};
	\draw[->] (6.7,1.8) -- (6.7,1.6);
	\draw[->] (8.1,1.8) -- (8.1,1.6);

	%
	%
	
	\begin{scope}[xshift = 2cm, yshift = 0.5cm]
		
		\node at (3.5, -2) [align = center] {Shortened \\ local run};
		\node at (5,-1.8) {reg $2$};
		\node at (5,-2.3) {reg $3$};
		\node at (5,-2.8) {reg $4$};
		\draw[<->] (7,-1.3) -- (10,-1.3);
		\node at (8.5, -1.1) {$s$};
		\draw[white,fill=blue!30] (7,-3) rectangle (7.5,-2.5);
		\draw[white,pattern=north west lines, pattern color=brown!80!black] (7.5,-3) rectangle (8.8,-2.5);
		
		\node at (8.15,-2.75) {$v_2'$};
		\draw[white,fill=green!57] (8.8,-3) rectangle (10,-2.5);
		
		\draw[white,fill=blue!50] (7,-2.5) rectangle (7.7,-2);
		\draw[white,,pattern=north west lines, pattern color=red!80!black] (7.7,-2.5) rectangle (9.4,-2);
		\draw[white,fill=green!40] (9.4,-2.5) rectangle (10,-2);
		
		\draw[white,fill=blue!35] (7,-2) rectangle (8.1,-1.5);
		\draw[white,fill=green!55,pattern=north west lines, pattern color=orange!80!white] (8.1,-2) rectangle (9.5,-1.5);
		\draw[white,fill=green!25] (9.5,-2) rectangle (10,-1.5);

		\draw[white] (7, -1.5) rectangle (10.5, -1.5);
		

		\draw[->] (7.6,-3.3) -- (7.6,-2.75);
		\draw[->] (7.9,-3.8) -- (7.9,-1.75);
		\node at (7.6, -3.5) {\small $m_2$};
		\node at (7.9, -4) {\small $m_1$};
		
		\node at (8.9, -3.8) {\begin{tabular}{l}
				fresh\\ 
				values
			\end{tabular}};
		\draw[->] (8.9,-3.3) -- (8.9,-1.75);
		\draw[->] (8.8,-3.3) -- (8.8,-2.25);
		\draw[->] (8.7,-3.3) -- (8.7,-2.75);
		
		\draw (5.5,-3) rectangle (10.5,-2.5);
		\draw (5.5,-3) rectangle (10.5,-2);
		\draw (5.5,-3) rectangle (10.5,-1.5);
		
		\draw (5.5,-3) rectangle (6.5,-2.5);
		\draw (5.5,-3) rectangle (7.5,-2.5);
		\draw (5.5,-3) rectangle (8.8,-2.5);
		
		\draw (5.5,-2.5) rectangle (7.7,-2);
		\draw (5.5,-2.5) rectangle (9.4,-2);
		
		\draw (5.5,-2) rectangle (6.7,-1.5);
		\draw (5.5,-2) rectangle (8.1,-1.5);
		\draw (5.5,-2) rectangle (9.5,-1.5);
		\draw (5.5,-2) rectangle (10.2,-1.5);
		
		\draw (5.8,-2.5) rectangle (6.8,-2);
		\node at (7.4,-1.8) {$v_1$};
		
	\end{scope}
	
	\draw (5.5,0) rectangle (15,0.5);
	\draw (5.5,0) rectangle (15,1);
	\draw (5.5,0) rectangle (15,1.5);
	\draw (5.5,0) rectangle (15,1.5);
	
	\draw[<->] (7,-0.2) -- (10,-0.2);
	\draw[<->] (11.5,-0.2) -- (14.5,-0.2);
	
	\node (s1) at (8.5,-0.4) {$s$};
	\node (s2) at (13,-0.4) {$s$};
\end{tikzpicture}

%% file: Figures/tree-example1.tex
\begin{tikzpicture} [xscale = 1, yscale = 0.6, every node/.style = {scale = 0.9}]

	\node [align = center , font = {\scriptsize}] (sp) at (6.8, 0.7-0.5) {Boss node $\node_1$ ($a_1$ in $\run$)};
	\node [align =center, font = {\scriptsize}] (sp) at (6.8, 0.7-1) {$\bossspec = \epsilon$};
	\node [align=center, font = {\scriptsize}] (sp) at (6.8, 0.7-1.5) {$v = 1$};
	\begin{scope}[xscale = 1.2, xshift = -10]
		\regbartworeg{-1}{0}
		\onerowtworeg{-0.5}{0}{$q_0$}{$1$}{orange}{$\_$}{white}
		\transtable{0.5}{0}{\color{OrangeRed}$\intlabel{\atrans_{b2}}$}
		\onerowtworeg{0.5}{0}{$q_1$}{$1$}{orange}{$\_$}{white}
		\transtable{1.5}{0}{\color{ForestGreen!80!white}$\extlabel{\atrans_{r3}}{2}$}
		\onerowtworeg{1.5}{0}{$q_3$}{$1$}{orange}{$2$}{cyan}
		\transtable{2.5}{0}{\color{Mahogany}$\extlabel{\atrans_{r4}}{1}$}
		\onerowtworeg{2.5}{0}{$q_4$}{$1$}{orange}{$2$}{cyan}
		\transtable{3.5}{0}{\color{Mahogany}$\intlabel{\atrans_{b4}}$}
		\onerowtworeg{3.5}{0}{$q_4$}{$1$}{orange}{$2$}{cyan}
	\end{scope}
	\draw[->]  (2, -1.2) -- (0, -2);
	\draw[->]  (2, -1.2) -- (4, -2);
	
	\begin{scope}[xshift = 0.7cm, yshift=0.5cm]
		\onerowtworeg{-2}{-3.5}{$q_0$}{$\_$}{white}{$2$}{cyan}
		\transtable{-1}{-3.5}{\color{OrangeRed}$\extlabel{\atrans_{r2}}{3}$}
		\onerowtworeg{-1}{-3.5}{$q_2$}{$3$}{LimeGreen}{$2$}{cyan}
		\transtable{0}{-3.5}{\color{ForestGreen!80!white} $\intlabel{\atrans_{b3}}$}
		\onerowtworeg{-0}{-3.5}{$q_3$}{$3$}{LimeGreen}{$2$}{cyan}
		\node [align=center, font = {\scriptsize}] (sp) at (-3.1, -2.5-0.8) {Boss node $\node_2$ ($a_2$)};
		\node [align = center, font = {\scriptsize}] (sp) at (-3.1,-3-0.8) {$\bossspec = {\color{ForestGreen!80!white} m_3} $};
		\node [align = center, font = {\scriptsize}] (sp) at (-3.1, -3.5-0.8) {$v = 2$};
	\end{scope}
	\begin{scope}[xshift = -0.2, yshift=0.5cm]
		\onerowtworeg{2}{-3.5}{$q_0$}{$\_$}{white}{$\_$}{white}
		\transtable{3}{-3.5}{\color{OrangeRed}$\extlabel{\atrans_{r2}}{1}$}
		\onerowtworeg{3}{-3.5}{$q_2$}{$1$}{orange}{$\_$}{white}
		\transtable{4}{-3.5}{\color{ForestGreen!80!white} $\intlabel{\atrans_{b3}}$}
		\onerowtworeg{4}{-3.5}{$q_3$}{$1$}{orange}{$\_$}{white}
		\transtable{5}{-3.5}{\color{Mahogany}$\intlabel{\atrans_{b4}}$}
		\onerowtworeg{5}{-3.5}{$q_3$}{$1$}{orange}{$\_$}{white}
		\node [align=center, font = {\scriptsize}] (sp) at (7.3, -2.5-0.8) {Follower node $\node_3$ ($a_2$)};
		\node [align = center, font = {\scriptsize}] (sp) at (7.3,-3-0.8) {$\followwordspec={\color{OrangeRed}m_2}, \followmessagespec = { \color{Mahogany}m_4} $};
		\node [align = center, font = {\scriptsize}] (sp) at (7.3, -3.5-0.8) {$v = 1$};
	\end{scope}
	\draw[->] (0, -2.5-1.7) -- (0, -2.5-1.7-0.7);
	\begin{scope}[xscale = 1.2, xshift = 0cm, yshift = 1.1cm]
		\onerowtworeg{-1}{-7}{$q_0$}{$3$}{LimeGreen}{$\_$}{white}
		\transtable{0}{-7}{\color{OrangeRed}$\intlabel{\atrans_{b2}}{}$}
		\onerowtworeg{-0}{-7}{$q_1$}{$3$}{LimeGreen}{$\_$}{white}
		\node [align=center, font = {\scriptsize}] (sp) at (-2, -3-0.8-3) {Boss node $\node_4$ ($a_1$)};
		\node [align = center, font = {\scriptsize}] (sp) at (-2,-3.5-0.8-3) {$\bossspec = {\color{OrangeRed} m_2}$};
		\node [align = center, font = {\scriptsize}] (sp) at (-2, -4-0.8-3) {$v = 3$};
	\end{scope}

\end{tikzpicture}

%% file: Figures/rearrangement-tree.tex
\begin{tikzpicture}[ >=stealth', shorten >=1pt,node distance=2cm,on grid, initial text = {}] 
	
	\tikzset{
		tnode/.style = {rectangle, draw, thick, minimum width=1cm,
			minimum height = 0.6cm}
	}


	\draw[black, ->, thick](-5,-2.2)--(-5,1.4);
	\node[black, thick] at (-5,1.6) {altitude};
	\foreach \y in {-2,...,1} \draw[black, thick] (-5.1, \y )node[left, thick]{\y} -- (-4.9, \y);
	
	 \node[tnode, fill=pink!] (root) { Boss};
	 \node[tnode] (n1) [above = 1cm of root, xshift = 2cm]{ Follower};
	 \node[tnode] (n2) [below = 1cm of root, xshift = 2cm]{ Boss};
	 \node[tnode] (n3) [right = 4cm of root]{ Boss};
	 \node[tnode] (n4) [above = 1cm of n3, xshift = 1.8cm]{ Follower};
	 \node[tnode] (n5) [below = 1cm of n3, xshift = 1.8cm]{ Boss};

	\node[tnode] (m1) [below = 1cm of root, xshift = -2cm]{ Boss};
	\node[tnode] (m2) [above = 1cm of m1, xshift = -2cm]{Follower};
	\node [tnode] (m3) [above = 1cm of m2, xshift = 2cm] {Follower};
	\node[tnode] (n9) [below =  1cm of m1, xshift = 2cm]{ Boss};

	\path[->]
	(root) edge [ ]  node {} (n1)
				edge node {} (n2)
				edge [ ] node {} (m1)
	(n1) edge node {} (n3)
	(n3) edge node {}  (n4)
			edge node {}   (n5)
			
	(m1) edge node {} (m2)
	(m2) edge node {} (m3)
	(m1) edge [ ] node {} (n9)
	;

		\path[->, dashed, teal, thick] 
		(root.east) edge node {} ([xshift = -0.2cm]n1.south)
		(n2.west) edge node {} (root.south)
		([xshift = 0.2cm]m1.north) edge node {} (root.west)
		(n3.west) edge node {} ([xshift = 0.2cm]n1.south)
		([yshift = 0.1cm]n3.east) edge node {} ([xshift = -0.2cm]n4.south)
		 
		([xshift = -0.2cm]n5.north) edge node {}  ([yshift = -0.1cm]n3.east) 
		
		([xshift = -0.2cm]m1.north) edge node {} ([yshift = -0.1cm]m2.east)
		([xshift = 0.2cm]m2.north) edge node {} (m3.west)
		 ([xshift = -0.2cm]n9.north) edge node {} (m1.east) 
		;
		
%

	

\end{tikzpicture}

%% file: General-target.tex
\subsection{Undecidability of the "target problem"}
\label{sec:undec-target}

A natural next problem, after \COVER, is the "target problem" (\TARGET).  
Our \COVER\ procedure heavily relies on the ability to add agents at no cost. For \TARGET\ we need to guarantee that those agents can then reach the target state, which makes the problem harder. 
In fact, \TARGET\ is undecidable, which indicates that our model lies at the frontier of decidability.

\begin{restatable}{proposition}{propTargetUndecidable}
\label{prop:target-undec}
\TARGET\ is undecidable for "BNRA", even with two registers.
\end{restatable}

\begin{proof}[Proof sketch]
We simulate a Minsky machine with two counters. As in Proposition~\ref{prop:reduction-LCS},  each agent starts by storing some other agent's identifier, called its ``predecessor''. It then only accepts messages from its predecessor. As there are finitely many agents, there is a cycle in the predecessor graph. 

In a cycle, we use the fact that \emph{all} agents must reach state $q_f$ to simulate faithfully a run of the machine: agents alternate between receptions and broadcasts so that, in the end, they have received and sent the same number of messages, implying that no message has been lost along the cycle.
We then simulate the machine by having an agent (the leader) choose transitions and the other ones simulate the counter values by memorizing a counter ($1$ or $2$) and a binary value ($0$ or $1$). For instance, an increment of counter $1$ takes the form of a message propagated in the cycle from the leader until it finds an agent simulating counter $1$ and having bit $0$. This agent switches to $1$ and sends an acknowledgment that propagates back to the leader. 
The full proof is provided in Appendix~\ref{app:target}.
\end{proof}

%% file: defs.tex
In this section, we establish the \NP-completeness of the restriction of \COVER\ to "BNRA" with one register per agent, called 1-BNRA. Here we simply sketch the key observations that allow us to abstract runs into short witnesses, leading to an \NP algorithm for the problem.
	
	In 1-BNRA, thanks to the "copycat principle", any  message can be broadcast with a fresh value, therefore one can always circumvent $\quotemarks{\diseqtestact}$ tests. In the end, our main challenge for 1-BNRA is $\quotemarks{\eqtestact}$ tests upon reception.
	For this reason, we look at clusters of agents that share the value in their registers. 

	Consider a "run" in which some agent $a$ reaches some state $q$,; we can duplicate $a$ many times to have an unlimited supply of agents in state $q$.
	Now assume that, at some point in the "run", agent $a$ stored a received value. Consider the last storing action performed by $a$: $a$ was in a state $q_1$ and performed transition $(q_1, \rec{m}{1}{\enregact}, q_2)$ upon reception of a message $(m,v)$. 
	Because we can assume that we have an unlimited supply of agents in $q_1$ thanks to the copycat principle, we can make as many agents as we want take transition $(q_1, \rec{m}{1}{\enregact}, q_2)$ at the same time as $a$ by receiving the same message $(m,v)$. These new agents end up in $q_2$ with value $\aval$, and then follow $a$ along every transition until they all reach $q$, still with value $v$. In summary, because $a$ has stored a value in the run, we can have an unlimited supply of agents in state $q$ with the same value as $a$. 
	
	Following those observations, we define an abstract semantics with abstract configurations of the form $(S, b, K)$ with $S, K \subseteq Q$ and $b \in Q \cup \set{\bot}$. The first component $S$ is a set of states that we know we can cover (hence we can assume that there are arbitrarily many agents in all these states).
	We start with $S = \set{q_0}$ and try to increase it. To do so, we use the two other components (the \emph{gang}) to keep track of the set of agents sharing a value $v$: $b$ (the \emph{boss}) is the state of the agent which had that value at the start, $K$ (the \emph{clique}) is the set of states covered by other agents with that value. As mentioned above, we may assume that every state of $K$ is filled with as many agents with value $v$ as we need. 
	We will thus define abstract steps which allow to simulate steps of the agents with the value we are following. When they cover states outside of $S$, we may add those to $S$ and reset $b$ to $q_0$ and $K$  to $\emptyset$, to then start following another value.
	We can bound the length of relevant abstract runs, and thus use them as witnesses for our \NP upper bound.

%% file: abstraction.tex
The \NP lower bound follows from a reduction from 3SAT. An agent $a$ sends a sequence of messages representing a valuation, with its identifier, to other agents who play the role of an external memory by broadcasting back the valuation. This then allows $a$ to check the satisfaction of a 3SAT formula.
%

\begin{restatable}{theorem}{thmNPComplete}
	\label{thm:np-complete-query-cover}
	The "coverability problem" for 1-BNRA is \NP-complete.
\end{restatable}

The formal proofs of the upper and lower bounds are given in Appendix~\ref{app:cover-one-reg}.

%% file: Appendix/Preuve-LCS-reduction.tex
\section{Proof of Proposition~\ref{prop:reduction-LCS}}
\label{app:reduction-lcs}

\propReductionLCS*
\begin{proof}
	We provide here a polynomial-time reduction from reachability for lossy channel systems with a single channel. A "lossy channel system" with a single channel is a finite-state machine that has the ability to buffer symbols in a lossy FIFO queue \cite{Schnoebelen2002verifying}. 
	Let $\los := (\lstates,\Sigma, \ltransitions)$ be a "lossy channel system", where $\lstates$ is a finite set of locations, $\Sigma$ is a finite set of symbols and $\ltransitions \subseteq \lstates \times \set{\popact{x}, \pushact{x} \mid x \in \Sigma} \times \lstates$; $\pushact{x}$ corresponds to writing $x$ at the end of the channel and $\popact{x}$ to reading $x$ at the beginning the channel. A configuration of $\los$ is a pair in $\lstates \times \Sigma^*$ denoting the location and the content of the channel. There exists a step from $(\lstate,w)$ to $(\lstate',w')$ using $\ltrans \in \ltransitions$, denoted $(\lstate,w) \lstep{\ltrans} (\lstate',w')$, when
	\begin{itemize}
		\item $\ltrans = (\lstate,\pushact{x},\lstate')$ for some $x \in \Sigma$ and $w' \subword w \cdot x$,
		\item $\ltrans = (\lstate,\popact{x},\lstate')$ for some $x \in \Sigma$ and $x \cdot w' \subword w$.
	\end{itemize}
	where $\subword$ denotes the "subword" order, which encodes the lossiness of the channel. 
	
	\AP The existence of a step is denoted $(\lstate,w) \lstep{} (\lstate',w')$, and its transitive closure is denoted $\lstep{*}$. The (location) ""reachability problem@@lcs"" asks, given $\los$ and two locations $\lstate_s, \lstate_f \in \lstates$, whether $(\lstate_s,\epsilon) \lstep{*} (\lstate_f, w)$ for some $w$.

	We construct a "signature protocol" $\prot$ with two registers (register $1$ is "broadcast-only", register $2$ is reception-only) and a state $q_f$ that may be covered if and only if $(\los, q_i, q_f)$ is a positive instance of the "reachability problem@@lcs". 
	
	\AP In some initial phase, each agent may become a ""link"" and store some other agent's identifier (its \emph{predecessor}); in that case, it will test further messages for equality so that only messages from the predecessor are accepted. Otherwise, it will become a ""root"" and will not receive any message in the future.  A depiction of this initial phase can be found in Figure~\ref{fig:lcs-choice}.
	
	The predecessor graph of a run of $\prot$, defined by the graph $(\agents,E)$ where $(a',a) \in E$ whenever $a'$ is the predecessor of $a$, will be a forest where each branch will simulate an execution of $\los$. A given agent $a$ of a given branch encodes one step of the execution in the "lossy channel system". 
A "root" agent is at the root of its tree and simply broadcasts the initial configuration $(\lstate_s,\epsilon)$ of $\los$; special character $\#$ encodes the end of the broadcast of the channel's content.  
	If $a$ is a "link" agent, then it will receive from its predecessor a location of the system and (a subword of) the content of the channel. Agent $a$ will in turn broadcast to the next agent in the branch, sending the new location of the system and the new content of the channel which $a$ rebroadcasts on-the-fly letter by letter as it receives it. Agent $a$ only modifies the beginning of the channel if it decides to encode a pop and the end of the channel if it decides to encode a push. Messages might get lost, which is why we cannot encode non-lossy channel systems.
	
	Observe that this construction only guarantees that agents have at most one predecesor; it does not guarantee that all agents are in the same branch or that any agent is the predecessor of at most one agent. This is fine because the information only propagates forward in a branch, and never propagates in between branches. 

	From state $\waitstate$, an agents will receive the location $\lstate$ from its predecessor and go to state $\startstate{\lstate}$; from there, it will non-deterministically pick a transition of $\los$ and apply it. See Figure~\ref{fig:lcs-trans} for a depiction of transitions from some state $\startstate{\lstate_1}$ (all $\startstate{\lstate}$, $\lstate \ne \lstate_1$, and corresponding transitions are omitted in the figure). 
	Finally, the objective state of our system is $q_f := \finstate{\lstate_f}$.

	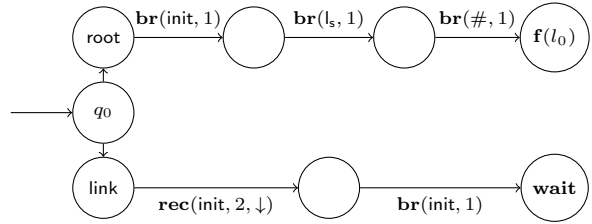
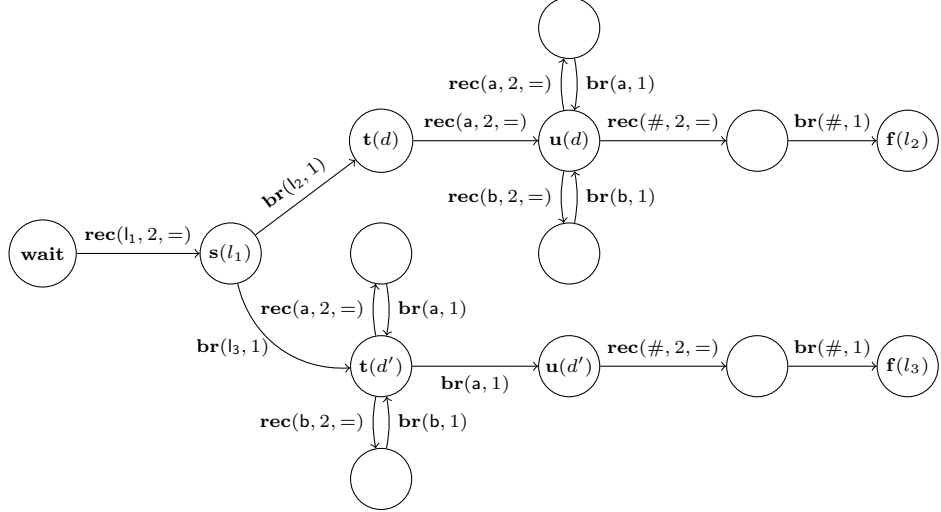
\begin{figure}[h]
	\begin{subfigure}[b]{0.99\textwidth}
	\centering
	\input{Figures/fig-lcs-choice}
	\caption{The initial part of $\prot$}\label{fig:lcs-choice}
	\end{subfigure}
	\begin{subfigure}[b]{0.99\textwidth}
	\centering
	\input{Figures/fig-lcs-trans}
	\caption{Part of $\prot$ encoding transitions from $\lstate_1$. Here, $\los$ has symbols $\Sigma = \set{\mathsf{a},\mathsf{b}}$ and there are two transitions from $\lstate_1$: $\ltrans = (\lstate_1, \popact{a}, \lstate_2)$ and $\ltrans' = (\lstate_1, \pushact{b}, \lstate_3)$.}\label{fig:lcs-trans}
	\end{subfigure}
	\caption{Depiction of the protocol $\prot$ built in the "lossy channel system" reduction of Proposition~\ref{prop:reduction-LCS}.}
	\end{figure}
	We claim that $(\prot, q_f)$ is a positive instance of \COVER if and only if $(\los, \lstate_f)$ is a positive instance of the reachability problem for "lossy channel systems".
	First, suppose that there exists $w \in \Sigma^*$ such that $(\lstate_0, \epsilon) \lstep{*} (\lstate_f, w)$. Decompose the witness into $(\lstate_0, w_0) \lstep{} (\lstate_1, w_1) \lstep{} (\lstate_2, w_2) \cdots \lstep{} (\lstate_n,w_n)$ with $\lstate_n = \lstate_f$ and $w_n =w$. 
	We build an "initial run" of $\prot$ that "covers" $q_f$. It has set of agents $\agents := \set{0,\dots, n}$. Agent $0$ becomes the "root" and for all $i \geq 1$, agent $i$ becomes a "link" with predecessor agent $i-1$. By induction on $i$, we build an execution using agents $0$ to $i$ such that agent $i$ ends on state $\finstate{\lstate_i}$ and the sequence of "message types" sent by agent $i$ admits as subword $\mathsf{init} \cdot \mathsf{\lstate_i} \cdot w_i \cdot \mathsf{\#}$. For $i=0$, this condition is easily met by making agent $0$ become "root". We make agent $i+1$ do the following. It receives from agent $i$ state $l_i$ and goes to $\startstate{l_i}$. It moves to $\transstateone{\ltrans}$ where $\ltrans \in \ltransitions$ is such that $(\lstate_i, w_i) \lstep{\ltrans} (\lstate_{i+1}, w_{i+1})$. It then follows the branch and gets to $\finstate{\lstate_{i+1}}$. 
	\begin{itemize}
		\item if $\ltrans =(\lstate_i,\pushact{x}, \lstate_{i+1})$ is a push, then $w_{i+1} = w_{i}' \cdot y$ where $w_i' \subword w_i$ and $y \in \set{\epsilon, x}$. 
		We can make agent $i+1$ broadcast $\mathsf{init} \cdot \mathsf{\lstate_i} \cdot w'_i \cdot x \cdot \mathsf{\#}$ which admits as subword $\mathsf{init} \cdot \mathsf{\lstate_{i+1}} \cdot w_{i+1} \cdot \mathsf{\#}$.
		\item if $\ltrans =(\lstate_i, \popact{x}, \lstate_{i+1})$ is a pop then $w_{i} = u \cdot w_{i+1}'$ where $w_{i+1} \subword w_{i+1}'$ and $x \subword u$. 
		By lossiness, agent $i+1$ only receives $x$ from $u$ and goes to $\transstatetwo{\ltrans}$. We make it broadcast $\mathsf{init} \cdot \mathsf{\lstate_{i+1}} \cdot w_{i+1}'  \cdot \mathsf{\#}$ which admits as subword $\mathsf{init} \cdot \mathsf{\lstate_{i+1}} \cdot w_{i+1} \cdot \mathsf{\#}$.
	\end{itemize}
	This concludes the induction step.
	When applied to $i=n$, this builds an "initial run" where agent $n$ ends on $\finstate{\lstate_n}$, which is a witness that $(\prot, q_f)$ is positive.
	
	Suppose now that $(\prot, q_f)$ is positive. Let $\run: \config_0 \step{*} \config_f$ where $\config_f$ covers $q_f$. The predecessor graph obtained at the end of $\run$ is a forest; there can be no cycle as it would imply that all agents in the cycle are "link", which is a contradiction by considering the agent of the cycle sending the first $\mathsf{init}$ message. Consider in $\run$ a branch of agents $a_0, \dots, a_n$ such that $a_{i}$ is the predecessor of $a_{i+1}$ for each $i \in \nset{0}{n-1}$, $a_0$ is the "root" and $a_n$ covers $q_f$. 
	
	From the "run" $\run$ projected onto this branch, it is quite simple to build an execution of $\los$ that covers $q_f$. By structure of the protocol, because $a_n$ covers $q_f$, every agent $a_i$ ends on some $\finstate{\lstate_i}$ and broadcasts a word of the form $\mathsf{init} \cdot \lstate_i \cdot w_i \cdot \mathsf{\#}$; this can be proven with an immediate backwards induction. It then suffices to analyse the behavior of $a_{i+1}$ to prove that $(\lstate_i, w_i) \lstep{} (\lstate_{i+1}, w_{i+1})$. In particular, because $a_0$ is a "root", $\lstate_0 = \lstate_s$ and $w_0 = \epsilon$, which concludes the proof. 
\end{proof}

%% file: Figures/fig-lcs-choice.tex
\begin{tikzpicture}[node distance=2cm,auto, xscale = 2]
	\tikzset{every node/.style = {font = {\scriptsize}}}
	\node[state,initial, initial text = ] (0) at (0,0) {$q_0$};
	\node[state] at (0,1) (root) {\textsf{root}};
	\node[state] at (0,-1) (link) {\textsf{link}};
	\node[state] at (1,1) (root1) {}; 	
	\node[state] at (2,1) (root2) {};
	\node[state] at (3,1) (root3) {$\finstate{\lstate_0}$};
	\node[state] at (1.5,-1) (link1) {};
	\node[state] at (3,-1) (wait){$\waitstate$};

	\path[->] 	
		(0) edge (root)
		(0) edge (link)
		(root) edge node[above] {$\br{\mathsf{init}}{1}$} (root1)
		(root1) edge node[above] {$\br{\mathsf{\lstate_s}}{1}$} (root2)
		(root2) edge node[above] {$\br{\mathsf{\#}}{1}$} (root3)
		(link) edge node[below] {$\rec{\mathsf{init}}{2}{\enregact}$} (link1)
		(link1) edge node[below]{$\br{\mathsf{init}}{1}$} (wait)
	;
\end{tikzpicture}

%% file: Figures/fig-lcs-trans.tex
\begin{tikzpicture}[node distance=2cm,auto]
	\tikzset{every node/.style = {font = {\scriptsize}}}
	\node[state] at (-0.5,1) (wait){$\waitstate$};
	\node[state, inner sep = 0pt] at (2,1) (l1) {$\startstate{\lstate_1}$};
	\node[state] at (4,2.5) (tr1) {$\transstateone{\ltrans}$};
	\node[state, inner sep = 0pt] at (6.5,2.5) (tr2) {$\transstatetwo{\ltrans}$};
	\node[state] at (6.5,4) (reba) {};
	\node[state] at (6.5,1) (rebb) {};
	\node[state] at (9,2.5) (endrec) {};
	\node[state, inner sep = 0pt] at (11,2.5) (endbr) {$\finstate{\lstate_2}$};
	\node[state, inner sep = 0pt] at (4, -0.5) (tr1prime) {$\transstateone{\ltrans'}$};
	\node[state] at (4,1) (rebaprime) {};
	\node[state] at (4, -2) (rebbprime) {};
	\node[state, inner sep = 0pt] at (6.5, -0.5) (tr2prime) {$\transstatetwo{\ltrans'}$};
	\node[state] at (9,-0.5) (endrecprime) {};
	\node[state, inner sep = 0pt] at (11,-0.5) (endbrprime) {$\finstate{\lstate_3}$};
	
	\path[->] 	
	(wait) edge node[sloped] {$\rec{\mathsf{\lstate_1}}{2}{\eqtestact}$} (l1)
	(l1) edge node[above, sloped] {$\br{\mathsf{\lstate_2}}{1}$} (tr1)
	(tr1) edge node[above] {$\rec{\mathsf{a}}{2}{\eqtestact}$} (tr2)
	(tr2) edge[bend left = 10] node[left] {$\rec{\mathsf{a}}{2}{\eqtestact}$} (reba)
	(reba) edge[bend left = 10] node[right] {$\br{\mathsf{a}}{1}$} (tr2)
	(tr2) edge[bend right = 10] node[left] {$\rec{\mathsf{b}}{2}{\eqtestact}$} (rebb)
	(rebb) edge[bend right = 10] node[right] {$\br{\mathsf{b}}{1}$} (tr2)
	(tr2) edge node[above] {$\rec{\mathsf{\#}}{2}{\eqtestact}$} (endrec)
	(endrec) edge node[above] {$\br{\#}{1}$} (endbr)
	(rebaprime) edge[bend left = 10] node[right] {$\br{\mathsf{a}}{1}$} (tr1prime)
	(tr1prime) edge[bend left = 10] node[left] {$\rec{\mathsf{a}}{2}{\eqtestact}$} (rebaprime)
	(tr1prime) edge[bend right = 10] node[left] {$\rec{\mathsf{b}}{2}{\eqtestact}$} (rebbprime)
	(rebbprime) edge[bend right = 10] node[right] {$\br{\mathsf{b}}{1}$} (tr1prime)
	(tr2prime) edge node[above] {$\rec{\mathsf{\#}}{2}{\eqtestact}$} (endrecprime)
	(endrecprime) edge node[above] {$\br{\#}{1}$} (endbrprime)
	(l1) edge[bend right = 40] node[left] {$\br{\mathsf{\lstate_3}}{1}$} (tr1prime)
	(tr1prime) edge node[below] {$\br{\mathsf{a}}{1}$} (tr2prime); 
\end{tikzpicture}

%% file: Appendix/definition-unfolding-trees.tex
\section{Definitions and Notations of \cref{sec:cover-general-case}}
\label{app:def-trees}
\label{sec:decidability-defs}

We will first define the notion of decomposition, which is a formalization of the observation in Example~\ref{ex:decomposition}. 
A ""decomposition"" is a tuple $\decsymb = (w_0, m_1, \ldots, m_\ell, w_\ell)$ with $w_0, \ldots, w_\ell \in \messages^*$, and $m_1, \cdots, m_\ell \in \messages$, with $m_i \neq m_j$ for all $i\neq j$. In particular we have $\ell \leq \size{\messages}$. 
A word $w \in \messages^*$ ""admits decomposition"" $\decsymb = (w_0, m_1, \ldots, m_\ell, w_\ell)$ if $w \subword w'_0 w'_1 \cdots w'_\ell$ where for all $j$, $w'_j$ can be obtained from $w_j$ by adding letters from $\set{m_1, \ldots, m_{j}}$. 
We denote by $\intro*\langdec{\decsymb}$ the language of words that "admit decomposition" $\decsymb$. 

We are now ready to formally define "unfolding trees" in the general case.

\begin{definition}
	\label{def:unfolding_tree}
	\AP An ""unfolding tree"" $\tree$ over $\prot$ is
	a finite tree where nodes $\node$ have three labels:
	\begin{itemize}
		\item a "local run" of $\prot$, written $\reintro*\localrunlabel{\node}$, starting in the initial state with distinct register values;
		
		\item a value in $\nats$, written $\reintro*\valuelabel{\node}$;
		
		\item\AP a \intro{specification} $\reintro*\speclabel{\node}$, which is either a word $\intro*\bosslabel{\node} \in \messages^*$ (""boss specification"") or a pair $(\intro*\followlabelword{\node}, \intro*\followlabelmessage{\node}) \in \messages^* \times \messages$ (""follower specification""). In the first case we say that the node is a ""boss node"", otherwise it is a ""follower node"".
	\end{itemize} 
	Moreover, all nodes $\node$ must satisfy the four following conditions:
	\begin{enumerate}[label= (\roman*), ref=(\roman*)]
		\item \label{item:condition1_non_initial_value} For each "non-initial value" $\aval \ne \valuelabel{\node}$ of $\localrunlabel{\node}$, $\node$ has a child $\node'$ which is a "boss node" such that $\vinput{\aval}{\localrunlabel{\node}}$ is a subword of $\bosslabel{\node'}$.
		
		\item \label{item:condition2_initial_value} For each "initial value" $\aval$ in $\localrunlabel{\node}$, there is a "decomposition" \\ $\decsymb = (w_0, m_1, w_1, \ldots, m_{\ell}, w_{\ell})$~s.t.:
		\begin{itemize}
			\item $\localrunlabel{\node}$ may be split into successive "local runs" $\localrun_0, \dots, \localrun_{\ell}$ where, for all $i \in \nset{1}{\ell}$, $w_i \subword \voutput{\aval}{\localrun_i}$ and $\vinput{\aval}{\localrun_i} \in \set{m_1, \dots, m_{i-1}}^*$,
			\item for all $i \in [1,\ell]$, $\node$ has a child $\node_i$ which is a "follower node" such that $\followlabelmessage{\node_i} = m_i$ and $\followlabelword{\node_i} \in\langdec{\decsymb_i}$ where $\decsymb_i = (w_0, m_1, w_1, \ldots, m_{i-1}, w_{i-1})$.	\end{itemize}
		
		\item \label{item:condition3_follower_node} If $\node$ is a "follower node" then $\valuelabel{\node}$ is not an "initial value" of $\localrun$, $\vinput{\valuelabel{\node}}{\localrun} \subword \followlabelword{\node}$ and 
		$\voutput{\valuelabel{\node}}{\localrun}$ contains $\followlabelmessage{\node}$.
		
		\item \label{item:condition4_boss_node} If $\node$ is a "boss node", then $\valuelabel{\node}$ is an "initial value" of $\localrunlabel{\node}$ and the "decomposition" $\decsymb$ of \ref{item:condition2_initial_value} for $\valuelabel{\node}$ satisfies that $\bosslabel{\node} \in \langdec{\decsymb}$.
	\end{enumerate}
	
	\AP Lastly, as before, given $\tree$ an "unfolding tree", we define its ""size@@tree"" by $\size{\tree} := \sum_{\node \in \tree} \size{\localrunlabel{\node}} + \size{\speclabel{\node}} + 1$. 
\end{definition}

We now explain this definition. Let $\node$ be a node of an "unfolding tree" $\tree$ and let $\localrun := \localrunlabel{\node}$. 
As before, $\localrun$ encodes the "local run" of a given agent, $\speclabel{\node}$ encodes the specification that this "local run" carries out and $\valuelabel{\node}$ encodes the value for which the "specification" is carried out.

Conditions \ref{item:condition1_non_initial_value} and \ref{item:condition2_initial_value} state that the "specifications" of the children of $\node$ are witnesses that messages received in the "local run" $\localrunlabel{\node}$ can be broadcast by other agents. 
Conditions \ref{item:condition3_follower_node} and \ref{item:condition4_boss_node} state that $\node$ is a witness that its "specification" is carried out. 

As before, condition \ref{item:condition1_non_initial_value} expresses that, for every "non-initial value" $\aval$ of $\localrun$, $\node$ must have a "boss" child witnessing that $\vinput{v}{\localrun}$ can indeed be broadcast. 
Because $\aval$ was first stored by a "reception step" of $\localrun$, any other (fresh) value with sequence of message types containing $\vinput{v}{\localrun}$ also works and we do not impose the value label of this child to be $v$. 

We now explain condition \ref{item:condition2_initial_value}, which states the existence of a "decomposition" for each "initial value", which serves as a summary of the broadcasts made with that value. Let $v$ be an "initial value" of $\localrun$. Consider a "run" where $\localrun$ is the "local run" of agent $a$. If another agent broadcasts with value $v$, it has first received and stored $\aval$. By duplicating agents, we may afterwards assume that we have an unlimited supply of messages $(m,v)$. Therefore the crucial information is the times at which each "message type" is first broadcast with $v$ by an agent other than $a$.

The "decomposition" $\decsymb = (w_0, m_1, w_1, \ldots, m_\ell, w_\ell)$ should be understood as follows: In the "run" we are representing, $a$ first broadcasts messages of $w_0$ with value $v$, after which other agents are able to broadcast $(m_1,v)$. Then $a$ broadcasts $w_1$ (and may receive $m_1$), after which other agents can broadcast $(m_2, v)$, and so on.
To sustain that description, we need to be able to split $\localrun$ into $u_0,\dots,u_\ell$, with each $u_i$ broadcasting $w_i$ and only receiving messages $m_1,\ldots, m_{i-1}$ over value $v$. We also need a "follower" child for each $m_i$ to witness that other agents may broadcast it.
For every $i$, the sequence of messages available with value $v$ during $u_i$ is $\voutput{v}{u_i}$ expanded by freely adding symbols from $\set{m_1,\dots, m_{i-1}}$. Therefore, the "follower" child $\node_i$ responsible for the broadcast of $(m_i,v)$ may first receive with value $v$ a subword of $w_0' \cdot w_1' \cdots w_{i-1}'$ where, for all $j \leq i-1$, $w_j$ is obtained from $\voutput{v}{u_i}$ by adding symbols from $\set{m_1, \dots, m_{j-1}}$, which we state as $\followlabelword{\node_i} \in \langdec{\decsymb_i}$.   

Condition~\ref{item:condition3_follower_node} directly states that a "follower node" $\node$ receives word $\followlabelword{\node}$ with value $\valuelabel{\node}$ and broadcasts message $(\followlabelmessage{\node}, \valuelabel{\node})$. Condition \ref{item:condition4_boss_node} expresses that a "boss node" witnesses the broadcast of a sequence of messages $\bosslabel{\node}$ with a single value; whereas in the "signature protocol" case, in this sequence, some messages may come from auxiliary agents encoded in "follower" children, which is why we have the condition that $\bosslabel{\node} \in \langdec{\decsymb}$ and not simply $ \voutput{\valuelabel{\node}}{\localrun} \subword \bosslabel{\node}$. 

%% file: Appendix/Preuve-trees-sound-complete.tex
\section{Proof of Proposition~\ref{prop:trees-sound-complete}}
\label{app:trees-sound-complete}

\treessoundcomplete*
 


We make the meaning of the "specification" labels more concrete by defining the criteria for an "initial run" to satisfy a "specification".

\begin{itemize}
	\item A "run" $\run$ satisfies a "boss specification" $\bossspec$ if there exists $\aval \in \nats$ such that $\bossspec$ is a subword of the sequence of messages sent with value $\aval$ in $\run$.
	
	\item A "run" $\run$ satisfies a "follower specification" $(\followwordspec, \followmessagespec)$ if there exist a value $\aval$ and an agent $a$ such that $\aval$ is not an initial value of $a$, the $\aval$-"input" of $a$ in $\run$ is a subword of $\followwordspec$ and agent $a$ broadcasts $\followmessagespec$ with value $\aval$ at some point.
\end{itemize}

For an "unfolding tree", satisfying a "specification" simply means that its root is labeled with that "specification" or a better one.

\begin{itemize}
	\item An "unfolding tree" satisfies a "boss specification" $\bossspec$ if its root $\node$ is a "boss node" and $\bossspec$ is a subword of its "specification" label $\bosslabel{\node}$.
	
	\item An "unfolding tree" satisfies a "follower specification" $(\followwordspec, \followmessagespec)$ if its root $\node$ is a "follower node" such that $\followmessagespec=\followlabelmessage{\node}$ and  $\followlabelword{\node}$ is a subword of $\followwordspec$.
\end{itemize}

Note that \COVER reduces to the existence of an "initial run" satisfying a "boss specification": it suffices to consider $\messages' := \messages \cup \set{m_f}$ with $m_f \notin \messages$ and to add a loop on $q_f$ broadcasting $m_f$. Therefore, it suffices to prove that "initial runs" of $\prot$ and "unfolding trees" satisfy the same "boss specifications". We prove the two directions in the two following sections: we construct an "unfolding tree" from a "run" in  Section~\ref{app:run-to-tree} and a "run" from an "unfolding tree" in Section~\ref{app:tree-to-run}.



\input{Appendix/Preuve-run-to-tree}

	\input{Appendix/Preuve-tree-to-run}

%% file: Appendix/Preuve-run-to-tree.tex
\subsection{Construction of an "unfolding tree" from an "initial run"}
\label{app:run-to-tree}


\begin{lemma}
	\label{lem:run-to-tree}
	If there exists an "initial run" $\run$ of $\prot$ satisfying some "specification" $\spec$ then there exists a finite "unfolding tree" $\tree$ over $\prot$ satisfying $\spec$.
\end{lemma}

\begin{proof}
	
	We proceed by strong induction on $\nats \times \set{\text{follower}, \text{boss}}$ ordered by lexicographic order (follower being lower than boss). The first component is the length of the "run" and the second component is the type of "specification".This means that, for a fixed "run" length, we prove our property for "boss specifications" then for "follower specifications". 
	
	If $\run$ is of length $0$, then its specification must be an empty "boss specification" and the tree with a single node labeled with an empty "local run" and an empty "specification" satisfies it.

	Let $\run$ be an "initial run", and assume that the property is true for "initial runs" whose length is less than the one of $\run$. Write $\agents$ the set of agents of $\run$.
	Because $\run$ satisfies $\spec$, there exist a value $v$ and an agent $a$ in $\run$ such that:
	\begin{itemize}
	\item if $\spec$ is a "boss specification" $\bossspec$, $\bossspec$ is a subword of the sequence of "message types" sent with value $v$ in $\run$ and $a$ is the agent which has $v$ as an "initial value",
	\item if $\spec$ is a "follower specification" $(\followwordspec, \followmessagespec)$, then $a$ is an agent whose $v$-"input" is a subword of $\followwordspec$ and that broadcasts $\followmessagespec$ with value $v$. 
	\end{itemize}

	If the last step of $u$ does not include a broadcast with value $v$, then we remove it and apply the induction hypothesis. Assume that the last step of $\run$ involves a broadcast with value $v$. 
	Let $u$ the "local run" of $a$ in $\run$. We set the root $\node$ of our tree $\tree$ to have local run $u$, value $v$ and "specification" $\spec$ as labels, and attach subtrees to it so that it forms an "unfolding tree".

	We do the following for every value $v'$ in $u$.
	If $v'$ is "non-initial" in $u$ and $v' \neq v$,
	because the last "step" of $\run$ is a broadcast with value $v$, 
	we apply the induction hypothesis on $\run$ without its last step to obtain an "unfolding tree" $\tree'$ whose root is a "boss node" with "boss specification" $\vinput{v'}{u}$, and we attach $\tree'$ below $\node$ in $\tree$. 
	If $v'=v$ and $v'$ is "non-initial" then $\spec$ is a "follower specification" and we do not need to add any children, as the $v'$-input of $u$ is covered by the specification. 
	Assume now that $v'$ is "initial" in $u$. We split $\run$ according to the first broadcast of each message type by agents other than $a$; we obtain $m_1, \dots, m_\ell \in \messages$ distinct message types such that $\run = \run_0 \cdot \run_1 \cdot \ldots \cdot \run_\ell$ where, for all $i \in \nset{1}{l}$:
	\begin{itemize}
	\item agents from $\agents \setminus \set{a}$ do not broadcast $m_i$ with value $v'$ in $\run_0 \cdot \ldots \cdot \run_{i-1}$,
	\item $\run_i$ starts with a broadcast of $m_i$ with value $v'$ by some agent $a_i \in \agents \setminus \set{a}$.
	\end{itemize} 
	
	Let $w_i$ the sequence broadcast by $a$ with value $v'$ in $\run_i$. This forms a "decomposition" $\decsymb := (w_0, m_1, w_1, \ldots, m_\ell, w_\ell)$.
	For all $i \in \nset{1}{\ell}$, we write $\decsymb_i$ for the "decomposition" $\decsymb_i = (w_0, m_1, w_1, \ldots, m_{i-1}, w_{i-1})$. 
	For every $i \in \nset{1}{\ell-1}$, let $\tilde{\run}_i$ the run $\run_0 \cdot \ldots \cdot \run_{i-1}$ plus the first step of $\run_i$ (where $(m_i,v')$ is broadcast).
	Let $w_i'$ the $v'$-"input" of $a_i$ in $\tilde{\run}_i$: it is in $\langdec{\decsymb_i}$ by construction of $\decsymb$. Agent $a_i$ is a witness that $\tilde{\run}_i$ satisfies "follower specification" $(w'_i, m_i)$. 
	If $\tilde{\run_i}$ is smaller than $\run$, we apply the induction hypothesis and obtain an "unfolding tree" $\tree_i$ satisfying "follower specification" $(w_i', m_i)$. If $\tilde{\run_i} = \run$, then $i = \ell$, $v'=v$ ($\run$ ends with a broadcast with value $v$) and $\run$ ends with a broadcast $(m_\ell,v)$ by agent $a_\ell$. Moreover, since $v'=v$, $v$ is "initial" for $a$ and $\spec$ is a "boss specification" $\bossspec$. Because in our induction order "boss specifications" are above "follower specifications", we can apply the induction hypothesis and obtain an "unfolding tree" $\tree_\ell$ satisfying "follower specification" $(w_\ell', m_\ell)$. We have obtained, for every $i \in \nset{1}{\ell}$, an "unfolding tree" $\tree_i$ satisfying "follower specification" $(w_i', m_i)$, which we attach below $\node$. 
	
	Overall we have attached one "boss node" below $\node$ for each non-initial value in $u$ (except $v$ if $\node$ is a "follower node"), thanks to which \ref{item:condition1_non_initial_value} is satisfied at $\node$, and there are well-chosen "follower nodes" below $\node$ for each initial value in $u$, thanks to which \ref{item:condition2_initial_value} is satisfied as well. Conditions \ref{item:condition3_follower_node} and \ref{item:condition4_boss_node} hold because $\run$ satisfies $\spec$ for agent $a$.
	We have therefore built an "unfolding tree" $\tree$ satisfying $\spec$. 
	\end{proof}

%% file: Appendix/Preuve-tree-to-run.tex
\subsection{Construction of an "initial run" from an "unfolding tree"}
\label{app:tree-to-run}

\begin{restatable}{lemma}{LemTreeToRun}
	\label{lem:tree-to-run}
	If there exists an "unfolding tree" over $\prot$ satisfying a "boss specification" $\bossspec \in \messages^*$ then there exists an "initial run" $\run$ of $\prot$ satisfying $\bossspec$.
\end{restatable}

We start by defining "partial runs", which are "runs" where some receptions are not matched by broadcasts. Intuitively, this will allow us to build "partial runs" from "unfolding trees" whose root is a "follower node", which implicitly rely on their parent for some broadcasts. We will therefore construct inductively "partial runs" from nodes of the tree; we will obtain complete "runs" by matching reception and broadcasts of different "partial runs".

\begin{definition}
	Let $\config, \config'$ two configurations. 
	
	A ""partial step"" $\config \pstep{} \config'$ is defined if either $\config \step{} \config'$ (normal "step") or there exist $m \in \messages$, $v\in \nats$ such that for all agent $a$ either $\config(a) = \config'(a)$ or $\config(a) \extbr{\delta}{v} \config'(a)$ with $\delta$ a transition receiving "message type" $m$ (""unmatched reception"").
	
	\AP A ""partial run"" is a sequence of "partial steps".
	Note that a "local run" can be seen as a "partial run" with a single agent. A "partial run" is ""initial@@partial"" if it starts in an "initial configuration".
	The \reintro{$v$-input} $\reintro*\vinput{v}{\run}$ of a "partial run" $\run$ is the sequence $m_0 \cdots m_k$ of "message types" corresponding to "unmatched receptions" with value $v$ in $\run$. Its \reintro{$v$-output} $\reintro*\voutput{v}{\run}$ is the sequence of "message types" corresponding to "broadcasts" with value $v$ in $\run$.
\end{definition}

We will prove the following, more general version of Lemma~\ref{lem:tree-to-run}:

\begin{lemma}
\label{lem:tree-to-run-technical}
For every "unfolding tree" $\tree$:
\begin{itemize}
	\item if $\tree$ satisfies a "boss specification" $\bossspec \in \messages^*$, then there exists an "initial run" $\run$ satisfying $\bossspec$,
	\item if $\tree$ satisfies a "follower specification" $(\followwordspec, \followmessagespec)$ then there exist an "initial partial run" $\run$ and a value $v$ such that:
	\begin{itemize}
	\item all "unmatched receptions" in $\run$ are with value $v$,
	\item $\vinput{v}{\run} \subword \followwordspec$, 
	\item $\voutput{v}{\run}$ contains $\followmessagespec$.
	\end{itemize}
\end{itemize}
\end{lemma}
\begin{proof}
We prove it by induction on the "size" of the "unfolding tree".
Let $\tree$ be a "unfolding tree", $\node$ its root, $u := \localrunlabel{\node}$ and $V$ the set of values appearing in $u$. 
We will combine $u$ with runs given by children of $\node$ to construct $\run$ satisfying the desired property. We build $\rho$ inductively in two steps. 

\subsubsection{Step 1: "Non-initial Values"}
\label{sec:tree-to-run-step-one}

For each non-initial value $v \ne \valuelabel{\node}$ of $u$, $\node$ has a "boss" child of specification $w$ such that $\vinput{v}{u} \subword w$.
By induction hypothesis, there is an "initial run" $\run'$ satisfying "boss specification" $w$. Up to renaming agents, assume that $\run$ and $\run'$ have disjoint agents.
We rename values in $\run'$ so that $w$ is broadcast in $\run'$ with value $v$, and $\run'$ has no other shared value with $\run$. 
We use the broadcasts made by $\run'$ to match the "unmatched receptions" with value $v$ in $\run$: this gives us a new partial "run" $\run$ with no "unmatched reception" with value $v$ and whose behaviour on other values of $V$ is the same as before.

\subsubsection{Step 2: "Initial Values"}
\label{sec:tree-to-run-step-two}


Let $v$ be an "initial value" of $u$, and $\decsymb = \\(w_0, m_1, w_1, \ldots, m_\ell, w_\ell)$ the "decomposition" from condition \ref{item:condition2_initial_value}. We have that, for all $j \in \nset{1}{\ell}$, $\node$ has a "follower" child $\node_j$ labelled by $\followlabelmessage{\node_j} = m_j$ and $\followlabelword{\node_j} \in \langdec{\decsymb_j}$ with $\decsymb_j = (w_0, m_1, w_1, \ldots, m_{j-1}, w_{j-1})$. 

The behavior of the run $\run$ with respect to $v$ is the one of $u$, as we have not added any broadcasts or receptions with $v$. Hence we can split $\run$ into $\run_0, \ldots, \run_\ell$ with $w_i \subword \voutput{\aval}{\run_i}$ and $\vinput{\aval}{\run_i} \in \set{m_1, \dots, m_{i}}^*$ for all $i$.

By induction hypothesis applied to $\node_j$, for all $j$, there exists an "initial partial run" $\Tilde{\run}_j$ whose only "unmatched receptions" are on $v$, $\vinput{v}{\Tilde{\run_j}} \subword \followlabelword{\node_{j}}$ and such that $\Tilde{\run}_j$ broadcasts $(m_j,v)$ in its last step. Again, we rename agents and values so that the sets of agents of $\run$ and of every $\Tilde{\run_j}$ are all disjoint and the only shared value between any of these runs is $v$.
As $\followlabelword{\node_{j}} \in \langdec{\decsymb_j}$ and $\vinput{v}{\Tilde{\run_j}} \subword \followlabelword{\node_{j}}$, we can split $\Tilde{\run}_j$ into $\Tilde{\run}_{j, 0}, \ldots, \Tilde{\run}_{j,j-1}$ so that  $\vinput{v}{\Tilde{\run}_{j, i}} \subword \Tilde{w}_{j,i}$ where $\Tilde{w}_{j,i}$ can be obtained by adding letters from $\set{m_1, \ldots, m_j}$ to $w_i$.

We use the following composition operation: consider $\run$ and one of the $\Tilde{\run}_j$. We can design a new run in which we execute both runs in parallel over disjoint sets of agents. We match each $\Tilde{\run}_{j, i}$ with $\run_i$ so that the broadcasts of $\run_i$ with value $v$ forming $w_i$ are received in $\Tilde{\run}_{j, i}$ and the only remaining missing broadcasts in that section of the run are of $m_1, \ldots, m_i$. 
We obtain a run section whose $v$-"output" still contains $w_i$ and whose "$v$-input" only contains $m_1, \ldots, m_i$.
This lets us get to a point where the next step in $\Tilde{\run}_j$ is a broadcast $(m_j,v)$ and $\run$ has been executed up to the beginning of $\run_j$. We may then use the $(m_j,v)$ broadcast at any moment in the rest of $\run$ either to complete an "unmatched reception" or to extend the "$v$-output" of $\run$.

 This construction is illustrated in Figure
 \begin{figure}
	 	\input{Figures/Fig-tree-to-run}
	 	\caption{An illustration of the composition operation from the proof of Lemma~\ref{lem:tree-to-run}}
	 	\label{fig:tree-to-run}
	 \end{figure}

The resulting run $\run'$ can still be split into $\run'_0 \cdots \run'_\ell$ where $w_i \subword \voutput{\aval}{\run_i}$ and $\vinput{\aval}{\run_i} \in \set{m_1, \dots, m_{i}}^*$ for all $i$. Its input on all values other than $v$ is the same as the one of $\run$. This procedure can thus be iterated.

If $\val = \valuelabel{\node}$ then we have $\bosslabel{\node} \in \langdec{\decsymb}$, and we need to ensure that $\bosslabel{\node}$ is broadcast in $\run$ with value $\val$.
Let $\bossspec_0, \ldots, \bossspec_\ell$ such that $\bosslabel{\node} \subword \bossspec_0\cdots\bossspec_\ell$ and for all $j$, $\bossspec_j$ can be obtained by adding letters from $\set{m_1,\ldots, m_j}$ in $w_j$.
We use the composition operation to extend the output of $\run$, so that $\bossspec_i \subword \voutput{v}{\run_i}$ for all $i$.

Then, in all cases, we apply the composition as many times as necessary to match the "unmatched receptions": while there is an "unmatched reception" of some $(m_j,v)$ we compose $\run$ with $\Tilde{\run}_j$ to eliminate it. We may be adding some "unmatched receptions" of $(m_i, v)$ for some $i<j$.
This procedure terminates as the number of "unmatched receptions" of $m_\ell, \ldots, m_1$ decreases at each iteration for the lexicographic ordering.

In the end we obtain a run $\run$ with no "unmatched reception" on $v$, and such that, if $v = \valuelabel{\node}$ and $\node$ is a "boss node", $\bosslabel{\node} \subword \voutput{v}{\run}$.

\subsubsection{Concluding the procedure}
We distinguish two cases depending on the type of "specification" of $\node$.
\begin{itemize}
\item If $\node$ is a "boss node" of "boss specification" $\bossspec$, we apply steps $1$ and $2$ to every value in $V$, which can be done because $\valuelabel{\node}$ is initial in $\localrun$. We then obtain an "initial run" $\run$ with no "unmatched receptions" satisfying "boss specification" $\bossspec$ for value $\valuelabel{\node}$.
\item if $\node$ is a "follower node", we apply steps $1$ and $2$ to every value in $V \setminus \set{\valuelabel{\node}}$ to obtain a "partial run" $\run$ with no "unmatched reception" on values different from $\valuelabel{\node}$. Moreover, because the behavior of $\run$ with respect to $\valuelabel{\node}$ is the one of $\localrun$, we have that $\vinput{\valuelabel{\node}}{\run} \subword \followwordspec$ and $\voutput{\valuelabel{\node}}{\run}$ contains $\followmessagespec$, concluding the proof. 
\end{itemize}
\end{proof}

%% file: Figures/Fig-tree-to-run.tex
\begin{tikzpicture}[xscale=0.9]
	
	\draw[->] (-0.5,-1.8) -- (3.5,-1.8);
	\draw[->] (5.2,-1.8) -- (7.2,-1.8);
	\draw[->] (8.9,-1.8) -- (11.9,-1.8);
	
	\node (0) at (-1,-1.8) {$\localrunlabel{\node}$};
 	\node (1) at (4.7,-1.8) {$\run_1$};
 	\node (2) at (8.4,-1.8) {$\run_2$};

	\node (b) at (0,-1.6) {$\brone{\color{blue}a\color{black}}$};
	\node (b) at (1,-2.1) {$\recsymb({\color{blue}b\color{black}})$};
	\node (b) at (2,-1.6) {$\brone{\color{blue}a\color{black}}$};
	\node (b) at (3,-2.1) {$\recsymb({\color{blue}c\color{black}})$};
	
	\node (b) at (5.6,-2.1) {$\recsymb({\color{blue}a\color{black}})$};
	\node (b) at (6.6,-1.6) {$\brone{\color{blue}b\color{black}}$};
	
	\node (b) at (9.4,-2.1) {$\recsymb({\color{blue}a\color{black}})$};
	\node (b) at (10.4,-2.1) {$\recsymb({\color{blue}b\color{black}})$};
	\node (b) at (11.4,-1.6) {$\brone{\color{blue}c\color{black}}$};

	\node (1) at (1.5,-2.6) {$\bosslabel{\node}=aac$};
	\node (1) at (6.5,-2.6) {$\run_1 \text{ satisfies } (a,b)$};
	\node (1) at (10.5,-2.6) {$\run_2 \text{ satisfies } (ab,c)$};

	\node (1) at (1.5,-3) {$\decsymb = (a,b, a, c, \epsilon)$};
	\node (1) at (6.5,-3) {$\decsymb_1 = (a)$};
	\node (1) at (10.5,-3) {$\decsymb_2 = (a, b, a)$};

	\draw[->] (-0.5,-4) -- (5.5,-4);
	\draw[->] (-0.5,-5) -- (3.5,-5);
	\draw[dashed] (-0.6, -3.6) rectangle (0.5, -5.1);
	
	\node (b) at (0,-4.8) {$\brone{\color{blue}a\color{black}}$};
	\node (b) at (1,-5.2) {$\recsymb({\color{blue}b\color{black}})$};
	\node (b) at (2,-4.8) {$\brone{\color{blue}a\color{black}}$};
	\node (b) at (3,-5.2) {$\recsymb({\color{blue}c\color{black}})$};

	
	\node (b) at (0,-4.2) {$\recsymb({\color{blue}a\color{black}})$};
	\node (b) at (4,-4.2) {$\recsymb({\color{blue}b\color{black}})$};
	\node (b) at (5,-3.8) {$\brone{\color{blue}c\color{black}}$};
	
	\node (0) at (-1,-4) {$\run_2$};
	\node (0) at (-1,-5) {$\localrunlabel{\node}$};
	
	\draw (-1.7, -5.5) rectangle (7, -3.5);
	\node[text width =4cm] (txt2) at (9.5, -4.3) {We compose $\localrunlabel{\node}$ with $\run_2$ to obtain a "partial run" $\run'$ which outputs the desired word $aac$};
	
	\draw[->] (-0.5,-6.5) -- (4.5,-6.5);
	\draw[->] (-0.5,-7.5) -- (6.5,-7.5);
	\draw[dashed] (-0.6, -6.2) rectangle (0.5, -7.6);
	\draw[dashed] (3.5, -6) rectangle (4.5, -7.9);
	
	\node (b) at (0,-7.3) {$\brone{\color{blue}a\color{black}}$};
	\node (b) at (1,-7.7) {$\recsymb({\color{blue}b\color{black}})$};
	\node (b) at (2,-7.3) {$\brone{\color{blue}a\color{black}}$};
	\node (b) at (4,-7.7) {$\recsymb({\color{blue}c\color{black}})$};
	\node (b) at (5,-7.7) {$\recsymb({\color{blue}b\color{black}})$};
	\node (b) at (6,-7.3) {$\brone{\color{blue}c\color{black}}$};
	
	\node (b) at (0,-6.7) {$\recsymb({\color{blue}a\color{black}})$};
	\node (b) at (3,-6.7) {$\recsymb({\color{blue}b\color{black}})$};
	\node (b) at (4,-6.3) {$\brone{\color{blue}c\color{black}}$};
	
	\node (0) at (-1,-6.5) {$\run_2$};
	\node (0) at (-1,-7.5) {$\run'$};
	
	\draw (-1.7, -8) rectangle (7, -5.9);
	\node[text width = 4cm] (txt1) at (9.5, -6.5) {We compose $\run'$ with $\run_2$ to match an "unmatched reception" $b$};
	
%
%
	
%
	
	\node (d) at (5,-8.3) {\Huge \vdots};
 	\node[text width=4cm] (d) at (9.5,-8.5) {... and iterate compositions with $\run_{1}$ and $\run_2$ until the "run" does not have any "unmatched reception".};
\end{tikzpicture}

%% file: Appendix/Preuves-reductions-branches.tex
\section{Proof of Lemma~\ref{lem:shortening-branches}}
\label{app:proofs-reduction-branches}

\lemShorteningBranches*

\begin{proof}
	Let $\tree_{\node}$, $\tree_{\node'}$ be the subtrees rooted in $\node$, $\node'$ respectively. 
	Let $\tree'$ be the tree obtained by replacing $\tree_{\node}$ with $\tree_{\node'}$; it is strictly smaller than $\tree$. We prove that $\tree'$ is an "unfolding tree".
	The only problematic node is the parent of $\node$ if it exists, because it does not have the same children in $\tree$ and in $\tree'$. Assume that $\node$ is not the root of $\tree$, and let $\node_p$ its parent in $\tree'$.
	The only problematic values are the ones that had $\node$ as witness in conditions \ref{item:condition1_non_initial_value} (if $\node$ is a "boss node") and \ref{item:condition2_initial_value} (if $\node$ is a "follower node"). Let $v$ such a value. 
	If $\node$ is a "boss node" ($v$ is non-initial), we have $\vinput{v}{\localrunlabel{\node_p}} \subword \bosslabel{\node} \subword \bosslabel{\node'}$ hence condition \ref{item:condition1_non_initial_value} is also satisfied. If $\node$ is a "follower node" ($v$ is initial), then, reusing the notations of condition \ref{item:condition2_initial_value}, $\node$ was such that $\followlabelmessage{\node} = m_i$ and $\followlabelword{\node} \in \langdec{\decsymb_i}$. In this case, we also have $\followlabelmessage{\node'} = m_i$ and, because $\followlabelword{\node'} \subword \followlabelword{\node}$ and $\langdec{\decsymb_i}$ is closed by subword, we have $\followlabelword{\node'} \in \langdec{\decsymb_i}$ and condition \ref{item:condition2_initial_value} is satisfied. In both cases, $\tree'$ is an "unfolding tree" with the same root specification as $\tree$ hence a "coverability witness".
\end{proof}

%% file: Appendix/Preuve-Tower-lemma.tex
\section{Generalization of Lemma~\ref{lem:towerbound_signature}}
\label{app:tower-lemma}
		

	\lemShortLocalRuns*

	The function $\towerfun$ is actually not the same as in Lemma~\ref{lem:towerbound_signature} although it is also a tower of exponentials.
	
	We start by defining the notion of "trace", which is a "local run" annotated with received values. 
	\AP A ""trace"" is a sequence in $(\set{\extlabel{\atrans}{\aval} \mid \atrans \in \transitions, \aval \in \nats} \cup \set{\intlabel{\atrans} \mid \atrans \in \transitions})^*$. The "trace" of a "local run" $\localrun$ is the "trace" $\trace{\localrun}$ corresponding to the "local steps" performed in $\localrun$. Given a "trace" $\atrace$, we write $(q,\localdata) \step{\atrace} (q',\localdata')$ when there exists a "local run" of "trace" $\atrace$ from $(q,\localdata)$ to $(q',\localdata')$.



We actually prove a more general version of the lemma. The previous lemma can be obtained simply by applying the following one with $W$ the set of "initial" values of $\localrun_0$ ($W$ then contains $r$ values) and $V$ the set of values appearing in $u_0$, $u$ or $u_f$.

\begin{lemma}
	There exists a primitive recursive function $\towerfun(n,r)$ such that, for every protocol $\prot$ with $r$ registers per agent, for every "local run" $\localrun: (q, \localdata) \step{*} (q', \localdata')$ in $\prot$, for every $V \subseteq \nats$ finite such that $V$ contains all message values appearing in $\localrun$,  for every $\Vinit \subseteq V$, there exists a "local run" $\localrun': (q, \localdata) \step{*} (q', \localdata')$ such that we have $\length{\localrun'} \leq \towerfun(\size{\prot} -r + \size{\Vinit},r)$ and:
	\begin{enumerate}
		\item \label{item:shorterrun_oldvalues} for all $\aval \in V$, $\vinput{\aval}{\localrun'} \subword \vinput{\aval}{\localrun}$,
		
		\item \label{item:shorterrun_anyvalue} for all $\aval' \in \nats \setminus V$, there exists $\aval \in \nats \setminus \Vinit$ such that $\vinput{\aval'}{\localrun'} \subword \vinput{\aval}{\localrun}$.
	\end{enumerate}
\end{lemma}

Intuitively, the set $V$ represents values that are already used and therefore cannot be used as fresh values, and $\Vinit$ represents values that should not be copied (\emph{e.g.}, initial values of the run). However, stating the lemma only with $\Vinit$ equal to the initial values in $u$ would not allow us to apply the lemma on $u$ an infix "local run" of a larger "local run" $u'$, because the initial values in $u$ would not correspond to the ones in $u'$. This is why we state the lemma for every $W \subseteq V$.  

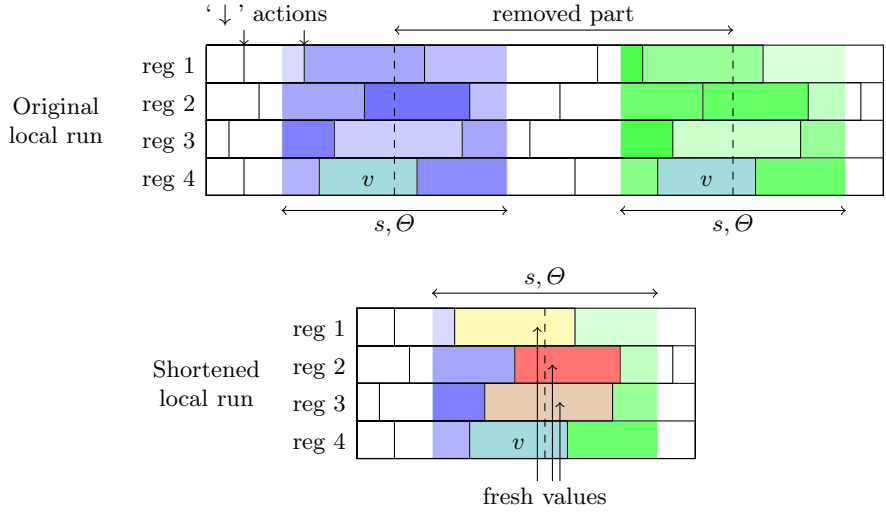
\begin{figure}[h]
	\input{Figures/fig-illustration-pumping}
	\caption{Illustration of the proof of Lemma~\ref{lem:short-local-runs}. $s \in \transitions^{2M}$ corresponds to the repeating sequence of transitions of length $2M$. Register $4$ contains value $v \in W$. Because both sides have the same $\Theta$, they coincide on values in $W$ such as $v$; only values that are not in $W$ are replaced by fresh values, hence $v$ is kept in the shortened local run.}
	\label{fig:pumping}
\end{figure}

\begin{proof}
	Given a "local run" $\localrun$, register $i$ is ""active"" in $\localrun$ if at least one $\quotemarks{\enregact}$ step on register $i$ is performed in $\localrun$. 
	
	We define the function $\towerfun(n,k)$ recursively as $\towerfun(n,0) = n+1$ and $\towerfun(n,k+1) = 2(\towerfun(n,k) +1) [(n+1)^{4(\towerfun(n,k)+1)} +1]$. This function is clearly primitive recursive, although non-elementary; it grows as a tower of exponentials of height $k$ where each floor of the tower is polynomial in $n$.
	
	We prove the ""shortening property"": \\
	Let $W$ a finite set a values and $n:=\size{\prot} -r + \size{\Vinit}$.
	Let a "local run" $\localrun: (q_i, \localdata_i) \step{*} (q_f, \localdata_f)$ with $k$ "active registers" with $\length{\localrun} > \towerfun(n,k)$ and let $V \subseteq \nats$ finite that contains every message value appearing in $\localrun$ such that $W \subseteq V$. We claim that $\localrun$ can be shortened into a local run $\localrun': (q_i, \localdata_i) \step{*} (q_f, \localdata_f)$ with $k$ active registers such that $\length{\localrun'} < \length{\localrun}$ and:
	\begin{itemize}
		\item for all $\aval \in V$, $\vinput{\aval}{\localrun'} \subword \vinput{\aval}{\localrun}$,
		\item for all $\aval' \in \nats \setminus V$, there exists $\aval \in \nats \setminus W$ such that $\vinput{\aval'}{\localrun'}$ is a subword of $\vinput{\aval}{\localrun}$.
	\end{itemize}
	
	We proceed by induction on the number $k$ of "active registers" in the "local run". If $k=0$, register values do not change in $\localrun$. As $\towerfun(n,0) = n+1 \geq \size{Q}+1$, $\localrun$ goes through the same state twice, hence all steps in between may be removed to get $\localrun'$. Then for all $v \in \nats$,  $\vinput{\aval}{\localrun'} \subword \vinput{\aval}{\localrun}$.
	
	Suppose that the property is true for any "local run" with $\leq k$ "active registers", and consider one $\localrun: (q_i,\localdata_i) \step{*} (q_f,\localdata_f)$ with $k+1$ "active registers" such that $\length{\localrun} > \towerfun(\size{\prot} -r + \size{\Vinit},k+1)$.

	First, if there exists an infix "local run" $\localrun_i$ of $\localrun$ of length $\towerfun(n,k)+1$ with only $k$ active registers, then it suffices to apply the induction hypothesis on $u_i$.
	Suppose now that there exists no such infix "local run".
	Let $I \subseteq \nset{1}{r}$ the set of "active registers" in $\localrun$, $|I| = k+1$. Let $M:= \towerfun(n,k){+}1$, we have $\towerfun(n,k+1) = 2M [((n+1)^2)^{2M}+1]$. 
	In any sequence of $M$ "local steps" in a row in $\localrun$, 
	there is a $\quotemarks{\enregact}$ transition on every register in $I$. We can assume that  no local configuration appears twice in $\localrun$ (otherwise we can simply cut the steps between those two appearances to get $\localrun'$). 
	
	In what follows we will consider two infixes of $\localrun$ of length $2M$, following the same sequence of transitions and with the same values of $W$ appearing at the same times in the same registers. Their existence is guaranteed by the length of $\localrun$. As a $\quotemarks{\enregact}$ action is performed twice on each register in those two infixes, we will be able to reduce the run as in \cref{fig:pumping}.
	
	For every $i$, let $\atrans_i$ the $i$-th transition in $\localrun$, and let $\theta_i\in \Vinit \cup \set{\bot}$ be such that 
	\[
	\theta_i = 
	\begin{cases}
		& 	v \text{ if the $i$th step of $u$ is a "reception step" $\extlabel{\delta_i}{v}$ with $v \in \Vinit$ } \\ 
		& 	\bot \text{ otherwise}
	\end{cases}
	\]

	For every $i \in \nset{0}{(n+1)^{4M}}$, we write $s_i$ the sequence $\delta_{2  M \cdot i+1}, \delta_{2  M \cdot i+2}, \cdots, \delta_{2 M \cdot i+2M}$ and $\Theta_i$ the sequence $\theta_{2  M \cdot i}, \theta_{2 M \cdot i+1}, \cdots, \theta_{2 M \cdot i+2M}$.
	There are $|\transitions|^{2M}$ possible sequences for $s_i$ and $[\size{\Vinit} +1]^{2M}$ possible sequences for $\Theta_i$.
	By the pigeonhole principle there exist two indices $i_a, i_b$ such that the sequences $s_{i_a}$ and $s_{i_b}$ are equal and also the sequences $\Theta_{i_a}$ and $\Theta_{i_b}$ are equal (as $|\transitions| (\size{\Vinit}+1) \leq (n+1)^2$ because $n=\size{\prot} -r + \size{\Vinit}$). 
	There exist two infix "local runs" $\localrun_a: (q_1, \localdata_1) \step{*} (q_2, \localdata_2)$, $\localrun_b: (q_3, \localdata_3) \step{*} (q_4, \localdata_4)$ in $\localrun$ such that $(q_2,\localdata_2)$ appears strictly before $(q_3,\localdata_3)$ in $\localrun$ and $\localrun_a$ and $\localrun_b$ both have the same sequences of transitions, which we call $s$, and of receptions of values from $\Vinit$, which we call $\Theta$.
	
	Although $\localrun_a$ and $\localrun_b$ have the same sequence of transitions, their "traces" may differ because their "reception steps" may have different values.
	We build a "trace" $\atrace$ such that $(q_1,\localdata_1) \step{\atrace} (q_4, \localdata_4)$ where the underlying sequence of transitions of $\atrace$ is $s$ and the receptions of values of $\Vinit$ match $\Theta$.
	
	For every active register $i \in I$, let $e_i \in \nset{1}{2M}$ denote the index of the first $\quotemarks{\enregact}$ on register $i$ in $s$ and $f_i \in \nset{1}{2M}$ the index of the last $\quotemarks{\enregact}$ on register $i$ in $s$. By hypothesis, because $s$ is of length $2M$, it contains at least two $\quotemarks{\enregact}$ on register $i$, one in the first half and one in the second half, hence $e_i \leq M < M +1 \leq f_i$. 
	
	For every $j \in \nset{1}{2M}$, let $\delta_j$ denote the $j$-th transition of $s$. First, if  $\delta_j$ is a "broadcast", we define the $j$-th "local step" of $\atrace$ as $\intlabel{\delta_j}$. 
	Suppose now that $\delta_j$ is a "reception" of the form $\rec{\amessage}{i}{\anact}$. The $j$-th "local step" of $\localrun_a$ (resp. $\localrun_b$) has underlying transition $\delta_j$ hence is a "reception step" of the form $\extlabel{\delta_j}{\aval_a}$ for some $\aval_a \in V$ (resp. $\extlabel{\delta_j}{\aval_b}$ for some $\aval_b \in V$).
	Because $V$ is finite and $\Vinit \subseteq V$, there exists an injective function $\phi: V \rightarrow \nats$ such that for all $v \in \Vinit$, $\phi(v) = v$ and for all $v \notin \Vinit$, $\phi(v) \notin V$.
	We define the $j$-th "local step" of $\atrace$ to be $\extlabel{\delta_j}{\aval}$ where:
	\begin{itemize}
		\item if $i \notin I$ then its value stays the same throughout $\localrun$, then we set $\aval = \aval_a (= \aval_b)$ 
		\item if $i \in I$ then  
		\begin{itemize}
			\item if $j < e_i$, $\aval = \aval_a$,
			\item if $e_i \leq j < f_i$, $\aval = \phi(\aval_a)$,
			\item if $f_i \leq j$, $\aval = \aval_b$.
		\end{itemize}
	\end{itemize}
	We now claim that $(q_1,\localdata_1) \step{\atrace} (q_4,\localdata_4)$. 
	First, for every active register $i$, the last $\quotemarks{\enregact}$ step on register $i$ has value $\localdata_4(i)$ in $\atrace$ (as we are in the case $f_i \leq j$). Hence if every "local step" is valid then the final "local configuration" is $(q_4, \localdata_4)$.
	For every $l \in \nset{0}{2M}$, let $\atrace_l$ denote the prefix of $\atrace$ of length $l$.
	We prove by induction on $l$ that $\atrace_l$ is valid from $(q_1, \localdata_1)$. It is trivially true for $l =0$. Assume that we have $(q_1, \localdata_1) \step{\atrace_l} (q, \localdata)$ and let $\locallabel$ such that $\atrace_l \cdot \locallabel = \atrace_{l+1}$. Let $\atrans$ the underlying transition of $\locallabel$.
	First, $q$ is the initial state of $\atrans$ because $\locallabel$ is valid at step $l+1$ of $\localrun_a$ (and $\localrun_b$). Hence if $\atrans$ is a "broadcast" then $\locallabel$ is valid from $(q,\localdata)$. Suppose now that $\locallabel$ is a "reception step"; let $\atrans =: \rec{m}{i}{\alpha}$.
	Let $\aval_a$, $\aval_b$ be the value of $\locallabel$ in $\localrun_a$ and $\localrun_b$ respectively.

	Let $\localdata_a, \localdata_b$ be the content of registers after the $l$-th step in $\localrun_a$ and $\localrun_b$ respectively. 
	 If $i \notin I$ then $\alpha$ is either $\quotemarks{\dummyact}$ or a test, which is valid as  $\localdata(i) = \localdata_a(i) = \localdata_b(i)$ (value of register $i$ does not change in $u$).	
	Suppose now that $i \in I$, the only problematic case is the one of tests, \emph{i.e.}, $\anact \in \set{\quotemarks{\diseqtestact},\quotemarks{\eqtestact}}$. 
	In this case, we prove that $\binrel{\aval}{\anact}{\localdata(i)}$. First, because the corresponding step is valid in $\localrun_a$ and $\localrun_b$, we have $\binrel{\localdata_a(i)}{\anact}{\aval_a}$ and $\binrel{\localdata_b(i)}{\anact}{\aval_b}$. We distinguish cases depending on $l+1$:
	\begin{itemize}
		\item $l+1<e_i$: $\localdata(i)= \localdata_a(i)$, $\aval = \aval_a$ and $\binrel{\localdata_a(i)}{\anact}{\aval_a}$. 
		\item $e_i \leq l+1 < f_i$: We have $\aval = \phi(\aval_a)$. 
		Moreover, because $e_i < l+1$, there is at least one $\quotemarks{\enregact}$ on register $i$ in $\atrace_l$. Consider the last such transition in $\atrace_l$; its index $j$ satisfies $e_i \leq j < f_i$ by definition of $e_i$, hence the value of the corresponding "reception step" in $\localrun_a$ is $\localdata_a(i)$ and its value in $\atrace_{l}$ is $\phi(\localdata_a(i))$. One has $\binrel{\localdata_a(i)}{\anact}{\aval_a}$ therefore (by injectivity of $\phi$ for $\anact = \quotemarks{\diseqtestact}$)
		$\binrel{\phi(\localdata_a(i))}{\anact}{\phi(\aval_a)}$.
		\item$f_i \leq l+1$: $\localdata(i)= \localdata_b(i)$ and $\aval = \aval_b$, and because the "internal step" is valid in $\localrun_b$ we have $\binrel{\localdata_b(i)}{\anact}{\aval_b}$. 
	\end{itemize}
	
	This proves that $\locallabel$ is valid from $(q,\localdata)$ which concludes the induction. 
	We have proven that $(q_1,\localdata_1) \step{\atrace} (q_4, \localdata_4)$; moreover $\atrace$ is of length $2M$ and there are at least $4M>2M+1$ steps between $(q_1,\localdata_1)$ and $(q_4, \localdata_4)$ in $\localrun$. Therefore, replacing this part of $\localrun$ with $(q_1,\localdata_1) \step{\atrace} (q_4, \localdata_4)$ yields a shorter "local run" $\localrun': (q_i, \localdata_i) \step{*} (q_f,\localdata_f)$. 
	
	It remains to prove the conditions on the $v$-"input" of $\localrun'$ for each $v$. If suffices to prove the condition for the part between $(q_1, \localdata_1)$ to $(q_4, \localdata_4)$,
	because the rest of $\localrun$ is left untouched. 
	
	Let $(q_m, \localdata_m)$ the "local configuration" after $M$ steps of $\atrace$ from $(q_1, \localdata_1)$; write $\localrun_1$ the local run from $(q_1,\localdata_1)$ to $(q_m, \localdata_m)$ corresponding to the first $M$ steps of $\atrace$ in $\localrun'$, and $\localrun_2$ the "local run" from $(q_m, \localdata_m)$ to $(q_4, \localdata_4)$ corresponding to the last $M$ steps of $\atrace$ in $\localrun'$.  
	
	Let $\aval \in V$. We claim that $\vinput{\aval}{\localrun_{1}}$ is a subword of $\vinput{\aval}{\localrun_a}$ and $\vinput{\aval}{\localrun_{2}}$ is a subword of $\vinput{\aval}{\localrun_b}$. Indeed, in the construction of $\atrace$, the "reception steps" in the first $M$ steps were those of $\localrun_a$ except that some values were replaced with fresh values in $\nats \setminus V$, and similarly with $\localrun_b$ and the last $M$ steps. Overall, this proves that $\vinput{\aval}{\localrun'}$ is a subword of $\vinput{\aval}{\localrun}$ for every $\aval \in V$ and values of $V$ satisfy condition  \ref{item:shorterrun_oldvalues}. 
	
	Let $\aval' \in \nats \setminus V$; $\aval'$ does not appear in $\localrun$. Either $\aval'$ does not appear in $\localrun'$ in which case the desired property is true, or there exists $\aval \in V$ such that $\aval' = \phi(\aval)$. As $\phi$ is injective and $\phi(\Vinit) = \Vinit$ and $\aval'\notin \Vinit$, $\aval \notin \Vinit$. 
	But then $\vinput{\aval'}{\localrun_{1}}$ is a subword of $\vinput{\aval}{\localrun_a}$ and $\vinput{\aval'}{\localrun_{2}}$ is a subword of $\vinput{\aval}{\localrun_b}$. Indeed, in $\localrun_1$, the "reception steps" with value $\phi(\aval)$ correspond to "reception steps" in $\localrun_a$ with value $\aval$, and similarly for $\localrun_2$ and $\localrun_b$. This proves condition \ref{item:shorterrun_anyvalue} for every $\aval' \in \nats \setminus V$.
	Overall, we have proven the existence of a "local run" $\localrun': (q_i,\localdata_i) \step{*} (q_f,\localdata_f)$ that satisfies conditions \ref{item:shorterrun_oldvalues} and \ref{item:shorterrun_anyvalue} and that is strictly shorter that $\localrun$, which proves the "shortening property".
	
	We build a "local run" of length less that $\towerfun(\size{\prot}-r+\size{\Vinit},\regnum)$ as follows. We start with $\localrun^{(0)} := \localrun$ and $V^{(0)}$ the set of values of messages appearing in $\localrun^{(0)}$. For every $k$ such that $\length{\localrun^{(0)}} > \towerfun(\size{\prot}-r+\size{\Vinit},\regnum)$, we apply the "shortening property" on $\localrun^{(k)}$ and $V^{(k)}$ to obtain $\localrun^{(k+1)}$ and define $V^{(k+1)}$ by $V^{(k)} \cup X$ where $X$ is the set of values of messages in $\localrun^{(k+1)}$, which is finite.
	The construction stops when $\length{\localrun^{(k)}} \leq \towerfun(\size{\prot}-r + \size{\Vinit},\regnum)$, which concludes the proof of the lemma. 
\end{proof}

%% file: Figures/fig-illustration-pumping.tex
	\begin{tikzpicture}
		
		
		
		

		
		
		
		
		
		\node at (4,1) [align = center] {Original \\ local run};
		
		\node at (5.5,1.7) {reg $1$};
		\node at (5.5,1.2) {reg $2$};
		\node at (5.5,0.7) {reg $3$};
		\node at (5.5,0.2) {reg $4$};
		\draw (6,0) rectangle (10.9,0.5);
		\draw (6,0.5) rectangle (10.3,1);
		\draw (6,1) rectangle (10.7,1.5);
		\draw (6,1.5) rectangle (11.2,2);
		
		\draw[white,fill=blue!20] (7,0) rectangle (10,2);
		\draw[white,fill=green!20] (11.5,0) rectangle (14.5,2);
		
		\draw[white,fill=blue!30] (7,0) rectangle (7.5,0.5);
		\draw[white,fill=BlueGreen!40] (7.5,0) rectangle (8.8,0.5);
		\node at (8.15,0.2) {$v$};
		\draw[white,fill=blue!45] (8.8,0) rectangle (10,0.5);
		\draw (6,0) rectangle (6.5,0.5);
		\draw (6,0) rectangle (7.5,0.5);
		\draw (6,0) rectangle (8.8,0.5);
		
		\draw[white,fill=green!47] (11.5,0) rectangle (12,0.5);
		\draw[white,fill=BlueGreen!40] (12,0) rectangle (13.3,0.5);
		\node at (12.65,0.2) {$v$};
		\draw[white,fill=green!57] (13.3,0) rectangle (14.5,0.5);
		\draw (6,0) rectangle (12,0.5);
		\draw (6,0) rectangle (13.3,0.5);
		
		\draw[white,fill=blue!50] (7,0.5) rectangle (7.7,1);
		\draw[white,fill=blue!20] (7.7,0.5) rectangle (9.4,1);
		\draw[white,fill=blue!35] (9.4,0.5) rectangle (10,1);
		\draw (6,0.5) rectangle (6.3,1);
		\draw (6,0.5) rectangle (7.7,1);
		\draw (6,0.5) rectangle (9.4,1);
		
		\draw[white,fill=green!70] (11.5,0.5) rectangle (12.2,1);
		\draw[white,fill=green!20] (12.2,0.5) rectangle (13.9,1);
		\draw[white,fill=green!40] (13.9,0.5) rectangle (14.5,1);
		\draw (6,0.5) rectangle (12.2,1);
		\draw (6,0.5) rectangle (13.9,1);
		
		\draw[white,fill=blue!35] (7,1) rectangle (8.1,1.5);
		\draw[white,fill=blue!55] (8.1,1) rectangle (9.5,1.5);
		\draw[white,fill=blue!25] (9.5,1) rectangle (10,1.5);
		\draw (6,1) rectangle (6.7,1.5);
		\draw (6,1) rectangle (8.1,1.5);
		\draw (6,1) rectangle (9.5,1.5);
		
		\draw[white,fill=green!57] (11.5,1) rectangle (12.6,1.5);
		\draw[white,fill=green!60] (12.6,1) rectangle (14,1.5);
		\draw[white,fill=green!25] (14,1) rectangle (14.5,1.5);
		\draw (6,1) rectangle (12.6,1.5);
		\draw (6,1) rectangle (14,1.5);
		\draw (6,1) rectangle (14.7,1.5);
		
		\draw[white,fill=blue!15] (7,1.5) rectangle (7.3,2);
		\draw[white,fill=blue!35] (7.3,1.5) rectangle (8.9,2);
		\draw[white,fill=blue!25] (8.9,1.5) rectangle (10,2);
		\draw (6,1.5) rectangle (6.5,2);
		\draw (6,1.5) rectangle (7.3,2);
		\draw (6,1.5) rectangle (8.9,2);
		\node at (6.8,2.4) [align = center] {$\quotemarks{\enregact}$ actions};
		\draw[->] (6.5,2.2) -- (6.5,2);
		\draw[->] (7.3,2.2) -- (7.3,2);

		\draw[white,fill=green!70] (11.5,1.5) rectangle (11.8,2);
		\draw[white,fill=green!40] (11.8,1.5) rectangle (13.4,2);
		\draw[white,fill=green!15] (13.4,1.5) rectangle (14.5,2);
		\draw (6,1.5) rectangle (11.8,2);
		\draw (6,1.5) rectangle (13.4,2);
		
		\draw[<->] (8.5, 2.2) -- (13,2.2);
		\node[align = center] at (10.75,2.4) {removed part};

		\draw[dashed] (8.5,0) -- (8.5,2);
		\draw[dashed] (13,0) -- (13,2);

		\begin{scope}[xshift = 2cm, yshift = -0.5cm]
		
		\node at (4, -2) [align = center] {Shortened \\ local run};
		\node at (5.5,-1.3) {reg $1$};
		\node at (5.5,-1.8) {reg $2$};
		\node at (5.5,-2.3) {reg $3$};
		\node at (5.5,-2.8) {reg $4$};
		\draw[<->] (7,-0.8) -- (10,-0.8);
		\node at (8.5, -0.6) {$s, \Theta$};
		\draw[white,fill=blue!30] (7,-3) rectangle (7.5,-2.5);
		\draw[white,fill=BlueGreen!40] (7.5,-3) rectangle (8.8,-2.5);

		\node at (8.15,-2.8) {$v$};
		\draw[white,fill=green!57] (8.8,-3) rectangle (10,-2.5);
		
		\draw[white,fill=blue!50] (7,-2.5) rectangle (7.7,-2);
		\draw[white,fill=brown!40] (7.7,-2.5) rectangle (9.4,-2);
		\draw[white,fill=green!40] (9.4,-2.5) rectangle (10,-2);
		
		\draw[white,fill=blue!35] (7,-2) rectangle (8.1,-1.5);
		\draw[white,fill=red!55] (8.1,-2) rectangle (9.5,-1.5);
		\draw[white,fill=green!25] (9.5,-2) rectangle (10,-1.5);
		
		\draw[white,fill=blue!15] (7,-1.5) rectangle (7.3,-1);
		\draw[white,fill=yellow!35] (7.3,-1.5) rectangle (8.9,-1);
		\draw[white,fill=green!15] (8.9,-1.5) rectangle (10,-1);
		
		\node at (8.5, -3.5) {fresh values};
		\draw[->] (8.7,-3.3) -- (8.7,-2.25);
		\draw[->] (8.6,-3.3) -- (8.6,-1.75);
		\draw[->] (8.4,-3.3) -- (8.4,-1.25);
		\draw[dashed] (8.5,-3) -- (8.5,-1);
		
		\draw (6,-3) rectangle (10.5,-2.5);
		\draw (6,-3) rectangle (10.5,-2);
		\draw (6,-3) rectangle (10.5,-1.5);
		\draw (6,-3) rectangle (10.5,-1);
		
		\draw (6,-3) rectangle (6.5,-2.5);
		\draw (6,-3) rectangle (7.5,-2.5);
		\draw (6,-3) rectangle (8.8,-2.5);
	
		\draw (6,-2.5) rectangle (6.3,-2);
		\draw (6,-2.5) rectangle (7.7,-2);
		\draw (6,-2.5) rectangle (9.4,-2);
		
		\draw (6,-2) rectangle (6.7,-1.5);
		\draw (6,-2) rectangle (8.1,-1.5);
		\draw (6,-2) rectangle (9.5,-1.5);
		\draw (6,-2) rectangle (10.2,-1.5);
		
		\draw (6,-1.5) rectangle (6.5,-1);
		\draw (6,-1.5) rectangle (7.3,-1);
		\draw (6,-1.5) rectangle (8.9,-1);
		\end{scope}

		\draw (6,0) rectangle (15,0.5);
		\draw (6,0) rectangle (15,1);
		\draw (6,0) rectangle (15,1.5);
		\draw (6,0) rectangle (15,2);
		
		\draw[<->] (7,-0.2) -- (10,-0.2);
		\draw[<->] (11.5,-0.2) -- (14.5,-0.2);
		
		\node (s1) at (8.5,-0.4) {$s, \Theta$};
		\node (s2) at (13,-0.4) {$s, \Theta$};
	\end{tikzpicture}

%% file: Appendix/Preuves-tree-bounds.tex
\section{Bounding the Size of the Minimal Coverability Witness}
\label{app:proofs_bounds}

In this section, we prove that we can obtain a computable bound on $\tree$, and that the problem is decidable in complexity class $\Fcomplexity{\omega^\omega}$. All bounds provided in this section are independent from $\tree$. 
We start by proving that one can bound the size of a node with respect to the size of its neighbors of higher altitude. 

\begin{restatable}{lemma}{lemBoundSuccessorHeight}
	\label{lem:bound-successor-height}
	Let $\node$ be a node of $\tree$ such that all neighbors of $\node$ of higher altitude have size bounded by $K$.
	Then $\size{\node} \leq (\towerfun(\size{\prot}, r)+2) \, (\size{\messages} \, r \, K + K)$, with $\psi$ the primitive-recursive function defined in Lemma~\ref{lem:short-local-runs}. 
\end{restatable}

\begin{proof}
	Let $u := \localrunlabel{\node}$. For each "initial" value $v$ in $u$, $\node$ has at most $\size{\messages}$ "follower" children, by definition of a "decomposition".
	Because there are at most $r$ "initial" values, this makes at most $\size{\messages} \, r$ "follower" children and each one of them requires at most $K$ messages. If $\node$ is a "follower node" or the root, then the number of messages it must broadcast is then bounded by $\size{\messages} \, r \, K +1$. If $\node$ is a "boss node" and not the root, then its parent
	has higher altitude hence $\size{\bossspec(\node)} \leq K$ and $\node$ must broadcast at most $\size{\messages} \, r \, K + K$ messages in total.
	This bounds by $\size{\messages} \, r \, K + K +1$ the number of messages $u$ needs to broadcast. Lemma~\ref{lem:short-local-runs} then lets us reduce the number of steps between those messages without breaking the "unfolding tree" conditions: the inputs of values that were already in the run can only decrease. For the other values, let $v$ be a value that was not in the run before, its input is a subword of the one of a previous value. Hence $\node$ has a "boss" child whose "specification" covers the "$v$-input" of the new local run. 
	This bounds $\size{u}$ by $(\towerfun(\size{\prot},r)+1) \, (\size{\messages}\, r \, K + K)$, with $\psi$ the function described in Lemma~\ref{lem:short-local-runs}; the bound on $\size{\node}$ follows.
\end{proof}

We now formally prove Lemma~\ref{lem:bound-length-at-height-h}:
\lemBoundLengthHeightH*

\begin{proof}
Let $f_0: k \mapsto ((\towerfun(k,k)+2) (k^2 +2))^{k+1}$. Also, let $N := \size{\prot}$.
First, if $\node$ has altitude $\altmax$, then it has no "follower" child, hence applying Lemma~\ref{lem:bound-successor-height} with $K=0$ bounds $\size{\node}$ by $\towerfun(\size{\prot},r) \leq \towerfun(N,N) \leq f_0(N)$. Let $\node$ be a node with $\altitude{\node} < \altmax$, and suppose that the statement is true for altitudes greater than $\altitude{\node}$. ~

Let $d =~\altmax - ~\altitude{\node}$. We apply Lemma~\ref{lem:bound-successor-height} with $K := f_0(N + d - 1)$: 
\begin{align*}
\size{\node} \leq & (\towerfun(N,N)+2) (N^2 K + K)\\ \leq & (\towerfun(N+d,N+d)+2) f_0(N+d-1) ((N+d)^2 +2)   \\ \leq & f_0(N+d) 
\end{align*}
\end{proof}

\begin{restatable}{lemma}{lemBoundMaxHeight}
	\label{lem:bound-max-height}
	There exists a function $f_1 \in \Ffunction{\omega^{\size{\messages}}}$ such that $\altmax \leq f_1(\size{\prot})$.
\end{restatable}
\begin{proof}
	Let $\altmax$ be the maximal altitude of a node in $\tree$. Consider a branch of $\tree$ reaching altitude $\altmax$; for every $j \in \nset{1}{\altmax}$, let $\node_j$ the first node of this branch to reach altitude $j$. For every $j \geq 1$, $\node_j$ is a "follower nodes" as otherwise its predecessor in the branch would have altitude greater than $j$ hence the branch must have crossed altitude $j$ before.
	
	We will bound $\altmax$ using the "Length function theorem". For simplicity, we will encode the $\followmessagespec$ part using a fresh character added to our alphabet. This is why we obtain a function in $\Ffunction{w^{\size{\messages}}}$ and not $\Ffunction{w^{\size{\messages}-1}}$; in fact, one could obtain the latter bound using Theorem 5.3. of~\cite{SchmitzS2011upperHigman}, but the proof would be more involved.

	Let $\# \notin \messages$ be a fresh letter. For all $i \in \nset{1}{\altmax}$ let $\node'_i = \node_{\altmax - i+1}$ and $w_i = \followlabelword{\node'_{i}} \cdot \# \cdot \followlabelmessage{\node'_{i}} \in (\messages \cup \set{\#})^*$.
	
	We cannot have $w_{i}\subword w_{j}$ for $i< j$: indeed, this would imply $\followlabelword{\node'_{i}} \subword \followlabelmessage{\node'_{j}}$ and $\followlabelmessage{\node'_{i}} = \followlabelword{\node'_{j}}$, which is a contradiction according to Lemma~\ref{lem:shortening-branches}.
	As a result, the sequence $(w_i)_{1\leq i \leq \altmax}$ is a "bad sequence" over the alphabet $\messages\cup\set{\#}$.
	Furthermore, by Lemma~\ref{lem:bound-length-at-height-h}, for all $i$, $\size{\node_{i}} \leq f_0(\size{\prot} + i)$ therefore $\size{w_i} \leq f_0^{(i)}(\size{\prot})$.
	Because $f_0$ is primitive-recursive, we can apply Theorem~\ref{thm:lengthfcttheorem}: there exists a function $f_1 \in \Ffunction{\omega^{\size{\messages}}}$ such that the sequence $(w_i)_{i \in \nset{1}{\altmax}}$ is of length at most $f_1(\size{\prot})$, hence $\altmax \leq f_1(\size{\prot})$. 
\end{proof}
\begin{figure}[h]
	\input{Figures/max-height-bound}
	\caption{Illustration of the proof of Lemma~\ref{lem:bound-max-height}.}
	\label{fig:max-height-bound}
\end{figure}

The bound on the maximal altitude, along with Lemma~\ref{lem:bound-length-at-height-h}, bounds the size of the root of $\tree$ by $f_0(\size{\prot} + f_1(\size{\prot}))$. Let $\altmin := \min_{\node \in \tree} \altitude{\node} \leq 0$ the minimal altitude in $\tree$. By a similar argument but with "boss nodes", we can bound $|\altmin|$:

\begin{restatable}{lemma}{lemBoundMinHeight}
	\label{lem:bound-min-height}
	There exists a function $f_2 \in \Ffunction{\omega^{\size{\messages}}}$ such that $|\altmin| \leq f_2(\size{\prot})$.
\end{restatable}

\begin{proof}
	Consider a branch of $\tree$ reaching $\altmin$, let $\node_1, \ldots, \node_{|\altmin|}$ be nodes of $\tree$ such that $\node_i$ is the first node of the branch to reach altitude $-i$. All these nodes are "boss nodes" because a "follower node" has higher altitude than its parent. 
	For all $i \in \nset{0}{\altmin}$ let $\node'_i = \node_{|\altmin|-i}$ and $w_i = \bosslabel{\node'_{i}} \in \messages^*$.
	By Lemma~\ref{lem:shortening-branches}, $(w_i)_{i \in \nset{1}{\altmin}}$ is a "bad sequence".
	By Lemmas~\ref{lem:bound-length-at-height-h} and \ref{lem:bound-max-height}, for all $i$, we have $\size{w_i} \leq f_0(\size{\prot}+f_1(\size{\prot}) +i)$. Let $N := \size{\prot}+f_1(\size{\prot})+1$; we have $\size{w_i} \leq f_0^{(i)}(N)$.  
	
	Because $f_0$ is primitive recursive, we can apply Theorem~\ref{thm:lengthfcttheorem}: we obtain $h \in \Ffunction{\omega^{\size{\messages}-1}}$ such that $\altmin \leq h(N)$. Let $f_2: n \mapsto h(n+ f_1(n) + 1)$; we have $\altmin \leq f_2(\size{\prot})$. Moreover, $f_2$ is a function of $\Ffunction{\omega^{\size{\messages}}}$ because $f_1, h \in \Ffunction{\omega^{\size{\messages}}}$ and $\Ffunction{\omega^{\size{\messages}}}$ is closed under composition.
\end{proof}

With the bounds on the maximal and minimal altitudes, Lemma~\ref{lem:bound-length-at-height-h} yields a bound on the size of all nodes in the tree.

\begin{lemma}
	\label{lem:bound-node-size}
	There exists $f_3 \in \Ffunction{\omega^{\size{\messages}}}$ such that, for all $\node$ in $\tree$, $\size{\node} \leq f_3(\size{\prot})$. 
\end{lemma}
\begin{proof}
	Using Lemmas~\ref{lem:bound-length-at-height-h}, \ref{lem:bound-max-height} and \ref{lem:bound-min-height}, we obtain that, for every node $\node$, the value $\altmax - \altitude{\node}$ is at most $f_1(\size{\prot}) + f_2(\size{\prot})$. Applying Lemma~\ref{lem:bound-length-at-height-h} proves that, for all $\node$ in $\tree$, $\size{\node} \leq f_0(f_1(\size{\prot}) + f_2(\size{\prot})) \leq f_3(\size{\prot})$ where $f_3: n \mapsto f_0(f_1(n) + f_2(n))$; $f_3$ is in $\Ffunction{\omega^{\size{\messages}}}$ as $\Ffunction{\omega^{\size{\messages}}}$ contains all primitive recursive functions and is closed under composition. 
\end{proof}

Finally, as we now have a bound on all nodes, we can easily bound the length of branches and the number of children of all nodes to get abound on the total "size@@tree" of the tree.

\begin{restatable}{proposition}{PropBoundTreeSize}
	\label{prop:bound-tree-size}
	There exists a function $f$ of class $\Ffunction{\omega^{\size{\messages}}}$ s.t. $\size{\tree} \leq f(\size{\prot})$.
\end{restatable}

\begin{proof}
	Recall that the "size@@tree" of $\tree$ is defined as the sum of the sizes of its nodes. 

	Let $N = \size{\prot}$. 
	By Lemma~\ref{lem:bound-node-size}, for each $\node$ in $\tree$ we have $\size{\node} \leq f_3(N)$, hence the "local run" of $\node$ contains at most $f_3(N)$ values. By minimality of $\tree$, each value requires at most $1$ "boss" child and $\size{\messages}$ "follower" children hence $\node$ has at most $(\size{\messages}+1) f_3(N) \leq N f_3(N)$ children. Moreover, there are less than $f_4(N) := N^{f_3(N)}$ possible specifications in $\tree$ hence each branch of $\tree$ is of length less than $f_4(N)$ (by minimality, a branch does not have twice the same specification). This bounds the total number of nodes in $\tree$ by $f_5(N) := (N f_3(N))^{f_4(N) +1}$ hence we obtain $\size{\tree} \leq f(N)$ where $f: n \mapsto f_3(n) f_5(n)$. Because $\Ffunction{\omega^{\size{\messages}}}$ is closed under composition with primitive recursive functions, we have $f \in \Ffunction{\omega^{\size{\messages}}}$.
\end{proof}

%% file: Figures/max-height-bound.tex
\begin{tikzpicture}
	
	\draw[dotted] (0,0) -- (4,0);
	\draw[dotted] (0,0.5) -- (4,0.5);
	\draw[dotted] (0,1) -- (4,1);
	\draw[dotted] (0,1.5) -- (4,1.5);
	\draw[dotted] (0,2) -- (4,2);
	\draw[dotted] (0,2.5) -- (4,2.5);
	\draw[dotted] (0,3) -- (4,3);
	
	\node at (-0.2,0) {$0$};
	\node at (-0.2,0.5) {$1$};
	\node at (-0.2,1) {$2$};
	\node at (-0.2,1.5) {$3$};
	\node at (-0.2,2) {$4$};
	\node at (-0.2,2.5) {$5$};
	\node at (-0.2,3) {$6$};
	
	\draw[very thick] (0.4,0) -- (0.6, 0.5) -- (0.8, 1) -- (1, 0.5) -- (1.2, 1) -- (1.3, 1.5) -- (1.7, 2) -- (1.8, 1.5) -- (2.1, 1) -- (2.3, 1.5) -- (2.4, 2) -- (2.8, 2.5);
	\draw (1, 0.5) -- (1.4, 0) -- (1.6, 0.5) -- (2.4, 1) -- (2.6, 1.5);
	\draw (1.6, 0.5) -- (2.2, 0) -- (2.5, 0.5);
	\draw (0.8, 1) -- (0.7, 1.5) -- (0.4, 1) -- (0.5, 1.5);
	\draw (0.7, 1.5) -- (0.9, 1) -- (1.1, 1.5) -- (0.8, 2);
	\draw (2.4, 2) -- (2.1, 1.5) -- (1.9, 2) -- (1.6, 2.5);
	\draw (2.4, 2) -- (2.1, 1.5) -- (1.9, 2) -- (1.6, 2.5);
	\draw (2.4, 2) -- (2.1, 1.5) -- (1.9, 2) -- (1.6, 2.5);
	\draw (1.9, 2) -- (2.2, 2.5);
	\draw (2.4, 2) -- (3, 1.5);
	\draw (1.9, 2) -- (2.2, 2.5);
	\draw (0.4, 0) -- (0.3, 0.5);
	
	\draw[->, >=triangle 45] (4.2, 0) -- (4.2, 3);

	\node (H) at (4.7, 3) {$\mathbf{alt}$};
	
	\draw[->] (0.4, -0.1) -- (0.8, -0.5);
	
	\node (R) at (2,-0.5) {root ($\mathbf{alt} =0$)};
	\node (AM) at (3,2.7) {$\altmax$};
	\draw[red, fill=red] (0.4, 0) circle (0.05);
	\draw[red, fill=red] (0.6, 0.5) circle (0.05);
	\draw[red, fill=red] (0.8, 1) circle (0.05);
	\draw[red, fill=red] (1.3, 1.5) circle (0.05);
	\draw[red, fill=red] (1.7, 2) circle (0.05);
	\draw[red, fill=red] (2.8, 2.5) circle (0.05);

	\node[red] at (0.7, 0.1) {\small $w_5$};
	\node[red] at (0.3, 0.6) {\small $w_4$};
	\node[red] at (0.5, 0.8) {\small $w_3$};
	\node[red] at (1.5, 1.3) {\small $w_2$};
	\node[red] at (1.3, 1.9) {\small $w_1$};
	\node[red] at (2.9, 2.3) {\small $w_0$};
	
	\node[draw, color=white, text width=6cm] at (8.8, 2.2) {\color{black} For all $i$, let $w_i = \followlabelword{\node'_i}\#\followlabelmessage{\node'_i}$ where $\node'_i$ is the first "follower node" at altitude $\altmax - i+1$.};
	
	\node[draw, color=white, text width=6cm] at (8.8, 0.5) {\color{black} By Lemma~\ref{lem:bound-length-at-height-h} we can bound the length of $w_i$ by $f_0(\size{\prot}+i)$, which lets us to bound the number of red dots using Lemma~\ref{lem:shortening-branches} and the "Length function theorem".};
%
\end{tikzpicture}

%% file: Appendix/decidability.tex
\section{Proof of Theorem~\ref{thm:decidable-cover}}
\label{app:decidable_cover}
We now prove the main result of this paper:

\decidablecover*

\begin{proof}
	The lower bound is given by the reduction from "lossy channel system" reachability in Proposition~\ref{prop:reduction-LCS}.
	
	For the upper bound, let $(\prot,q_f)$ be an instance of \COVER. By Propositions~\ref{prop:trees-sound-complete} and \ref{prop:bound-tree-size}, $(\prot,q_f)$ is positive if and only if it has a "coverability witness" of size bounded by $f(\size{\prot})$ where $f \in \Ffunction{\omega^{\size{\messages}}}$. Up to renaming agents and values, we can moreover assume that all agents and values appearing in this "coverability witness" are bounded by $f(\size{\prot})$, which bounds the size of the description of such a "coverability witness" by a polynomial in $f(\size{\prot})$.
	An algorithm for \COVER consists in enumerating all such descriptions and accepting if one finds a "coverability witness". 
	This can all be done in time exponential in $f(\size{\prot})$, thus this algorithm terminates in time $f'(\size{\prot})$ where $f' \in \Ffunction{\omega^{\size{\messages}}}$. This proves that 
	\COVER lies in complexity class $\Fcomplexity{\omega^\omega}$.
\end{proof}

%% file: Appendix/Preuve-Target.tex
\section{Proof of Proposition~\ref{prop:target-undec}}
\label{app:target}

\propTargetUndecidable*

\begin{proof}
	We present a reduction from the halting problem for Minsky machines to \COVER for "signature BNRA".
	
	A Minsky Machine with two counters is a tuple $M= (\Loc, \Delta, \Cpt, \ell_0, \ell_f)$ where $\Loc$ is a finite set of locations, $\Cpt = \{\cpt_1, \cpt_2\}$ is a set of two counters, $\Delta \subseteq \Loc \times \set{\dec{\cpt},\inc{\cpt},\testz{\cpt} \mid \cpt \in \Cpt} \times \Loc$ is a finite set of transitions, $\ell_0 \in \Loc$ is an initial location and $\ell_f \in \Loc$ is a final location. A \emph{configuration} of a Minsky machine is a tuple $(\ell, v_1, v_2) \in \Loc \times \nats \times \nats$ where $v_1$ (resp. $v_2$) stands for the value of the counter $\cpt_1$ (resp. $\cpt_2$). 
	We write $(\ell, v_1, v_2) \stepMM{} (\ell', v'_1, v'_2)$ if there is $\delta \in \Delta$ such that:
	\begin{itemize}
		\item $\delta = (\ell, \inc{\cpt_i}, \ell')$ and $v'_i = v_i+1$, $v_{3-i} = v'_{3-i}$;
		\item $\delta = (\ell, \dec{\cpt_i}, \ell')$ and $v'_i = v_i-1$, $v_{3-i} = v'_{3-i}$;
		\item $\delta = (\ell, \testz{\cpt_i}, \ell')$ and $v'_i = v_i = 0$, $v_{3-i} = v'_{3-i}$.
	\end{itemize}
	An execution of the machine is a sequence $(\ell_1, v_1^{(1)}, v_2^{(1)}) \stepMM{} (\ell_2, v_1^{(2)}, v_2^{(2)}) \stepMM{} \dots \stepMM{} (\ell_k, v_1^{(k)}, v_2^{(k)})$. 
	The halting problem asks whether $\ell_f$ is coverable. This problem is well-known to be undecidable \cite{minsky}.

	\begin{figure}
			\centering
			\resizebox{.99\linewidth}{!}{
			\input{Figures/target-init}
			}
			\caption{Partial depiction of the protocol built in Proposition~\ref{prop:target-undec}. Only one transition, which is $(\ell_0, \delta, \ell_1)$, is represented in the $\prot_{\mathsf{loc}}$ part above; similarly, only one increment and one decrement transitions are depicted in the $\prot_{\mathsf{count}}$ part below. The rebroadcast loops rebroadcast all transitions acting on $\cpt_2$ and all acknowledgements; the one on $(\cpt_1,0)$ also rebroadcasts all transitions with $\dec{\cpt_1}$ and with zero-tests, and the one on $(\cpt_1,1)$ rebroadcasts all transitions with $\inc{\cpt_1}$.}
			\label{fig:target-init}
	\end{figure}
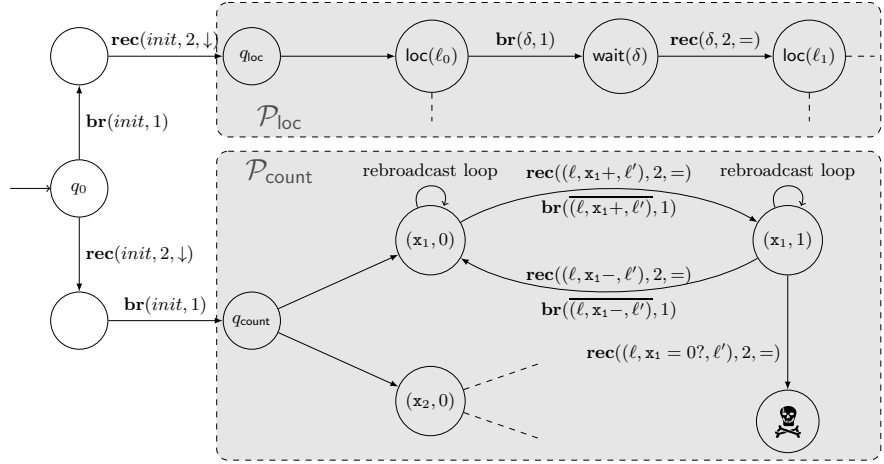
	
	Fix a Minsky Machine $M= (\Loc, \Delta, \Cpt, \ell_0, \ell_f)$. We build a "signature protocol" $\prot$ with a state $q_f$ such that $(\prot, q_f)$ is a positive instance of \TARGET if and only if $\ell_f$ is coverable in $M$. 
	The protocol is represented in Figure~\ref{fig:target-init}. 
	As in Proposition~\ref{prop:reduction-LCS}, in an initial phase, each agent picks a predecessor by storing its identifier and only listens to this predecessor afterwards. We call \emph{cycle} a sequence of agents $a_1, a_2, \ldots, a_n = a_1$ where agent $a_i$ is the predecessor of $a_{i+1}$ for all $i <n$. 
	As all agents have to reach the end state, they must all pick a predecessor.
	As there are finitely many agents in a run, a cycle will necessarily be formed in any run satisfying the \TARGET requirement.
	
	The rest of the construction aims at faithfully simulating the machine in a cycle: Agents in $\prot_{\mathsf{loc}}$ sends a sequence of instructions and waits after each one for a confirmation that it was executed. Agents in $\prot_{\mathsf{count}}$ simulate counter values. The messages circulating in a cycle contain either a transition $\delta \in \Delta$ or an acknowledgment $\overline{\delta}$ with $\delta \in \Delta$. 
	An agent $a$ in $\prot_{\mathsf{count}}$ first picks a counter $\cpt_i$ it simulates, and goes to state $(\cpt_i,0)$. If $a$ is in ($\cpt_i,0)$ and receives $\delta$ corresponding to an increment of $\cpt_i$, it goes to $(\cpt_i,1)$ and broadcast an acknowledgment $\overline{\delta}$, and conversely for decrements. If $\delta$ is a zero-test $\testz{\cpt_i}$ and $a$ is on state $(\cpt_i,1)$ then it stops, making the whole cycle fail. Otherwise it broadcasts the acknowledgment $\overline{\delta}$. Other messages are rebroadcast as such. 

	An agent $a$ in $\prot_{\mathsf{loc}}$ starts in state $\mathsf{loc}(\ell_0)$. When in state $\mathsf{loc}(\ell)$, it picks and broadcasts a transition $\delta = (\ell, \anact, \ell') \in \Delta$, waits for the acknowledgment $\overline{\delta}$ and goes to $\mathsf{loc}(\ell')$. In the case where $\delta$ is a zero-test, we have $\overline{\delta} = \delta$: there is no need for a distinct acknowledgment because there is no action to perform (if the test fails then no message is transmitted). When in $\mathsf{loc}(\ell_f)$, $a$ broadcasts a special message $\mathsf{end}$ that propagates in the cycle and makes everyone go to $q_f$. When it receives itself the message $end$, it goes to $q_f$. 
	
	It is quite easy to see that, if $\ell_f$ can be covered in $M$, one can build a run of $\prot$ where all agents end in $q_f$. Let $N$ the highest counter value in the execution of $M$ covering $q_f$. The run of $\prot$ first puts all its agents in the same cycle; exactly one agent $a_\mathsf{lead}$ goes in $\prot_{\mathsf{loc}}$ and $2N$ agents go in $\prot_{\mathsf{count}}$; half of these simulate $\cpt_1$ and half $\cpt_2$, so that the largest counter value is never exceeded. It then suffices to faithfully simulate the execution of $M$: $a_\mathsf{lead}$ selects the corresponding sequence of transitions, their effect is always applied as we have enough agents simulating each counter. After each round the number of agents in state $(\cpt_i,1)$ is exactly the value of $\cpt_i$ at this point in the run of the machine, hence zero-tests never cause failure. In the end $a_\mathsf{lead}$ reaches $\mathsf{loc}(\ell_f)$ and broadcasts $\mathsf{end}$, allowing every agent in $\prot_{\mathsf{count}}$ to get to $q_f$.

	For the converse implication, suppose that we have a run $\run$ of $\prot$ where all agents end in $q_f$. As mentioned before, there must be a cycle of agents $a_1, \dots, a_n$ in this run. Observe that all agents alternate between broadcasts and receptions, so that to reach $q_f$ they must all have made the same number of broadcasts and receptions. This implies that no message was lost along the cycle.
	
	Note that there may be several agents in $\prot_{\mathsf{loc}}$ along the cycle; however, they must all broadcast exactly the same sequence of transitions, otherwise one of them would lack an acknowledgment and would not get to $q_f$. Let $a$ be the agent that first reaches $\mathsf{loc}(\ell_f)$ and $a'$ the first agent in $\prot_{\mathsf{loc}}$ after $a$ in the cycle; there are only agents in $\prot_{\mathsf{count}}$ between $a$ and $a'$ in the cycle, we call these agents \emph{intermediate agents}. The intermediate agents faithfully encode the two counters and all decrements and zero-tests pass, otherwise $a'$ would lack an acknowledgment. Therefore, the sequence of transitions of $a$ defines an execution of $M$ that covers $\ell_f$, which concludes the proof.  
\end{proof}

%% file: Figures/target-init.tex
\begin{tikzpicture}[xscale=1,node distance=2cm,auto, yscale = 0.5]
	\tikzstyle{initial}= [initial by arrow,initial text=,initial
	distance=.7cm]
	\tikzset{every state/.style =  {minimum size = 1cm}}
	
	\node[state,initial] (0) at (0,0) {$q_0$};
	\node[state] [above = of 0, yshift=-20, xshift=0] (x) {};
	\node[state] [below= of 0, yshift=20, xshift=0] (y) {};
	\node[state] [right = of x] (1) {$q_{\mathsf{loc}}$};
	\node[state] [right = of y] (2) {$q_{\mathsf{count}}$};
	\node[state] [right =  of  1] (loc0) {$\mathsf{loc}(\ell_0)$};
	\node[state] [right = of loc0] (trans) {$\mathsf{wait}(\delta)$};
	\node[state] [right = of trans] (loc1) {$\mathsf{loc}(\ell_1)$}; 
	\node[state] [right = of 2, yshift = 40] (counter10) {$(\cpt_1,0)$};
	\node[state] [right = of 2, yshift = -40] (counter20) {$(\cpt_2,0)$};
	\node[state, draw = none, right = of counter10] (c1aux) {};
	\node[state, draw = none, right = of counter20] (c2aux) {};
	\node[state] [right = of c1aux] (counter11) {$(\cpt_1,1)$};
	\node[state] [below = of counter11] (dead) {\Large $\skull$};
	
	\begin{scope}[on background layer]
	\draw[rounded corners=2mm,dashed,fill=black!10] (2.4,6.5) -- (14,6.5) -- (14,1.8) -- (2.4,1.8)-- cycle;
	\draw[rounded corners=2mm,dashed,fill=black!10] (2.4,1.3) -- (14,1.3) -- (14,-9.5) -- (2.4,-9.5)-- cycle;
	\end{scope}
	\node at (3.5, 2.5) {\color{black!70} \Large $\prot_{\mathsf{loc}}$};
	\node at (3.5, 0.6) {\color{black!70} \Large  $\prot_{\mathsf{count}}$};
	
;	\path[-latex, bend left=0] 	
	(0) edge node[right, xshift=0,yshift = 0] {$\br{init}{1}$} (x)
	;
	\path[-latex, bend right=0] 	
	(0) edge node[right, xshift = 0, yshift = 0] {$\rec{init}{2}{\enregact}$} (y)
	;
	\path[-latex] 	
	(x) edge node {$\rec{init}{2}{\enregact}$} (1)
	(y) edge node {$\br{init}{1}$} (2)
	;
	\path[-latex]
	(counter10) edge[bend left = 50] node[above] {$\rec{(\ell, \inc{\cpt_1}, \ell')}{2}{\eqtestact}$} node[below] {$\br{\overline{(\ell, \inc{\cpt_1}, \ell')}}{1}$} (counter11)
	(counter11) edge[bend left = 50] node[above] {$\rec{(\ell, \dec{\cpt_1}, \ell')}{2}{\eqtestact}$} node[below] {$\br{\overline{(\ell, \dec{\cpt_1}, \ell')}}{1}$} (counter10)
	(2) edge (counter10)
	(2) edge (counter20)
	(1) edge (loc0)
	(loc0) edge node[above] {$\br{\delta}{1}$} (trans)
	(trans) edge node[above] {$\rec{\delta}{2}{\eqtestact}$} (loc1)
	(counter10) edge[loop above] node[above, align = center] {rebroadcast loop} ()
	(counter11) edge node[left, yshift = -10] {$\rec{(\ell, \testz{\cpt_1}, \ell')}{2}{\eqtestact}$} (dead)
	(counter11) edge[loop above] node[above] {rebroadcast loop} (); 

	\node[ right = of counter20, yshift = 0.7cm, xshift = -0.7cm] (dash1) {};
	\node[ right = of counter20, yshift = -0.7cm, xshift = -0.7cm] (dash2) {};
	\node [right = of loc1, xshift = -1.5cm] (dash3) {};
	\node [below = of loc1, yshift = 1.5cm] (dash4) {};
	\node [below = of loc1, yshift = 1.5cm] (dash4) {};
	\node [below = of loc0, yshift = 1.5cm] (dash5) {};
	\path[dashed, draw] 
	(counter20) -- (dash1)
	(counter20) -- (dash2)
	(loc1) -- (dash3)
	(loc0) -- (dash5)
	(loc1) -- (dash4);

\end{tikzpicture}

%% file: Appendix/one-reg-proofs-v3.tex
\section{Proofs of Section \ref{sec:cover-1BNRA}}
\label{app:cover-one-reg}

In this section the register argument in receptions and broadcasts is always $1$, hence we remove it.
Our new set of operations is 
\[
Op^{\messages} = \set{\brone{\amessage}, \recone{\amessage}{\dummyact}, \recone{\amessage}{\enregact}, \recone{\amessage}{\eqtestact}, \recone{\amessage}{\diseqtestact} \mid \amessage \in \messages}.
\]
Also, given a "configuration" $\config$, we write $\data{\config}(a)$ for $\data{\config}(a,1)$. 

Given a "run" $\run$, we write $\intro*\agentsof{\run}$ for its set of agents, and we define the set of states appearing in it $\intro*\statesin{\run}:= \set{q \in Q \mid \exists i, \exists a, \st{\config_i}(a) = q}$  as well as its set of values $\intro*\valsof{\run} := \set{\aval \in \nats \mid \exists i,j, a, \data{\config_i}(a,j) = \aval}$.  

First, to simplify the proofs, we eliminate reception transitions with action $\quotemarks{\ne}$. This is feasible as we can execute several copies of a "run" in parallel (with distinct values) so that every broadcast is made in each copy with a different value. Hence if a agent receives a message, it can always receive it with a value different from its own, making disequality tests useless. We can thus replace them with receptions with $\quotemarks{\dummyact}$.

\subsection{Removing Disequality Tests}
\label{sec:one-diseq-tests}

We start by formalising the intuition that a configuration contains more agents than another one up to renaming.

\begin{definition}
	We define a preorder over the set of configurations as follows: $\config \intro*\lessthan \config'$ if there exists an injective function $\pi: \agentsof{\config} \rightarrow \agentsof{\config'}$ such that, for all $a \in \agentsof{\config}$, $\config(a) = \config'(\pi(a))$. 
\end{definition}


\begin{restatable}{lemma}{lemRemoveDiseq}
	\label{lem:removing_diseq_tests}
	Let $(\prot, q_f)$ an instance of the "coverability problem". This instance is positive if and only if $(\tilde{\prot}, q_f)$ is positive, where $\tilde{\prot}$ is equal to $\prot$ where every disequality test $\quotemarks{\diseqtestact}$ is replaced by dummy action $\quotemarks{\dummyact}$.  
\end{restatable}

\begin{proof}
	First, if $(\prot, q_f)$ is positive then so is $(\tilde{\prot}, q_f)$, as one can easily lift any "initial run" in $\prot$ to an equivalent "initial run" in $\tilde{\prot}$ (transitions are less guarded  in $\tilde{\prot}$ that in $\prot$). 
	
	Suppose now that $(\tilde{\prot}, q_f)$ is a positive instance of the "coverability problem". There exists an "initial run" $\tilde{\run}: \tilde{\config}_0 \step{*} \tilde{\config}$ in $\tilde{\prot}$ that covers $q_f$. We prove by induction on the length of $\tilde{\run}$ that there exists an "initial run" $\run$ reaching a configuration $\config$ such that $\tilde{\config} \lessthan \config$ (note that if $\tilde{\config}$ covers a state, then so does $\config$). 
	
	If $\tilde{\config} = \tilde{\config}_0$ then $\run = \tilde{\run}$ suffices. Suppose that $\tilde{\run}$ has length $k \geq 1$, and that the result if true for "runs" of length $k-1$. Decompose $\tilde{\run}$ into $\tilde{\run}_{k-1}: \tilde{\config}_0 \step{*} \tilde{\config}_{k-1}$ of length $k-1$ and a final step $\tilde{\config}_{k-1} \step{} \tilde{\config}_k$. 
	By induction hypothesis, there exists $\run_{k-1}: \config_0 \step{*} \config_{k-1}$ such that $\tilde{\config}_{k-1} \lessthan \config_{k-1}$: there exists an injective function $\pi : \tilde{\agents} \rightarrow \agents$
	such that, for all $a \in \tilde{\agents}$, $\tilde{\config}_{k-1}(a) = \config_{k-1}(\pi(a))$, where $\tilde{\agents} := \agentsof{\tilde{\run}}$ and $\agents := \agentsof{\run}$. If $\tilde{\config}_{k-1} \step{} \tilde{\config}_k$ involves no reception transition from $\tilde{\prot}$ whose corresponding transition in $\prot$ has action $\quotemarks{\diseqtestact}$, then we directly lift this step into a step appended at the end of $\run_{k-1}$ (making $\pi(a)$ take a transition whenever $a$ does so in $\tilde{\config}_{k-1} \step{} \tilde{\config}_k$). 
	Otherwise, write $\tilde{\agents}_{\diseqtestact}$ the subset of $\tilde{\agents}$ corresponding to agents taking in $\tilde{\config}_{k-1} \step{} \tilde{\config}_k$ a reception transition from $\tilde{\prot}$ whose corresponding transition in $\prot$ has action $\quotemarks{\diseqtestact}$ . 
	Write $(q, \brone{m}, q') \in \transitions$ the broadcast transition used in this step.  Using the "copycat principle", we add to $\config_{k-1}$ a fresh agent $a_{\mathsf{new}}$ with state $q$ and a register value that does not appear in $\config_{k-1}$. 
	We first mimic this broadcast step at the end of $\run_{k-1}$, making any agent $\pi(a) \in \pi(\tilde{\agents} \setminus \tilde{\agents}_{\diseqtestact})$ take the transition that $a$ takes in $\tilde{\config}_{k-1} \step{} \tilde{\config}_k$. We then add a new step where $a_{\mathsf{new}}$ broadcasts using transition $(q, \brone{m}, q')$, and every agent $\pi(a) \in \pi(\tilde{\agents}_{\diseqtestact})$ takes the transition corresponding to the transition taken by $a$ in $\tilde{\config}_{k-1} \step{} \tilde{\config}_k$. Such a transition is a reception with action $\quotemarks{\diseqtestact}$ in $\prot$; however, because $a_{\mathsf{new}}$ does not share its register value with any process from $\tilde{\agents}$, all disequality conditions are satisfied and this step is valid. In the end, every agent $\pi(a) \in \pi(\tilde{\agents})$ has taken the transition in $\prot$ corresponding to the one $a$ took in $\tilde{\prot}$ in step $\tilde{\config}_{k-1} \step{} \tilde{\config}_k$, hence the configuration $\config_k$ reached by the constructed "run" is such that $\tilde{\config}_k \lessthan \config_k$. 
\end{proof}

\subsection{Abstraction}

We now define our abstraction. We formalize the definition of a "gang":

\begin{definition}
	Let $(Q,\transitions, q_0)$ be a protocol.
	
	A ""gang"" is a pair $\gang = (\boss, \clique) \in (Q \cup \set{\noboss}) \times \powerset{Q}$. 
	The element $\boss$ is the \reintro{boss} and the set $\clique$ is the \intro{clique} of the "gang". 
	
	Let $\run = \config_0 \step{} \config_1 \step{} \cdots \step{} \config_k$ be a "run" and $\aval \in \valsof{\run}$. The "gang" of value $\aval$ in $\run$, written $\intro*\gangof{\aval}{\run}$, is the "gang" $(\intro*\bossof{\aval}{\run}, \intro*\cliqueof{\aval}{\run})$ such that, 
	\begin{itemize}
		\item if there exists $a_0 \in \agentsof{\run}$ such that, 
		for every 
		$i \in \nset{0}{k}$, 
		$\data{\config_i}(a_0) = \aval$ then $\bossof{\aval}{\run} := \st{\config_k}(a_0)$, otherwise $\bossof{\aval}{\run} := \noboss$, 
		\item  $\cliqueof{\aval}{\run} := \set{q \in Q \mid \exists i \leq k, \exists a \in \agents\setminus \set{a_0}, \config_i(a) = (q,\aval)}$
	\end{itemize}
\end{definition}

We define "abstract runs" as follows:

\begin{definition}
	\label{def:abstract-configuration}
	An ""abstract configuration"" over $\agents$ is a tuple of $2^Q \times \gangset$ where $\gangset$ designates the set of all "gangs". We write $\aconfigs{\agents}$ the set of "abstract configurations" over $\agents$ and $\allaconfigs := \bigcup_{\agents \subseteq \nats \text{ finite }}\aconfigs{\agents}$ the set of all "abstract configurations". 
	
	Given a set of states $S \subseteq Q$, a "message type" $m$ and a set of operations $O$, we define $\cliquesucc{S}{m}{A} = \set{s' \in Q \mid \exists s \in S, a\in O, (s, \recone{m}{a}, s') \in \Delta}$.
	
	Given two "abstract configurations" $\aconfig = (\covset, \boss, \clique)$ and $\aconfig' = (\covset', \boss', \clique')$, there is an ""abstract step"" from $\aconfig$ to $\aconfig'$, denoted $\aconfig \step{} \aconfig'$, when $\clique' \subseteq \covset'$, $\boss' \in \covset' \cup \set{\noboss}$ and one of the following cases is satisfied.
	\begin{enumerate}
		\item \emph{""Broadcast from clique"":} There exists $(\statebr, \brone{m}, \statebr') \in \transitions$ such that:
		\begin{enumerate}[label = (\arabic*)] 
			
			\item\label{item:broadcast_from_clique_boss} Either $\boss = \boss'$ or there exists $(\boss, \recone{\amessage}{\anact}, \boss') \in \transitions$ for some action $\anact$.
			
			\item\label{item:broadcast_from_clique_clique}$ \clique' = \clique \cup \set{\statebr'}\cup \cliquesucc{\clique}{m}{\set{\eqtestact, \dummyact, \enregact}} \cup \cliquesucc{\covset}{m}{\set{\enregact}}$.
		\end{enumerate}
		
		\item \emph{""Broadcast from boss"":} there exists $\amessage \in \messages$ such that $(\boss, \brone{m}, \boss') \in \transitions$
		\begin{enumerate}[label = (\arabic*)]
			
			\item\label{item:broadcast_from_boss_boss} $\boss, \boss' \ne \noboss$ (technically implied by the existence of $(\boss, \brone{m}, \boss')$ but written here to match other cases)
			
			\item\label{item:broadcast_from_boss_clique} $\clique'= \clique \cup \cliquesucc{\clique}{m}{\set{\eqtestact, \dummyact, \enregact}} \cup \cliquesucc{\covset}{m}{\set{\enregact}}$.
		\end{enumerate}

		\item \emph{""External broadcast"":} There exists $(\statebr, \brone{m}, \statebr') \in \transitions$ such that
		\begin{enumerate}[label = (\arabic*)]
			
			\item\label{item:external_broadcast_boss}Either $\boss = \boss'$ or:
			\begin{itemize} 
				\item $\boss' \ne \noboss$ and there exists $(\boss, \recone{\amessage}{\dummyact}, \boss') \in \transitions$, or
				\item $\boss' = \noboss$ and there exists $(\boss, \recone{\amessage}{\enregact}, \boss') \in \transitions$.
			\end{itemize}
			
			\item\label{item:external_broadcast_clique} $\clique'= \clique \cup \cliquesucc{\clique}{m}{\set{\dummyact}}$.
			
		\end{enumerate}
		\item \emph{""Gang reset"":} $\covset' = \covset \cup \clique \cup \set{\boss}$, $\clique' = \emptyset$ and $\boss'= q_0$
	\end{enumerate}

	Given a concrete "run" $\run: \config_0 \step{*} \config_k$, we write \AP  $\intro*\absproj{\aval}{\run}$ for the "abstract configuration" $(\covset, \gangof{\aval}{\run})$ where $\covset$ is the set of all states appearing in $\run$. 
	
	The \emph{initial abstract configuration} is $\aconfiginit := (\set{q_0}, q_0, \emptyset)$. 
	As in the concrete case, an ""abstract run"" is a sequence $\arun = \aconfig_0, \dots, \aconfig_k$ such that $\aconfiginit$ is the initial configuration and, for all $i$, $\aconfig_i \step{} \aconfig_{i+1}$. We denote such a run $\aconfig_0 \step{*} \aconfig_k$. Similarly, we denote by $\aconfig \step{*} \aconfig'$ the existence of a sequence of steps from $\aconfig$ to $\aconfig'$.
\end{definition}

The intuition is that we will keep track of one value at a time, while assuming that we have unlimited supplies of agents in the states we covered so far. We follow the "gang" of one value through the run, which allows us to discover new states. A "gang reset" lets us add those new states to the set of covered ones and switch to another value.

First of all we observe that if there is an "abstract run" covering a state then there is a short one.

\begin{lemma}
	\label{lem:short-run}
	For every $\aconfig \in \allaconfigs$ such that $\aconfiginit \step{*} \aconfig$, there exists an "abstract run" $\arun: \aconfig_0 \step{*} \aconfig$ of less that $(|Q|+2)^3$ steps.
\end{lemma}

\ifproofs
\begin{proof}
	Note that $\covset$ may never decrease along an "abstract run" and that $\clique$ may only decrease at "gang resets".
	We can hence enforce in the abstract semantics that, at least every $|Q|+2$ steps without "reset", either $\covset$ or $\clique$ has increased. Indeed, otherwise the configuration has looped as the "boss" may only take $|Q| +1$ values. We may also enforce that $\covset$ has strictly increased between two "resets", as otherwise one may remove anything that happened between the two "resets". Therefore, there are at most $|Q|-1$ "gang resets" in total, and each portion of the run with no "reset" has at most $(|Q|+2)(|Q|+1)$ steps, yielding the bound. 
\end{proof}
\fi

It remains to prove that our abstraction is sound and complete.

\subsubsection{Completeness}
\label{one-completeness}

In this subsection we prove Lemma~\ref{lem:abstraction_complete}. To do so, we take a concrete "run" $\run$ in our model and any value $v$ appearing in the reached configuration. We prove that there exists an "abstract run" leading to the abstract configuration $(\statesin{\run}, \gangof{v}{\run})$. 

To construct the "abstract run", we will first show that for all $S$ such that $\statesin{\run} \subseteq S$, we can keep track of the set of agents carrying a value $v$: $(S, q_0, \emptyset) \step{\ast}(S, \bossof{\aval}{\run}, K)$ with $\cliqueof{\aval}{\run}\subseteq K$. 
Then, it is left to show that for all concrete "run" $\run$ and for all value $v$, $\aconfiginit \step{\ast} (S, q_0, \emptyset)$ with $\statesin{\run} \subseteq S$. 



\begin{lemma}
	\label{lem:proof_completeness_covset_constant}
	For all "initial runs" $\run: \config_0 \step{*} \config$, $S \subseteq Q$ and $\aval \in \valsof{\run}$, if $\statesin{\run} \subseteq S$ then there exists $\clique$ such that $(S, q_0, \emptyset) \step{*} (S,\bossof{\aval}{\run}, \clique)$ and $\cliqueof{\aval}{\run} \subseteq \clique$. 
\end{lemma}

\begin{proof}
	Let $\agents = \agentsof{\run}$.	
	As $\aval$ appears in $\run$, it must appear in $\config_0$; let $a_0$ be the unique agent such that $\data{\config_0}(a_0) = \aval$. We write $\run : \config_0 \step{} \config_1 \step{} \dots \step{} \config_k = \config$. For every $i \leq k$, let $\run_i : \config_0 \step{*} \config_i$ be the prefix of $\run$ of length $i$. We set $\aconfig^0 = (\covset, q_0, \emptyset)$.
	
	We construct by induction on $i$ a sequence of "abstract configurations" $\aconfig^i = (\covset, \boss_i, \clique_i)$ such that $\aconfig^0 \step{*} \aconfig^i$ and $\cliqueof{\aval}{\run_i} \subseteq \clique_i$.
	The statement is clear for $i=0$. 
	
	Suppose now that $(\covset, q_0, \emptyset) \step{*} \aconfig^i$. 
	If suffices to prove that $\aconfig^i \step{} \aconfig^{i+1}$. 
	We consider the last step of $\run_{i+1}$, which is referred to under the name $s_{i+1}$ in what follows; $s_{i+1}: \config_i \step{} \config_{i+1}$. Let $\agentbr$ the agent making the broadcast transition in $s_{i+1}$ and $A_{\recsymb}$ the set of agents receiving this broadcast in $s_{i+1}$. Let $(\statebr, \brone{\amessage}, \statebr') \in \transitions$ denote the transition taken by $\agentbr$ in $s_{i+1}$.
	
	We now make the following case distinction to determine the type of the "abstract step" $\aconfig^i \step{} \aconfig^{i+1}$:
	\begin{enumerate}
		\item\label{proof_completeness:case_broadcast_clique} if $\data{\config_{i}}(\agentbr) = \aval$ but there exists $j<i$ such that $\data{\config_{j}}(\agentbr) \ne \aval$ then it is a "broadcast from clique",
		\item\label{proof_completeness:case_broadcast_boss} if, for all $j \leq i$, $\data{\config_j}(\agentbr) = \aval$ then it is a "broadcast from boss",
		\item\label{proof_completeness:case_external_broadcast} otherwise it is an "external broadcast". 
	\end{enumerate}
	Note that $\agentbr$ may not change its register value in $s_{i+1}$ hence $\data{\config_i}(\agentbr) = \data{\config_{i+1}}(\agentbr)$. 
	
	Let $\agentboss$ the agent such that $\data{\config_0}(\agentboss) = \aval$. In case~\ref{proof_completeness:case_broadcast_boss}, $\agentboss = \agentbr$; in the other two cases, $\agentboss \ne \agentbr$. 
	
	It is easy to check that in all cases, the "gang" given by applying the abstract semantics covers $\gangof{\aval}{\run_{i+1}}$:

	\begin{enumerate}[label = (\arabic*)]
		
		\item In case~\ref{proof_completeness:case_broadcast_boss}, this condition is automatically satisfied. In the other two cases, we look at what $\agentboss$ does in $s_{i+1}$. If it remains idle then we have $\boss_i = \boss_{i+1}$. Otherwise it takes a reception transition as $\agentbr \ne \agentboss$. In case~\ref{proof_completeness:case_broadcast_clique} the condition is then satisfied.
		In case~\ref{proof_completeness:case_external_broadcast}, this reception cannot have action $\quotemarks{\eqtestact}$ as the broadcast is from an agent with register value that is not $\aval$ ($\data{\config_i}(\agentbr) \ne \aval$ by hypothesis). For the same reason, if this reception has action $\quotemarks{\enregact}$ then $\boss_{i+1}= \noboss$. If this reception has action $\quotemarks{\dummyact}$ and $\boss_i \ne \noboss$ then $\boss_{i+1} \ne \noboss$ as $\agentboss$ keeps value $\aval$.
		
		\item In all cases $K_{i+1}$ is defined by adding to $K_i$ all states reachable with value $\aval$ from it by receiving the broadcast message $m$ from a state of $K_i$ with value $\aval$ or from a state of $S$ with value $v' \neq v$ (plus $\statebr'$ in case~\ref{proof_completeness:case_broadcast_clique}). As $\cliqueof{\aval}{\run_i} \subseteq K_i$, necessarily  $\cliqueof{\aval}{\run_{i+1}} \subseteq K_{i+1}$
	\end{enumerate}
	
	We have proven that $\aconfig^i \step{} \aconfig^{i+1}$, which concludes the induction step. Applying the result with $i = k$ proves Lemma~\ref{lem:proof_completeness_covset_constant}. 
\end{proof}

We may now prove completeness of the abstraction. Intuitively, we start with $S= \set{q_0}$ and we increase it in the following way: we look at the first state $q \notin S$ that is covered in $\run$, we follow the "gang" associated with the value of an agent that covered $q$ until the "gang" covers it too (using Lemma~\ref{lem:proof_completeness_covset_constant}), then we do a gang reset to add $q$ to $S$. 

\begin{lemma}
	\label{lem:abstraction_complete}
	If $\run$ is an "initial run" and $\aval \in \valsof{\run}$, then there exist $S, \clique$ such that $\aconfiginit \step{*} (S, \bossof{\aval}{\run}, \clique)$ and $\statesin{\run} \subseteq S$ and $\cliqueof{\aval}{\run} \subseteq \clique$. 
\end{lemma}

\begin{proof}
	Let $\run$ an "initial run".
	We construct by induction an increasing sequence of sets $C_0, \ldots, C_m$ such that $C_m = \statesin{\run}$ and for all $i$, $C_i \subseteq \statesin{\run}$, $q_0 \in C_i$, and there exists $S_i$ such that $\aconfiginit \step{*} (S_i, q_0, \emptyset)$ and $C_i \subseteq S_i$. 
	First, we set $C_0 = \set{q_0}$ and the property is verified as $\aconfiginit \step{*} (\set{q_0}, q_0, \emptyset)$.
	
	Now suppose we constructed $C_0, \ldots, C_j$. If $C_j = \statesin{\run}$ we can stop. Otherwise let $\run_p: \config_0 \step{*}\config_p$ the longest suffix of $\run$ such that  $\statesin{\run_p} \subseteq C_j$. Write $s$ the step immediately after $\run_p$ in $\run$. By maximality of $\run_p$, $\run_p \cdot s$ covers some state $q$ that is not in $C_j$. Let $a$ be an agent that is in $q$ after step $s$, let $\aval$ be its value at that point. We set $C_{i+1} = C_i \cup \set{q}$.
	
	By induction hypothesis, there exist $S_{j}$ such that $\aconfiginit \step{*} (S_{j}, q_0, \emptyset)$ and $C_j \subseteq S_{j}$. Furthermore, as $\statesin{\run_p} \subseteq C_j$, by Lemma~\ref{lem:proof_completeness_covset_constant} there exists $\clique_j$ such that $\aconfiginit \step{*} (S_{j}, \bossof{\aval}{\run_p}, \clique_j)$ and $\cliqueof{\aval}{\run_p} \subseteq \clique_j$.
	
	If $q \in \clique_j$ then applying a "gang reset" suffices. If not, we mimic step $s$ with a step in the abstract semantics, as in Lemma~\ref{lem:proof_completeness_covset_constant} so that $q$ is added to the "clique", then apply a gang reset to reach $(S_{i+1}, q_0, \emptyset)$ with $C_{i+1} = S_i \cup \set{q} \subseteq S_{i+1}$. 
	
	This concludes our induction.
	
	In the end, there exist $S_m$ such that $\aconfiginit \step{*} (S_{m}, q_0, \emptyset)$ and $\statesin{\run} \subseteq S_{j}$. Lemma~\ref{lem:proof_completeness_covset_constant} allows us to conclude the proof.¨
\end{proof}

\subsubsection{Soundness}
\label{sec:one-soundness}

It is left to prove that our abstraction is sound, which we will do by considering an "abstract run" $\aconfiginit \step{\ast} \aconfig = (S, \boss, \clique)$ and constructing a concrete "run" $\run : \config_0 \step{\ast} \config $ that follows the abstract run with an exponential amount of agents to make sure that we can send some of them through all possible transitions at each step and never run out.

\begin{lemma}
	\label{cor:soundness}
	For all $\sigma_0 \in \aconfiginitset$ and $\sigma = (S, b, K) \in \Sigma$ such that $\sigma_0 \step{*} \sigma$, for all $s \in S$, there exists a reachable configuration $\gamma$ covering $s$.
\end{lemma}

We in fact prove the following stronger lemma, which directly implies Lemma~\ref{cor:soundness}.

\begin{lemma}
	\label{lem:correctness-construction}

	Let $\sigma_0 \in \aconfiginitset$, and $\sigma_0 \to \sigma_1 \to \cdots \to \sigma_n$ an "abstract run". For all $i$ let $(S_i, b_i, K_i) := \sigma_i$. Let $M = \size{\Delta}+1$.
	
	For all $i$, there exist a set of agents $\agents_i$, an "initial run" $\run_i : \config_0 \step{*} \config_i$ over $\agents_i$, agents $a_0, \cdots, a_n \in \agents_i$ and values $v_0, \ldots, v_n \in \nats$ such that:
	\begin{itemize}
		\item for all $s \in S_i$, there are at least $M^{n-i}$ agents (different from $a_i$) in state $s$ 
		
		\item for all $s \in K_i$, there are at least $M^{n-i}$ agents (different from $a_i$) in state $s$ with value $v_i$
		
		\item if $b_i \neq \noboss$, then $a_i$ is in state $b_i$ with value $v_i$.
	\end{itemize}
\end{lemma}

\begin{proof}
	
	We proceed by induction on $i$.
	We set $\agents_0 = \set{1, \ldots, M^n}$, and we set $\config_0(a) = (q_0, a)$ for all $a$. Clearly $\config_0$ satisfies the requirements with respect to $\sigma_0$, with $a_0 = v_0 \in \agents$.
	
	Now assume we constructed $\config_0 \step{*} \cdots \step{*} \config_{i}$ over $\agents_i$ satisfying the conditions of the lemma, we construct $\config_{i+1}$ using a case distinction on the form of the transition $\sigma_i \to \sigma_{i+1}$.
	For each $s \in S\setminus K$ we define $\agents_{i,s}$ as the set of agents in state $s$ in $\config_{i}$. We have $\size{\agents_{i,s}} \geq M^{n-i}$ thus we can extract $M = \size{\Delta}+1$ disjoint sets of agents $(\agents_{i,s}^d)_{d \in \Delta\cup\set{\epsilon}}$ from it, each set having $M^{n-i-1}$ agents.
	Similarly, for each $s \in K$ we define $\agents_{i,s}$ as the set of agents in state $s$ \textbf{with value $\mathbf{v_i}$} in $\config_{i}$. We have $\size{\agents_{i,s}} \geq M^{n-i}$ thus we can extract $\size{\Delta}+1$ disjoint sets of agents $(\agents_{i,s}^d)_{d \in \Delta\cup\set{\epsilon}}$ from it each set having $M^{n-i-1}$ agents.
	\\
	
	\textbf{Case 1: } If $\sigma_i \to \sigma_{i+1}$ is a \kl{broadcast from clique} $d = (q, \brone{m}, q')$ with $q \in K_i$, then we make all agents $a \in \agents_{i,q}^{d}$ (which all have value $v_i$) execute that transition one by one.
	None of those broadcasts are received by any other agent, except for the last one:
	If $b \neq b'$ then there is a transition $(b, \recone{m}{\alpha}, b')$ and we make $a_i$ execute it upon receiving the broadcast. We then set $a_{i+1} = a_i$.
	For all $k' \in K_{i+1} \setminus (K_i \cup \set{q'})$ there exists a transition $d'=(k, \recone{m}{\alpha}, k')$ such that either $\alpha$ is $\eqtestact$ or $*$ and $k \in K_i$ or $\alpha$ is $\enregact$ and $k\in S$.
	In both cases we make all agents of $\agents_{i,k}^{d'}$ take that transition.
	We set $v_{i+1} = v_i$.
	\\
	
	\textbf{Case 2: }If $\sigma_i \to \sigma_{i+1}$ is a \kl{broadcast from boss} $d = (b_i, \brone{m}, b_{i+1})$, then we make $a_i$ (which has value $v_i$) execute that transition, and we set $a_{i+1} = a_i$.
	The agents receiving that message are as follows:
	
	For all $k' \in K_{i+1} \setminus K_i $ there exists a transition $d'=(k, \recone{m}{\alpha}, k')$ such that $\alpha$ is either $\eqtestact$ or $*$ and $k \in K_i$ or $\alpha$ is $\enregact$ and $k\in S$.
	In both cases we make all agents of $\agents_k^{d'}$ take that transition.
	By definition of an "abstract run", we must have $b_i \in S_i$.
	Hence we can make all agents of $\agents_{i,s}^{d}$ execute $d$, with no agent receiving the corresponding broadcasts.
	We set $v_{i+1} = v_i$.
	\\
	
	\textbf{Case 3: } If $\sigma_i \to \sigma_{i+1}$ is an \kl{external broadcast} $d = (q, \brone{m}, q')$ , then we make all agents $a \in \agents_q^{d}$ execute that transition one by one. None of those broadcasts are received by any other agent, except for the last one:
	If $b_i \neq b_{i+1}$ then there is a transition $(b_i, \recone{m}{\alpha}, b')$ and either $b_{i+1} = b'' \neq \noboss$ and $\alpha = *$ or $b_{i+1} = \noboss$ and $\alpha=\enregact$. In both cases we make $a_i$ execute that transition, and we set $a_{i+1} = a_i$.
	For all $k' \in K_{i+1} \setminus K_i$ there exists a transition $d'=(k, \recone{m}{*}, k')$ with $k \in K_i$. We make all agents of $\agents_k^{d'}$ take that transition.
	We set $v_{i+1} = v_i$.
	\\
	
	\textbf{Case 4: }  If $\sigma_i \to \sigma_{i+1}$ is a \kl{gang reset} then no agent moves and we select some $a_{i+1}$ in $\agents_{q_0}$ and set $v_{i+1}$ to be its value.
	\\
	
	Throughout the case distinction we have ensured that:
	\begin{itemize}
		\item If $b_{i+1} \neq \noboss$ then $a_{i+1}$ is an agent of value $v_{i+1}$.
		
		\item If the step is not a gang reset, then $v_{i+1} = v_i$ and for all $k' \in K_{i+1} \setminus K_i$, there exists $d \in \Delta$ from some $k$ to $k'$ such that all agents of $\agents_{i,k}^d$ take that transition. Furthermore, if $d$ is of the form $(k,\recone{m}{\enregact},k')$ then the broadcasting agent has value $v_i$, thus all those agents keep value $v_i = v_{i+1}$. For all $k \in K_{i}$, the agents of $\agents_{i,k}^\epsilon$ do not move between configurations $\config_{i}$ and $\config_{i+1}$, hence they have state $k$ and value $v_{i+1}$ in $\config_{i+1}$.
		
		\item If the step is a gang reset, the conditions of the lemma hold trivially.
	\end{itemize}
	
	As a result, we have ensured that the conditions of the lemma were respected.
	This concludes our induction.
\end{proof}

To obtain Lemma~\ref{cor:soundness}, we apply Lemma~\ref{lem:correctness-construction} to an "abstract run" $\sigma_0 \to \cdots \sigma_n = \sigma$ from $\sigma_0$ to $\sigma$ by setting $i = n$.

\subsection{Conclusion}

\begin{proposition}
	\label{prop:sound-and-complete}
	Let $q_f$ be a state, there exists a reachable "configuration" covering $q_f$ if and only if there exists a reachable "abstract configuration" $(S,b,K)$ with $q_f \in S$.
\end{proposition}

\begin{proof}
	The two implications follow from Lemmas~\ref{lem:abstraction_complete} and~\ref{cor:soundness} respectively.
\end{proof}

\textbf{NP-hardness.} We present here a reduction from the 3SAT problem to the "cover problem" in 1-BNRAs.

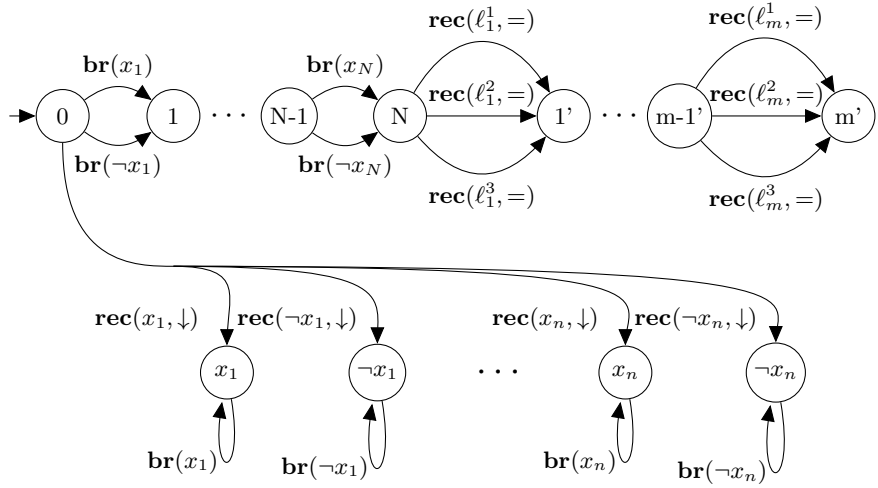
\begin{figure}[h]
	\input{Figures/fig-np-hard}
	\caption{The "protocol" used for the NP-hardness proof.}
	\label{fig:np-hard}
\end{figure}

\begin{proposition}
	\label{prop:np-hard-query-cover}
	The "cover problem" is NP-hard.
\end{proposition}

\begin{proof}
	Let $x_1, \ldots, x_n$ be variables and $\query = \bigwedge_{j=1}^m C_j$ with, for all $j$, $C_j = \ell_j^1 \lor \ell_j^2 \lor \ell_j^3$ and $\ell_j^1, \ell_j^2, \ell_j^3 \in \set{x_i, \neg x_i \mid 1 \leq i \leq n}$. 
	
	Consider the "protocol" displayed in Figure~\ref{fig:np-hard}.
	Our alphabet of messages is the set of literals $\set{x_i, \neg x_i \mid 1 \leq i \leq n}$.
	Agents may either receive a message, and repeat it forever or it may broadcast one of $x_i, \neg x_i$ for each $i$ and then try to receive a message with one of $\ell_j^1, \ell_j^2, \ell_j^3$ for each $j$, with their own register value.
	
	If $\query$ is satisfied by some assignment $\nu$, then we construct a run where an agent $a$ broadcasts the satisfied literals while going from $0$ to $N$ and other agents receive the messages and go to the corresponding states in the lower part of the protocol while storing the register value of $a$.
	Then for each $j$ we select some $\ell_j^p$ satisfied by $\nu$. There exists $i$ such that there is an agent in state $\ell_j^p$, which broadcasts $\ell_j^p$ along with the initial register value of $a$, allowing $a$ to go to the next state.
	As a result, there is a "run" in which an agent $a$ reaches $m'$.
	
	Now suppose there is a "run" $\run$ over some set of agents $\agents$ such that some agent $a \in \agents$ is in state $m'$ in the final configuration.
	For each $i$, $a$ has broadcast either $x_i$ or $\neg x_i$, but not both.
	Let $\nu$ be the valuation assigning $\top$ to $x_i$ if and only if $a$ has broadcast it.
	For each $j$ $a$ has received one of $\ell_j^1, \ell_j^2, \ell_j^3$ along with its own initial register value (which we call $r$). For this to happen, $a$ must have broadcast this literal before, hence it is satisfied by $\nu$.
	
	As a result, $\nu$ satisfies a literal of each clause of $\query$, and thus satisfies $\query$. This concludes our reduction.
\end{proof}


\thmNPComplete*

\begin{proof}
	The lower bound is given by Proposition~\ref{prop:np-hard-query-cover}.
	For the upper bound, say we are given a "protocol" $\prot = (Q, \messages, \Delta, q_0, \regnum)$ and a state $q_f$. By Proposition~\ref{prop:sound-and-complete}, there is a reachable configuration covering $q_f$ if and only if there is an "abstract run" to an "abstract configuration" $(S,b, K)$ with $q_f \in S$.
	Furthermore, by Lemma~\ref{lem:short-run} if there is such an "abstract run" then there is one with at most $(\size{Q}+2)^3$ steps. 
	Thus we can simply guess such an "abstract run" and verify it in polynomial time.
	As a result, the "cover problem" is in \NP. 
\end{proof}

%% file: Figures/fig-np-hard.tex
\begin{tikzpicture}[xscale=0.5,AUT style,node distance=2cm,auto,>= triangle
	45]
	\tikzstyle{initial}= [initial by arrow,initial text=,initial
	distance=.7cm]
	
	\node[state,initial, minimum width=0.1pt] (0) at (0,0) {0};
	
	\node[state] [right of=0, xshift=-15pt] (1) {1};
	
	\node [right of=1, xshift=-35pt] (dots) {\large $\cdots$};

	\node[state] [right of=dots, xshift=-35pt] (N-1) {N-1};
	\node[state] [right of=N-1, xshift=-15pt] (N) {N};
	
	\node[state] [right of=N, xshift=5pt] (1') {1'};
	
	\node [right of=1', xshift=-35pt] (dots2) {\large $\cdots$};

	\node[state] [right of=dots2, xshift=-35pt] (m-1) {m-1'};
	\node[state] [right of=m-1, xshift=7pt] (m) {m'};
	
	\coordinate[below of=1] (stop);
	\node[state] [below right of=stop, xshift=-20pt] (r1) {$x_1$};
	\node [above left of=r1, yshift=-20pt, xshift=10pt] (r1') {$\recone{x_1}{\enregact}$};
	\node[state] [right of=r1] (r1b) {$\neg x_1$};
	\node [above left of=r1b, yshift=-20pt, xshift=10pt] (r1b') {$\recone{\neg x_1}{\enregact}$};
	\node [right of= r1b, xshift=-10pt] (dots3) {\Large $\cdots$};
	\node[state] [right of=dots3, xshift=-10pt] (rn) {$x_n$};
	\node [above left of=rn, yshift=-20pt, xshift=10pt] (rn') {$\recone{x_n}{\enregact}$};
	\node[state] [right of=rn] (rnb) {$\neg x_n$};
	\node [above left of=rnb, yshift=-20pt, xshift=10pt] (rnb') {$\recone{\neg x_n}{\enregact}$};
	
	\draw (0) .. controls +(0,-2) and +(-1,0) .. (stop);
	\draw[->] (stop) .. controls +(2,0) and +(0,1) .. (r1);
	\draw[->] (stop) .. controls +(6,0) and +(0,1) .. (r1b);
	\draw[->] (stop) .. controls +(14,0) and +(0,1) .. (rn);
	\draw[->] (stop) .. controls +(18,0) and +(0,1) .. (rnb);
	\path[->, bend left=20]
	(0) edge node[above] {$\brone{x_1}$} (1)
	(N-1) edge node[above] {$\brone{x_N}$} (N) 	
	;
	\path[->, bend left=40]
	(N) edge node[above] {$\recone{\ell_1^1}{\eqtestact}$} (1')
	(m-1) edge node[above] {$\recone{\ell_m^1}{\eqtestact}$} (m) 	
	;
	\path[->, bend right=20] 
	(0) edge node[below] {$\brone{\neg x_1}$} (1)
	(N-1) edge node[below] {$\brone{\neg x_N}$} (N) 
	;
	\path[->, bend right=30] 
	(N) edge node[below] {$\recone{\ell_1^3}{\eqtestact}$} (1')
	(m-1) edge node[below] {$\recone{\ell_m^3}{\eqtestact}$} (m)	
	;
	\path[->]
	(N) edge node[above] {$\recone{\ell_1^2}{\eqtestact}$} (1')
	(m-1) edge node[above] {$\recone{\ell_m^2}{\eqtestact}$} (m) 	
	;
	\path[->, loop below]
	(r1) edge node[left] {$\brone{x_1}$} (r1)
	(r1b) edge node[left] {$\brone{\neg x_1}$} (r1b) 	
	(rn) edge node[left] {$\brone{x_n}$} (rn)
	(rnb) edge node[left] {$\brone{\neg x_n}$} (rnb) 	
	;
\end{tikzpicture}

%% file: m.bbl
\begin{thebibliography}{10}
\providecommand{\url}[1]{\texttt{#1}}
\providecommand{\urlprefix}{URL }
\providecommand{\doi}[1]{https://doi.org/#1}

\bibitem{AbdullaAKR14}
Abdulla, P.A., Atig, M.F., Kara, A., Rezine, O.: Verification of dynamic
  register automata. In: 34th International Conference on Foundation of
  Software Technology and Theoretical Computer Science, {FSTTCS} 2014. LIPIcs,
  vol.~29, pp. 653--665. Schloss Dagstuhl - Leibniz-Zentrum f{\"{u}}r
  Informatik (2014). \doi{10.4230/LIPIcs.FSTTCS.2014.653}

\bibitem{AbdullaAKR15}
Abdulla, P.A., Atig, M.F., Kara, A., Rezine, O.: Verification of buffered
  dynamic register automata. In: Networked Systems, {NETYS} 2015. Lecture Notes
  in Computer Science, vol.~9466, pp. 15--31. Springer (2015).
  \doi{10.1007/978-3-319-26850-7\_2}

\bibitem{AbdullaJ1996verif}
Abdulla, P.A., Jonsson, B.: Verifying programs with unreliable channels.
  Information and Computation  \textbf{127}(2),  91--101 (1996).
  \doi{10.1006/inco.1996.0053}

\bibitem{Balasubramanian18}
Balasubramanian, A.R., Bertrand, N., Markey, N.: Parameterized verification of
  synchronization in constrained reconfigurable broadcast networks. In: Tools
  and Algorithms for the Construction and Analysis of Systems, {TACAS} 2018.
  Lecture Notes in Computer Science, vol. 10806, pp. 38--54. Springer (2018).
  \doi{10.1007/978-3-319-89963-3\_3}

\bibitem{BalasubramanianGW22}
Balasubramanian, A.R., Guillou, L., Weil{-}Kennedy, C.: Parameterized analysis
  of reconfigurable broadcast networks. In: Foundations of Software Science and
  Computation Structures, {FoSSaCS} 2022. Lecture Notes in Computer Science,
  vol. 13242, pp. 61--80. Springer (2022). \doi{10.1007/978-3-030-99253-8\_4}

\bibitem{BolligRS21}
Bollig, B., Ryabinin, F., Sangnier, A.: Reachability in distributed memory
  automata. In: Annual Conference on Computer Science Logic, {CSL} 2021.
  LIPIcs, vol.~183, pp. 13:1--13:16. Schloss Dagstuhl - Leibniz-Zentrum
  f{\"{u}}r Informatik (2021). \doi{10.4230/LIPIcs.CSL.2021.13}

\bibitem{BZ83}
Brand, D., Zafiropulo, P.: On communicating finite-state machines. Journal of
  the ACM  \textbf{30}(2),  323–342 (1983). \doi{10.1145/322374.322380}

\bibitem{ChambartS08ordinal}
Chambart, P., Schnoebelen, P.: The ordinal recursive complexity of lossy
  channel systems. In: Annual {IEEE} Symposium on Logic in Computer Science,
  {LICS} 2008. pp. 205--216. {IEEE} Computer Society (2008).
  \doi{10.1109/LICS.2008.47}

\bibitem{DBLP:journals/computing/ChiniMS22}
Chini, P., Meyer, R., Saivasan, P.: Liveness in broadcast networks. Computing
  \textbf{104}(10),  2203--2223 (2022). \doi{10.1007/s00607-021-00986-y}

\bibitem{DelzannoST13}
Delzanno, G., Sangnier, A., Traverso, R.: Parameterized verification of
  broadcast networks of register automata. In: Reachability Problems , {RP}
  2013. Lecture Notes in Computer Science, vol.~8169, pp. 109--121. Springer
  (2013). \doi{10.1007/978-3-642-41036-9\_11}

\bibitem{DelzannoSTZ12}
Delzanno, G., Sangnier, A., Traverso, R., Zavattaro, G.: On the complexity of
  parameterized reachability in reconfigurable broadcast networks. In: {IARCS}
  Annual Conference on Foundations of Software Technology and Theoretical
  Computer Science, {FSTTCS} 2012. LIPIcs, vol.~18, pp. 289--300. Schloss
  Dagstuhl - Leibniz-Zentrum f{\"{u}}r Informatik (2012).
  \doi{10.4230/LIPIcs.FSTTCS.2012.289}

\bibitem{DelzannoSZ2010Adhoc}
Delzanno, G., Sangnier, A., Zavattaro, G.: Parameterized verification of ad hoc
  networks. In: {CONCUR} 2010. Lecture Notes in Computer Science, vol.~6269,
  pp. 313--327. Springer (2010). \doi{10.1007/978-3-642-15375-4\_22}

\bibitem{emerson1998model}
Emerson, E.A., Namjoshi, K.S.: On model checking for non-deterministic
  infinite-state systems. In: Annual {IEEE} Symposium on Logic in Computer
  Science, {LICS} 1998. pp. 70--80. {IEEE} Computer Society (1998).
  \doi{10.1109/LICS.1998.705644}

\bibitem{EsparzaFM1999verification}
Esparza, J., Finkel, A., Mayr, R.: On the verification of broadcast protocols.
  In: 14th Annual {IEEE} Symposium on Logic in Computer Science, Trento, Italy,
  July 2-5, 1999. pp. 352--359. {IEEE} Computer Society (1999).
  \doi{10.1109/LICS.1999.782630}

\bibitem{datanetsinequalityfomegaomegaomega}
Haddad, S., Schmitz, S., Schnoebelen, P.: The ordinal-recursive complexity of
  timed-arc petri nets, data nets, and other enriched nets. In: Proceedings of
  the 27th Annual {IEEE} Symposium on Logic in Computer Science, {LICS} 2012,
  Dubrovnik, Croatia, June 25-28, 2012. pp. 355--364. {IEEE} Computer Society
  (2012). \doi{10.1109/LICS.2012.46}

\bibitem{Higman52}
Higman, G.: Ordering by divisibility in abstract algebras. Proceedings of the
  London Mathematical Society  \textbf{s3-2}(1),  326--336 (1952).
  \doi{10.1112/plms/s3-2.1.326}

\bibitem{Lasota16}
Lasota, S.: Decidability border for petri nets with data: {WQO} dichotomy
  conjecture. In: Kordon, F., Moldt, D. (eds.) Application and Theory of Petri
  Nets and Concurrency - 37th International Conference, {PETRI} {NETS} 2016,
  Toru{\'{n}}, Poland, June 19-24, 2016. Proceedings. Lecture Notes in Computer
  Science, vol.~9698, pp. 20--36. Springer (2016).
  \doi{10.1007/978-3-319-39086-4\_3},
  \url{https://doi.org/10.1007/978-3-319-39086-4\_3}

\bibitem{LazicNORW08}
Lazic, R., Newcomb, T.C., Ouaknine, J., Roscoe, A.W., Worrell, J.: Nets with
  tokens which carry data. Fundam. Informaticae  \textbf{88}(3),  251--274
  (2008). \doi{10.1007/978-3-540-73094-1_19}

\bibitem{minsky}
Minsky, M.L.: Computation: Finite and Infinite Machines. Prentice-Hall, Inc.,
  USA (1967)

\bibitem{Rezine17}
Rezine, O.: Verification of networks of communicating processes: Reachability
  problems and decidability issues. Ph.D. thesis, Uppsala University, Sweden
  (2017)

\bibitem{Rosa-Velardo17}
Rosa{-}Velardo, F.: Ordinal recursive complexity of unordered data nets.
  Information and Computation  \textbf{254},  41--58 (2017).
  \doi{10.1016/j.ic.2017.02.002}

\bibitem{ArnaudErratum}
Sangnier, A.: Erratum to parameterized verification of broadcast networks of
  register automata (2023),
  \url{https://www.irif.fr/~sangnier/publications.html}

\bibitem{Schmitz16}
Schmitz, S.: Complexity hierarchies beyond elementary. {ACM} Transactions on
  Computation Theory  \textbf{8}(1),  3:1--3:36 (2016). \doi{10.1145/2858784}

\bibitem{SchmitzS2011upperHigman}
Schmitz, S., Schnoebelen, P.: Multiply-recursive upper bounds with {Higman's}
  lemma. In: International Colloquium on Automata, Languages and Programming,
  {ICALP} 2011. Lecture Notes in Computer Science, vol.~6756, pp. 441--452.
  Springer (2011). \doi{10.1007/978-3-642-22012-8\_35}

\bibitem{WSTS}
Schmitz, S., Schnoebelen, P.: The power of well-structured systems. In:
  D'Argenio, P.R., Melgratti, H.C. (eds.) {CONCUR} 2013 - Concurrency Theory -
  24th International Conference, {CONCUR} 2013, Buenos Aires, Argentina, August
  27-30, 2013. Proceedings. Lecture Notes in Computer Science, vol.~8052, pp.
  5--24. Springer (2013). \doi{10.1007/978-3-642-40184-8\_2}

\bibitem{Schnoebelen2002verifying}
Schnoebelen, P.: Verifying lossy channel systems has nonprimitive recursive
  complexity. Information Processing Letters  \textbf{83}(5),  251--261 (2002).
  \doi{10.1016/S0020-0190(01)00337-4}

\end{thebibliography}
